\documentclass[12pt]{article}

\usepackage[margin=1in]{geometry}

\usepackage{amsmath}
\usepackage{amssymb}
\usepackage{amsfonts}
\usepackage{amsthm}
\usepackage{bm}
\usepackage{mathrsfs}
\usepackage{dsfont}
\let\mathbbm\mathds

\usepackage{algorithm}
\usepackage{algpseudocode}

\usepackage{graphicx}

\usepackage{xcolor}

\usepackage{natbib}

\usepackage[
    colorlinks=true,
    linkcolor=red,
    citecolor=blue,
    urlcolor=blue
]{hyperref}

\newtheorem{theorem}{Theorem}
\newtheorem{proposition}[theorem]{Proposition}
\newtheorem{corollary}{Corollary}
\newtheorem{lemma}{Lemma}

\theoremstyle{definition}

\theoremstyle{remark}
\newtheorem{remark}{Remark}

\def \simiid {\overset{\text{i.i.d.}}{\sim}}	
\def \simind {\overset{\text{ind.}}{\sim}}

\def \d {{\rm d}}

\def \bv {v}

\def \A {\mathcal{A}}
\def \B {\mathcal{B}}
\def \C {\mathcal{C}}
\def \sX {\mathcal{X}}

\def \E {\mathbb{E}}

\def \P {\mathbb{P}}
\def \R {\mathbb{R}}
\def \N {\mathbb{N}}

\def \sX {\mathcal{X}}
\def \sY {\mathcal{Y}}
\def \sO {\mathcal{O}}
\def \sP {\mathcal{P}}

\def \sK {\mathcal{K}}

\newcommand{\PRP}{P_\mathrm{RP}}
\newcommand{\PCD}{P_\mathrm{CD}}
\newcommand{\PMG}{P_\mathrm{MG}}
\newcommand{\tPMG}{\tilde{P}_\mathrm{MG}}
\newcommand{\PMGi}{P_{\mathrm{MG}, i}}
\newcommand{\PNR}{P_\mathrm{NR}}
\newcommand{\QNR}{G_\mathrm{NR}}
\newcommand{\Prev}{P_\mathrm{R}}

\newcommand{\Var}{\text{Var}}
\newcommand{\Cov}{\text{Cov}}

\title{\textbf{A fast non-reversible sampler for Bayesian mixture models}}

\author{
    Filippo Ascolani\thanks{Department of Statistical Science, Duke University, Durham, NC, USA.}
    \and
    Paolo Manildo\thanks{Department of Statistical Sciences, University of Padua, Padua, Italy}
    \and
    Giacomo Zanella\thanks{Department of Decision Sciences and BIDSA, Bocconi University, Milan, Italy
    \\ GZ acknowledges support from the European Research Council (ERC), through StG ``PrSc-HDBayLe'' grant ID 101076564. FA and GZ thank Lorenzo Calda for helpful preliminary simulations and discussions.
    }
}

\begin{document}

\maketitle

\abstract{Mixtures models are a cornerstone of Bayesian modelling, and it is well-known that sampling from the resulting posterior distribution can be a hard task.  
In particular, popular reversible Markov chain Monte Carlo schemes  are often slow to converge  
when the number of observations $n$ is large.  
In this paper we introduce a novel and simple non-reversible sampling scheme for Bayesian mixture models (with fixed or varying number of components), which is shown to drastically outperform classical samplers in many scenarios of interest, especially during convergence phase 
and when components in the mixture have non-negligible overlap.
At the theoretical level, we show that the performance of the proposed non-reversible scheme cannot be worse than the standard one, in terms of asymptotic variance, by more than a factor of four; and we provide a scaling limit analysis suggesting that the non-reversible sampler can reduce the convergence time from $\mathcal{O}(n^2)$ to $\mathcal{O}(n)$. 
We also discuss why the statistical features of mixture models make them an ideal case for the use of non-reversible discrete samplers. }

\vspace{0.3cm}
\noindent
\textbf{Key words:} MCMC, clustering, lifted Markov chains, scaling limits, Wright-Fisher diffusion, asymptotic variance
\vspace{0.3cm}

\maketitle

\vspace{-2em}

\section{Introduction}

\subsection{Bayesian finite mixture models}

Let $K \in \N$ and consider a finite mixture model \citep{marin2005bayesian, fruhwirth2006finite, mclachlan2019finite} defined as
\begin{equation}\label{eq:original_model}
\begin{aligned}
        Y_i \mid \bm{\theta}, \bm{w} &\simiid \sum_{k = 1}^Kw_kf_{\theta_k}(\cdot)\qquad\qquad&i=1,\dots,n
        \\  
\theta_k &\simiid p_0&k=1,\dots,K\\
\bm{w}&\sim \text{Dir}(\bm{\alpha}),
\end{aligned}
\end{equation}
where $\bm{\theta} = (\theta_1, \dots, \theta_K)$, $\bm{w} = (w_1, \dots, w_K)$, $\bm{\alpha} = (\alpha_1, \dots, \alpha_K)$ and $\text{Dir}(\bm{\alpha})$ denotes the Dirichlet distribution on the $(K-1)$-dimensional simplex $\Delta_{K-1}\subset \R^K$ with parameters $\bm{\alpha}$.
Here $f_\theta$ is a probability density on a space $\sY$ depending on a parameter $\theta \in \Theta \subset \R^d$, to which a prior distribution with density $p_0$ is assigned. 
For example, one could have $\sY=\Theta=\R^d$ and $f_\theta(y) = N(y \mid \theta, \Sigma)$ for some fixed $\Sigma$, where $N(y \mid \theta, \Sigma)$ denotes the density of the multivariate normal with mean vector $\theta$ and covariance matrix $\Sigma$ at a point $y$.

A popular strategy to perform posterior computations with model \eqref{eq:original_model} (also for maximum likelihood estimation, as originally noted in \cite{dempster1977maximum}) consists in augmenting the model  
as follows
\begin{equation}\label{eq:augmented_model}
    Y_i \mid c, \bm{\theta}, \bm{w} \simind f_{\theta_{c_i}}(y), \quad c_i \mid \bm{\theta}, \bm{w} \simiid \text{Cat}(\bm{w}), \quad \bm{w} \sim \text{Dir}(\bm{\alpha}), \quad \theta_k \simiid p_0,
\end{equation}
where $c= (c_1, \dots, c_n) \in [K]^n$, with $[K] = \{1, \dots, K\}$, is the set of allocation variables and Cat$(\bm{w})$ denotes the Categorical distribution with probability weights $\bm{w}$. 
Given a sample $Y = (Y_1, \dots, Y_n)$, the joint posterior distribution of $(c, \bm{\theta}, \bm{w})$ then reads
\begin{equation}\label{eq:augmented_posterior}
\begin{aligned}
\pi(c, \bm{\theta}, \bm{w}) \, &\propto \, \left[ \prod_{i = 1}^nw_{c_i}f_{\theta_{c_i}}(Y_i)\right]\text{Dir}(\bm{w} \mid \bm{\alpha})\prod_{k = 1}^Kp_0(\theta_k)\\
& \propto \, \left[\prod_{k = 1}^Kw_k^{n_k(c) + \alpha_k - 1}\right]\prod_{k = 1}^K\prod_{i \, : \, c_i = k}f_{\theta_{k}}(Y_i)p_0(\theta_k),
\end{aligned}
\end{equation}
where $n_k(c) = \sum_{i = 1}^n \mathbbm{1}(c_i = k)$ and $\mathbbm{1}$ denotes the indicator function. 
In particular, it is possible to integrate out $(\bm{\theta}, \bm{w})$ from \eqref{eq:augmented_posterior} to obtain the marginal posterior distribution of $c$ given by
\begin{equation}\label{eq:marg_c}
\pi(c) \, \propto \, \left[\prod_{k = 1}^K\Gamma\left(\alpha_k + n_k(c)\right)\right]\prod_{k = 1}^K\int_{\Theta}\prod_{i \, : \, c_i = k}f_{\theta_{k}}(Y_i)p_0(\theta_k)\, \d \theta_k,
\end{equation}
from which we deduce the so-called full-conditional distribution of $c_i$
\begin{align}\label{eq:predictive}
\pi(c_i = k \mid c_{-i}) \, &\propto \, \left[\alpha_k + n_k(c_{-i})\right] p(Y_i \mid Y_{-i}, c_{-i}, c_i = k)
&k\in[K]
\end{align} 
where $c_{-i} = (c_1, \dots, c_{i-1}, c_{i+1}, \dots, c_n)$, $Y_{-i} = (Y_1, \dots, Y_{i-1}, Y_{i+1}, \dots, Y_n)$  and
\[
p(Y_i \mid Y_{-i}, c_{-i}, c_i = k) = \int_\Theta f_{\theta}(Y_i)\frac{\prod_{j \neq i \, : \, c_j = k}f_{\theta}(Y_j)p_0(\theta)}{\int_\Theta \prod_{j \neq i \, : \, c_j = k}f_{\theta'}(Y_j)p_0(\theta') \, \d \theta'} \, \d \theta
\]
is the predictive density of observation $Y_i$ given $Y_{-i}$ and the allocation variables.
If $p_0$ is conjugate with respect to the density $f_\theta$, then $p(Y_i \mid Y_{-i}, c_{-i}, c_i = k)$ and thus $\pi(c_i = k \mid c_{-i})$ are available in closed form. For example, if $f_\theta(y) = N(y \mid \theta, \Sigma)$ and $p_0(\theta) = N(\theta \mid \theta_0, \Sigma_0)$ it holds that $p(Y_i \mid Y_{-i}, c_{-i}, c_i = k) = N(Y_i \mid \bar{\mu}, \bar{\Sigma})$, where
\[
\bar{\Sigma} = \Sigma + \left(\Sigma_0^{-1}+n_k(c_{-i})\Sigma^{-1} \right)^{-1}, \quad \bar{\mu} = \left(\Sigma_0^{-1}+n_k(c_{-i})\Sigma^{-1} \right)^{-1}\left(\Sigma_0^{-1}\theta_0+n_k(c_{-i})\Sigma^{-1}\bar{Y}_{k, -i} \right)
\]
and $\bar{Y}_{k, -i} = n^{-1}_k(c_{-i})\sum_{j \neq i \, : \, c_j = k}Y_j$. Analogous expressions are available for likelihoods in the exponential family, see e.g.\ \citet[Sec.3.3]{robert2007bayesian} for details.

\subsection{The Marginal Gibbs (MG) sampler}\label{sec:marginal}
Most popular algorithms for finite mixture models are based on the augmentation in \eqref{eq:augmented_model}, see e.g.\ \cite{diebolt1994estimation}. Here we consider the random-scan\footnote{Here we consider the random-scan strategy since it simplifies some of the proofs and comparisons below. We expect the behaviour of $\PMG^n$, where $P^k = P\dots P$ denotes the $k$-th power of a Markov kernel $P$, to be roughly comparable to the one of the deterministic-scan version (which updates $c_i$ for $i=1,\dots,n$ sequentially at each iteration) in  most cases of interest, and thus stick to the random-scan for simplicity.} Gibbs sampler which iterates updates from $\pi(c_i \mid c_{-i})$ for randomly chosen $i\in [n]$. Its Markov kernel $\PMG$ is defined as
\begin{equation} \label{eq:marginal}
\begin{aligned}
    \PMG(c, c') &= \frac{1}{n}\sum_{i = 1}^n\PMGi(c, c'), 
    &c,c'\in[K^n]
\end{aligned}
\end{equation}
with $\PMGi(c, c') = \delta_{c_{-i}}( c_{-i}')\pi(c_i \mid c_{-i})$.
The associated pseudocode is given in Algorithm \ref{alg:PMG}.
\begin{algorithm}[htbp]
\begin{algorithmic}
\State Initialize $c^{(0)}\in [K]^n$. 
\For{$t \geq 1$}
    \State Sample $i\sim \text{Unif}\left(\{1, \dots, n\}\right)$, where $\text{Unif}$ denotes the uniform distribution.
    \State Sample $c^{(t)}_i \sim \pi(c_i \mid c_{-i}^{(t-1)})$, with $\pi(c_i = k \mid c_{-i})$ as in \eqref{eq:predictive}, 
    and set $c^{(t)}_{-i}=c_{-i}^{(t-1)}$.
   \EndFor
\end{algorithmic}
\caption{(Marginal sampler $\PMG$)
\label{alg:PMG}}
\end{algorithm}
We refer to $\PMG$ as \emph{marginal} sampler, since it targets the marginal posterior distribution of $c$ defined in \eqref{eq:marg_c}. 
Once a sample from $\pi(c)$ is available, draws from $\pi(c,\bm{\theta}, \bm{w})$ can be obtained by sampling from $\pi(\bm{\theta}, \bm{w} \mid c)$, so that Algorithm \ref{alg:PMG} can be used to perform full posterior inference on $\pi(c, \bm{\theta}, \bm{w})$. 

Being an irreducible and aperiodic Markov kernel on a finite space, $\PMG$ is uniformly ergodic for every fixed $n$, see e.g.\ \citet[Theorem 4.9]{levin2017markov} and \citet[Sec.3.3]{roberts2004general} for discussion about uniform ergodicity. 
However, as we will see later on, its rate of convergence tends to deteriorate quickly as $n$ increases. A popular alternative to the marginal sampler is the so-called \emph{conditional} sampler introduced in \cite{diebolt1994estimation}, which directly targets $\pi(c, \bm{\theta}, \bm{w})$ defined in \eqref{eq:augmented_posterior} alternating updates of $(\bm{\theta}, \bm{w}) \mid c$ and $c \mid (\bm{\theta}, \bm{w})$. 
We postpone the discussion of this algorithm to Section \ref{sec:conditional} of the supplement, since the latter is always dominated by $\PMG$ in terms of mixing speed (see e.g.\ Proposition \ref{prop:asymp_variances_cond}).

\subsection{Illustrative example}
It is well-known that $\PMG$ can be slow to converge when $n$ is large \citep{celeux2000computational, lee2009bayesian}. As a first illustrative example, we take 
model \eqref{eq:original_model} with $K = 2$, $f_\theta(y) = N(y \mid \theta, 1)$, $p_0(\theta) = N(\theta \mid 0, 1)$, $\bm{\alpha} = (0.5, 0.5)$, and we consider the posterior distribution given $(Y_1,\dots,Y_n)$ generated as
\begin{align}\label{eq:truth_example}
Y_i &\simiid 0.9N(0.9, 1)+0.1N(-0.9, 1), 
& i = 1, \dots, n
\end{align}
with $n=2000$.  
This is a relatively simple one-dimensional problem, with data generated from two components with a reasonable degree of separation between them (the two means are almost two standard deviations away from each other).
\begin{figure}[h]
\centering
\includegraphics[width=.31\textwidth]{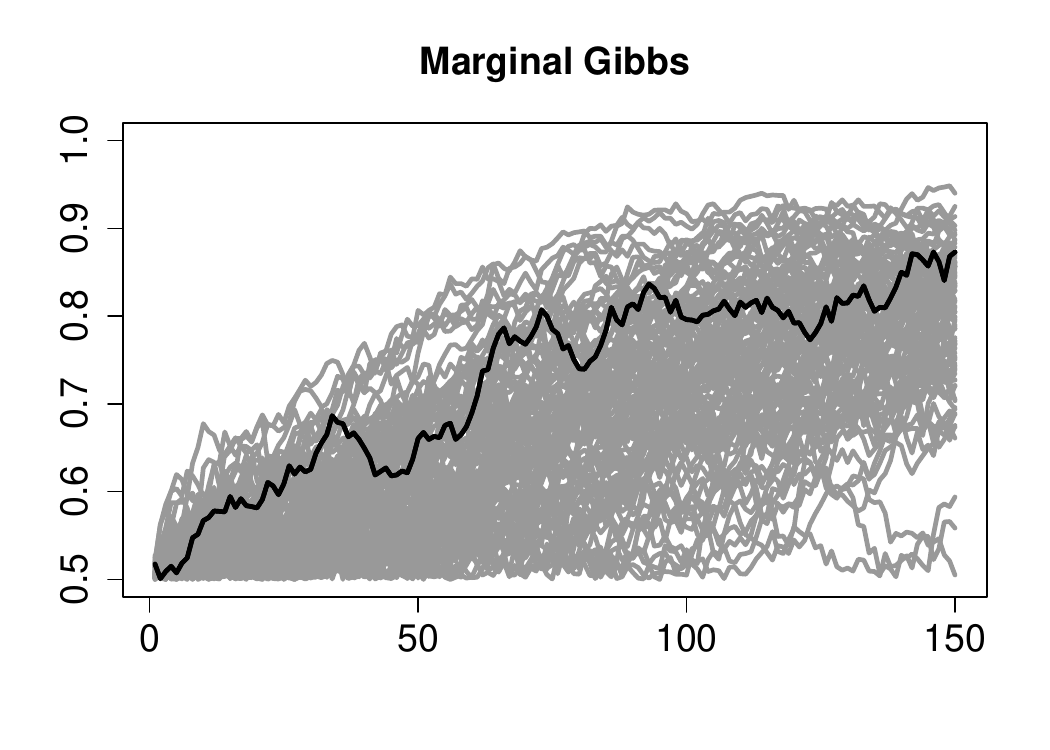} \,
\includegraphics[width=.31\textwidth]{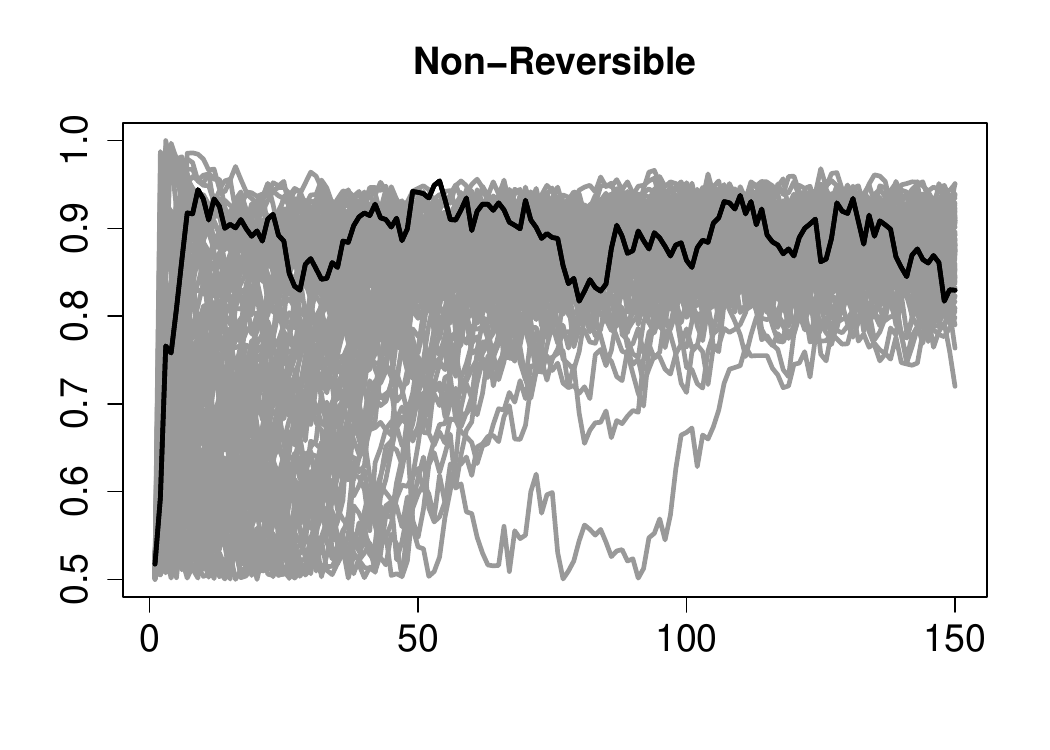} \,
\includegraphics[width=.31\textwidth]{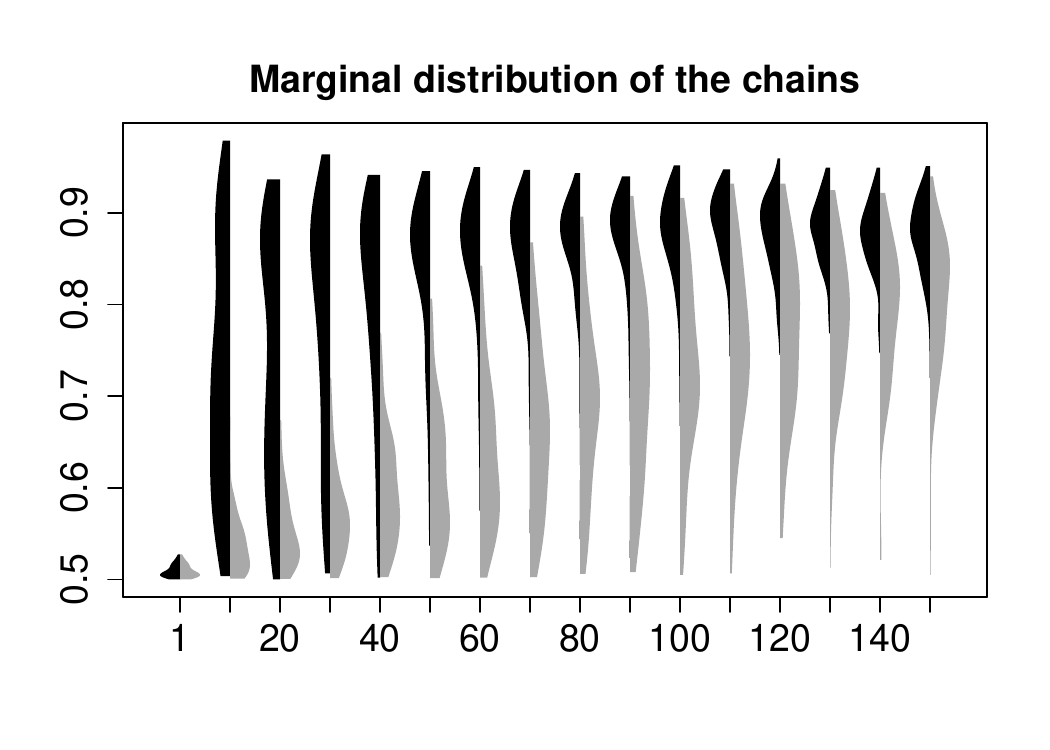} 
 \caption{\scriptsize{
Left and center: traceplots of $100$ independent runs of the size of the largest cluster for $150$ iterations (after a thinning of size $n$, i.e.\ after $150\times n$ total updates) of $\PMG$ (left) and $\PNR$ (center) in Algorithm \ref{alg:PNR}. Right: empirical estimates of the marginal distribution at every $10\times n$ iterations for $\PMG$ (gray) and $\PNR$ (black). The model is the one in \eqref{eq:original_model} with $K = 2$, $f_\theta(y) = N(y \mid \theta, 1)$, $p_0(\theta) = N(\theta \mid 0, 1)$ and $\bm{\alpha} = (0.5, 0.5)$, while the data are generated as in \eqref{eq:truth_example}. Initial configurations $c^{(0)}$ are sampled uniformly from $[K]^n$.
  }}
 \label{fig:example}
\end{figure}

We use $\PMG$ to sample from the resulting posterior $\pi(c)$, leading to a Markov chain $\{c^{(t)}\}_{t=0,1,2,\dots}$ on $[K]^n$. 
The left panel of Figure \ref{fig:example} displays $100$ independent traceplots of the size of the largest cluster in $\{c^{(t)}\}_t$, with all chains independently initialized by sampling $c^{(0)}$ uniformly from $[K]^n$.
Most runs are still far from $0.9$ (value around which we expect the posterior to concentrate) after $150\times n$ iterations. Indeed trajectories exhibit a typical random-walk behaviour, with slow convergence to stationarity. The center panel instead 
shows the traceplots generated by the
same number of runs and iterations of $\PNR$, the non-reversible scheme we propose in Section \ref{sec:PNR_def} (see Algorithm \ref{alg:PNR} therein for pseudocode and full details). The chain appears to reach the high-probability region and forget the starting configuration much faster. This is also clear from the right panel, which displays empirical estimates of the marginal distributions of the Markov chains induced by $\PMG$ and $\PNR$ over time.

\subsection{Lifted samplers for discrete spaces}\label{section:lifted} 
Our proposed sampler is inspired by classical non-reversible MCMC constructions \citep{diaconis2000analysis, fearnhead2018piecewise}, which loosely speaking force the algorithm to persistently move in one direction as much as possible. 
To illustrate the idea, consider an arbitrary probability distribution $\pi$ on a countable space $\C$ and an augmented distribution
\[
\tilde{\pi}(c, v) = \frac{\pi(c)}{2}, \quad c \in \C, v \in \{-1, +1 \},
\]
on the space $\sX = \C \times \{-1, + 1\}$, so that $\pi$ is the marginal distribution of $\tilde{\pi}$ over $\C$. 
Given two Markov kernels $\{q_{+1}(c,\cdot)\}_{c\in \mathcal{C}}$ and $\{q_{-1}(c,\cdot)\}_{c\in \mathcal{C}}$ on $\mathcal{C}$, let $P_{\text{lift}}$ be the non-reversible $\tilde{\pi}$-invariant Markov kernel defined in Algorithm \ref{alg:lifted}. 
\begin{algorithm}[htbp]
\begin{algorithmic}
\State Generate $\tilde{c} \sim q_{v}(c,\cdot)$.
\State Set $(c', v') = (\tilde{c}, v)$ with probability
\[
\min\left\{1, \frac{\pi(\tilde{c})q_{-v}(\tilde{c},c)}{\pi(c)q_{v}(c,\tilde{c})} \right\}.
\]
Otherwise set $(c', v') = (c, - v)$.
\end{algorithmic}
\caption{Generating a sample $(c', v') \sim P_{\text{lift}}((c, v), \cdot)$
\label{alg:lifted}}
\end{algorithm}
The kernels are usually chosen so that $q_v(c,\cdot)$ and $q_{-v}(c,\cdot)$ have disjoint support and the variable $v\in\{-1,+1\}$ encodes a direction (or velocity) along which the chain is exploring the space: such direction is reversed only when a proposal is rejected (see Algorithm \ref{alg:lifted}). 
As a simple example, take $\C = \N$ and $q_v(c,\cdot) = \delta_{c + v}(\cdot)$, so that $v = +1$ implies that the chain is moving towards increasing values and viceversa with $v = -1$. 
Within this perspective $v$ can be seen as a ``memory bank'' which keeps track of the past history of the chain and  
introduces momentum. 
The kernel $P_{\text{lift}}$ is often referred to as a \emph{lifted} version of the (reversible) Metropolis-Hastings kernel with proposal distribution $q=0.5q_{+1}+0.5q_{-1}$ and target distribution $\pi$.
Lifted kernels can mix significantly faster than their reversible counterparts \citep{diaconis2000analysis} and are in general at least as efficient as the original method under mild assumptions \citep{bierkens2016non,andrieu2021peskun, gagnon2024theoretical}. 
However, whether or not lifting techniques achieve a notable speed-up 
depends on the features of $\pi$ and the choice of $q_{v}$. 
For example, if proposed moves are often rejected, then the direction $v$ will be reversed frequently and the chain will exhibit an almost reversible behaviour; while if the sampler manages to make long `excursions' (i.e.\ consecutive moves without flipping direction) one expects to observe significant gains obtained by lifting. 

\textbf{Non-reversible samplers for mixture models. }
General techniques to construct non-reversible samplers for discrete spaces have been proposed in the literature, see e.g.\ \citet[Sec.3]{gagnon2024asymptotic} for constructions on partially-ordered discrete spaces and \citet{power2019accelerated} for discrete spaces with algebraic structures. 
We are, however, not aware of successful applications of these methodologies to mixture models. 
Part of the reason could be that, in order to build a lifted counter-part of $\PMG$ for mixture models, one would need to define some notion of direction or partial ordering on $[K]^n$, or more generally on the space of partitions and it is not obvious how to do so in a way that is computationally efficient and results in long excursions with persistence in direction (thus leading to substantial speed-ups). 
For example, one could directly rely on the cartesian product structure of $[K]^n$ and attach a velocity component to each coordinate, applying for example the discrete Hamiltonian Monte Carlo algorithm of \cite{nishimura2020discontinuous}: this however would not pair well with the geometry of posterior distributions $\pi(c)$ arising in mixture models and likely result in short excursions and little speed-up.

To tackle this issue, we take a different perspective on $[K]^n$, moving from the kernel $\PMG$, which is a mixture over data-points (i.e.\ over $i\in[n]$), to a kernel $\Prev$ which is a mixture over pairs of clusters (see Section \ref{sec:def_reversible} for definition). This allows us to derive an effective non-reversible sampler $\PNR$ targeting $\pi(c)$, built as a mixture of lifted samplers associated to pairs of clusters (see Section \ref{sec:PNR_def} for definition). 
Note that, while we designed our sampler to be effective for posterior distributions of mixture models, the proposed scheme can in principle be used with any distribution $\pi$ on $[K]^n$. 

\subsection{Beyond finite mixtures: random partition models}\label{sec:beyond_finite}

Selecting the number of components $K$ in model \eqref{eq:original_model} is often a hard task and can have a large impact on the resulting posterior. Therefore it is common to either place a prior on $K$ or allow for $K = \infty$, i.e.\ infinite mixtures. In order to avoid working with a varying or infinite number of parameters, it is customary to consider the implied prior on the space of partitions of $[n]$: conditional on the cluster assignments, the model is defined as before. More formally, denote with
\[
\Pi_K = \left\{ \rho = \left(S_1,\dotsc,S_K\right) \,\bigg|\, S_i \subset [n],\, S_i \neq \emptyset, \, S_i \cap S_j = \emptyset \,\,\forall i \neq j,\, \cup_{k = 1}^KS_k = [n]  \right\}
\]
the space of partitions of $[n]$ in $K$ non empty clusters and with $\Pi = \cup_K \Pi_K$ the set of all possible partitions. A popular probability measure on $\Pi$ is given by
\begin{equation}\label{eq:prior_PY}
\nu(\rho) = \frac{\prod_{k = 0}^{K(\rho)-1}(\alpha + \delta k)}{\alpha_{(n)}}\prod_{k = 1}^{K(\rho)}(1-\delta)_{(n_k(\rho)-1)}, \qquad \delta \in [0, 1), \alpha > -\delta,
\end{equation}
where $a_{(m)} = a(a+1)\dots(a+m-1)$ and $K(\rho)$ and $n_k(\rho)$ denote respectively the number and size of the clusters in $\rho$. This distribution arises from a Pitman-Yor process \citep{pitman1997two}, which is a random object whose realizations are discrete probability measures with infinitely many atoms:  the case $\delta = 0$ corresponds to the celebrated Dirichlet process \citep{Ferguson1973}.
Given $\nu \in \sP(\Pi)$ and a density $f_\theta$, the associated random partition model reads
\begin{equation}\label{eq:model_random}
Y_i \mid \rho, \bm{\theta} \simind f_{\sum_{k = 1}^{K(\rho)}\theta_k\mathbbm{1}(i \in S_k)}(y), \quad \theta_k \mid \rho \simiid p_0, \quad \rho \sim \nu.
\end{equation}
Integrating out the cluster specific parameters, as in \eqref{eq:marg_c}, the posterior becomes $\mu(\rho) \, \propto \, \nu(\rho)\prod_{k = 1}^{K(\rho)}\int_{\Theta}\prod_{i \, : \, i \in S_k}f_{\theta_{k}}(Y_i)p_0(\theta_k)\, \d \theta_k$
 and the marginal algorithm then proceeds analogously to Algorithm \ref{alg:PMG}. 
The extension of the non-reversible sampler to this more complicated setting, on the space of partitions and with varying number of clusters, is discussed in Section \ref{sec:infinite} of the supplement.

\subsection{Related literature}\label{sec:lit}

\textbf{Bayesian mixture models. }
Bayesian mixture models are a classical topic which has received a lot of attention in the last decades, both with a fixed $K$ \citep{marin2005bayesian, fruhwirth2006finite} and with a random or infinite number of components \citep{lee2013defining, de2013gibbs}. 
The challenges related to sampling from the resulting posterior distribution have been also discussed extensively, see e.g.\ examples in \citep{diebolt1994estimation, celeux2000computational,stephens2000dealing, lee2009bayesian, hobert2011improving, hastie2015sampling, dahl2022sequentially}, and 
the marginal and conditional samplers we compare with are arguably the most popular schemes that are routinely used to accomplish such tasks \citep[Section 1.4]{marin2005bayesian}. Beyond finite mixtures we also compare our proposal to split-and-merge samplers \citep{jain2004split, dahl2022sequentially}, which try to move large subsets of allocation variables at each iteration.

\textbf{Lifted MCMC for statistical models with discrete parameters. }
Non-reversible MCMC samplers have been previously designed for and applied to Bayesian statistical models with discrete parameters, such as variable selection, permutation-based and graphical models; see e.g.\ \citet{power2019accelerated,gagnon2024asymptotic,
schauer2024causal} and references therein. 
However, posterior distributions arising from such models are usually strongly concentrated and highly non-smooth, limiting the length of excursions and speed-ups obtained with lifted chains. 
As a result, one often ends up observing large gains (e.g.\ hundred-fold) when targeting uniform or prior distributions (which are usually quite flat) and more modest gains (e.g.\ two-fold) when targeting actual posterior distributions used in practice; see e.g.\ \citet[Figures 1 and 3]{
schauer2024causal}, \citet[Table 1]{power2019accelerated} or \citet[Figure 1]{gagnon2024asymptotic}\footnote{This is in contrast with applications of lifting techniques to discrete models arising in Statistical Physics (see e.g.\ \citealp{vucelja2016lifting}), which often feature a higher degree of symmetry and smoothness, thus making non-reversible MCMC methods more effective; see e.g.\ \citet[Table 1]{power2019accelerated} for numerical examples  
and \citet{faulkner2024sampling} for a recent review.}.  
Instead, a key feature of our proposed sampler is that, in many cases of interest, the speed-up relative to its reversible counter-part remains large even in the presence of observed data (i.e.\ for the actual posterior). We argue that this is due to statistical features of mixture models that make them well-suited to appropriately designed non-reversible samplers (such as $\PNR$); see Section \ref{sec:specificities} for more details.

\subsection{Structure of the paper}
In Section \ref{sec:non_reversible} we introduce our proposed non-reversible Markov kernel $\PNR$, which targets $\pi(c)$ in \eqref{eq:marg_c}. In Section \ref{sec:as_var} we first show that, after accounting for computational cost, $\PNR$ cannot perform worse than $\PMG$, in terms of asymptotic variance, by more than a factor of four.
This is done by combining some variations of classical results on asymptotic variances of lifted samplers  
with a Peskun comparison between $\PMG$ and an auxiliary reversible kernel $\Prev$. 
In Section \ref{sec:optimal_scaling} we continue the comparison between $\PMG$ and $\PNR$ for finite mixtures, showing that in the prior case the latter improves on the former by an order of magnitude, i.e.\ reducing the convergence time from $\mathcal{O}(n^2)$ to $\mathcal{O}(n)$. This is done through a scaling limit analysis, which proves that, after rescaling time by a factor of $n^2$, the evolution of the frequencies $n_k(c)$ evolving according to $\PMG$ converges to a Wright-Fisher process \citep{ethier1976class}, which is a diffusion on the probability simplex. 
In contrast, when the chain evolves according to $\PNR$, we obtain convergence to a non-singular piecewise deterministic Markov process \citep{davis1984piecewise} after rescaling time by only a factor of $n$. 
Finally, Sections \ref{sec:simulations} and \ref{sec:data} provide various numerical experiments, with respectively simulated and real data, where $\PNR$ is shown to significantly outperform the alternative schemes in sampling from mixture models posterior distributions. 
The supplement contains all the proofs, together with several additions: the latter include the extension to random partition models (Section \ref{sec:infinite}), the comparison with the conditional sampler (Section \ref{sec:conditional}), and various simulation studies (Section \ref{sec:additional_sim}). R code to replicate all the numerical experiments can be found at \url{https://github.com/paolomanildo/NonReversible_Mixtures}.

\section{A non-reversible marginal sampler}\label{sec:non_reversible}

\subsection{A reversible sampler that operates over pairs of clusters}\label{sec:def_reversible}
Given a fixed $K$, let $\pi(c)$ be an arbitrary probability distribution on $[K]^n$, such as \eqref{eq:marg_c} in the context of finite mixtures, and denote the set of ordered pairs in $[K]$, with cardinality $K(K-1)/2$, as 
$\mathcal{K} = \left\{(k, k') \in [K]^2 \, : \, k < k' \right\}$. 
As an intermediate step towards defining $\PNR$, we first consider a $\pi$-reversible Markov kernel on $[K]^n$ defined as
\begin{align}\label{eq:P_rev}
\Prev(c,c')&= \sum_{(k, k')\in \mathcal{K}} p_c(k, k')P_{k, k'}(c,c')&c,c'\in[K]^n\,,\\
\label{eq:prob_components}
p_c(k, k')&= \frac{n_k(c) + n_{k'}(c)}{(K-1)n},
&
(k, k') \in \mathcal{K}
\end{align}
where $p_c$ is a probability distribution on $\mathcal{K}$ for each $c\in[K]^n$, i.e.\ $\sum_{(k, k') \in \mathcal{K}}p_c(k, k')=1$, and $P_{k, k'}$ is the $\pi$-reversible Metropolis-Hastings (MH) kernel that proposes to move a uniformly drawn point from cluster $k$ to cluster $k'$ or viceversa with probability $1/2$.
The resulting kernel $\Prev$ is defined in Algorithm \ref{alg:Prev} where, for ease of notation, for every $c \in [K]^n$, $i\in[n]$ and $k\in[K]$ we write $(c_{-i}, k) \in [K]^n$ for the vector $c$ with the $i$-th entry $c_i$ replaced by $k$. Sampling from $p_c$ can be performed efficiently based on two samples from the uniform distribution: see Algorithm \ref{alg:sampling_clusters} in Section \ref{sec:sampling_weights} of the supplement.

\begin{algorithm}[htbp]
\begin{algorithmic}
\State Sample $(k, k') \sim p_c$ as in Algorithm \ref{alg:sampling_clusters}.
    \State Set 
    $(k_-,k_+)=(k,k')$ with probability $1/2$ and $(k_-,k_+)=(k',k)$ otherwise
    \State 
    If $n_{k_-}(c) = 0$ set $c'=c$. \State If $n_{k_-}(c)> 0$ sample $i \sim \text{Unif}\left(\left\{i' \in [n]  \, \mid \, c_{i'} = k_-\right\}\right)$ and set $c' = (c_{-i}, k_{+})$ with probability $\min\{1, r\left(c,i,k_-, k_{+} \right)\}$ and $c' = c$ otherwise, where
\begin{equation}\label{eq:acc_ratio}
r(c,i,k_-, k_{+}) = \left(\frac{n_{k_-}(c)}{n_{k_+}(c) + 1}\right)\frac{\pi(c_i = k_+ \mid c_{-i})}{\pi(c_i = k_- \mid c_{-i})},
\end{equation}
\end{algorithmic}
\caption{
Generating a sample $c'\sim \Prev(c,\cdot)$  
\label{alg:Prev}}
\end{algorithm}

Despite the fact that $\Prev$ is a mixture with weights $p_c$ depending on the current state $c$, invariance with respect to $\pi$ is preserved, as proven in the next lemma. The key point is that $p_c(k, k')$ only depends on $n_{k}(c) + n_{k'}(c)$ and $P_{k, k'}$ leaves the latter quantity unchanged.
\begin{lemma}\label{lemma:erg_Prev}
The Markov kernel $\Prev$ defined in Algorithm \ref{alg:Prev} is $\pi$-reversible. Moreover, if $\pi(c) > 0$ for every $c \in [K]^n$ it is also irreducible, aperiodic and uniformly ergodic.
\end{lemma}

Both $\PMG$ and $\Prev$ can be interpreted as reversible Metropolis-Hastings schemes that propose single-point moves.
Specifically, $\PMG$ and $\Prev$ 
propose moving datapoint $i$ to cluster $k$ with, respectively, probabilities
\begin{align*}
a_{\text{MG}}(i, k) 
&=
\frac{\pi(c_i = k \mid c_{-i})}{n}
\quad\hbox{and}\quad
a_{\text{R}}(i, k)= \frac{n_{c_i}(c) + n_k(c)}{n_{c_i}(c)}
    \frac{
    \mathbbm{1}(c_i \neq k)
    }{2(K-1)n}\,,
\end{align*}
for $(i,k)\in[n]\times [K]$.
For $\PMG$ the Metropolis-Hastings acceptance probability is always one, while for $\Prev$ it is not.
It is interesting to note that
\[
a_{\text{R}}(i, k) \geq \frac{1}{2(K-1)n} \geq \frac{1}{2(K-1)}a_{\text{MG}}(i, k)\,,
\]
meaning that the proposal probabilities of $\Prev$ can be at most $2(K-1)$ times smaller than the ones of $\PMG$. 
This connection will help providing formal comparison results between $\Prev$ and $\PMG$ in Section \ref{sec:as_var} 
(see Theorem \ref{thm:asymp_variances} for more details). 
We postpone details on these comparisons to Section \ref{sec:as_var} and now focus on how to leverage the mixture representation of $\Prev$ in \eqref{eq:P_rev} to build effective non-reversible algorithms. 

\subsection{The proposed non-reversible sampler}\label{sec:PNR_def}
Consider the extended target distribution
\begin{align}\label{eq:aug_invariant}
\tilde{\pi}(c,v )&:=  \pi(c)\left(\frac{1}{2}\right)^{K(K-1)/2}  
\end{align}
for $c\in[K]^n,\; 
v = (v_{k, k'})_{(k, k')\in \mathcal{K}}\in\{-1, +1\}^{K(K-1)/2}$. Given a refresh operator $F$ defined later in \eqref{eq:F_kernel}, let $\PNR = F\bar{P}_{NR} F$ where
\begin{align}\label{eq:mixture_representation}
\bar{P}_{NR}((c,v),(c',v')) = \sum_{(k, k')\in \mathcal{K}} p_c(k, k')\tilde{P}_{k, k'}((c,v),(c',v'))
\end{align}
with $p_c$ defined as in \eqref{eq:prob_components}
and 
$\tilde{P}_{k, k'}$ defined in Algorithm \ref{alg:Pkk}.
The kernel $\tilde{P}_{k, k'}$ operates on the $k$-th and $k'$-th clusters, and it is a lifted counter-part of $P_{k, k'}$, with associated velocity component $v_{k, k'}$.
In this construction, the velocity vector $\bv$ is $K(K-1)/2$ dimensional and only one of its component is actively used to move $c$ at each iteration. 
\begin{algorithm}[htbp]
\begin{algorithmic}
\State Set $(k_-,k_+)=(k,k')$ if $v_{k, k'} = +1$, and $(k_-,k_+)=(k',k)$ if $v_{k, k'} = -1$
\State If $n_{k_{-}}(c) = 0$, set $(c',\bv') = (c,\bv^{(\textit{flip})})$, with $\bv^{(\textit{flip})}=(v_{-(k, k')},-v_{k, k'})$
\State If $n_{k_{-}}(c) > 0$, sample $i \sim \text{Unif}\left(\{i' \in [n]  \, \mid \, c_{i'} = k_{-}\}\right)$ and set $(c',\bv') = ((c_{-i}, k_+),\bv)$ with probability $\min\{1, r(c,i,k_-, k_{+})\}$ and otherwise  $(c',\bv') = (c,\bv^{(\textit{flip})})$, with $r(c,i,k_-, k_{+})\}$ defined in \eqref{eq:acc_ratio}
\end{algorithmic}
\caption{Generating a sample $\left(c', \bv'\right)\sim \tilde{P}_{k, k'}((c, \bv),\cdot)$
\label{alg:Pkk}}
\end{algorithm}
The resulting pseudo-code is given in Algorithm \ref{alg:PNR}.
\begin{algorithm}[htbp]
\begin{algorithmic}
\State Update $(c,v) \sim F((c,v),(c,\cdot))$
\State Sample $(k, k') \sim p_c$ as in Algorithm \ref{alg:sampling_clusters}.
\State Sample $(c',v')\sim \tilde{P}_{k,k'}((c,v),\cdot)$ as in Algorithm \ref{alg:Pkk}.
\State Update $(c',v') \sim F((c',v'),(c',\cdot))$
\end{algorithmic}
\caption{One step of the non-reversible kernel $(c',v') \sim \PNR((c,v),(c',v'))$
\label{alg:PNR}}
\end{algorithm}

We define the kernel $F$ as the refresh operator
\begin{equation} \label{eq:F_kernel}
F((c,v),(c',v')) = \mathbbm{1}(c = c')\left\{\frac{\xi}{n}\left(\frac{1}{2}\right)^{\frac{K(K-1)}{2}} + \left(1-\frac{\xi}{n}\right) \mathbbm{1}(v = v')\right\}
\end{equation}
which depends on a parameter $\xi \in [0,n]$. While $\xi>0$ ensures aperiodicity, we do not expect the value of $\xi$ to impact practical performances provided it is set to a small value: this is empirically confirmed in Section \ref{sec:xi} of the supplement and therefore we always set $\xi = 1$. Moreover, the application of $F$ at the beginning and the end of each iteration is needed for technical reasons and allows for a neater theory, see Section \ref{sec:general_results} for details.  The next lemma shows that $\PNR$ is a valid $\tilde{\pi}$-invariant kernel.
\begin{lemma}\label{lemma:ergodicity}
For any probability distribution $\pi$ on $[K]^n$, $\PNR$ defined in Algorithm \ref{alg:PNR} is $\tilde{\pi}$-invariant, with $\tilde{\pi}$ as in \eqref{eq:aug_invariant}.
Moreover, if $\xi > 0$ and $\pi(c) > 0$ for every $c \in [K]^n$, then $\PNR$ is irreducible, aperiodic and uniformly ergodic.
\end{lemma}
In Section \ref{sec:variant} of the supplement we consider a variant of $\PNR$, which keeps the same pair $(k, k')$ for multiple iterations. 

\subsubsection{Specificities of mixture models that make $\PNR$ work well}\label{sec:specificities}

We now discuss at an informal level some of the specificities of the posterior distribution $\pi(c)$ arising from mixture models that make $\PNR$ work well.

\vspace{0.3cm}
\noindent
\textbf{Lack of identifiability and concentration. }
An important statistical feature of mixture models is that cluster labels are in general not identifiable as $n\to\infty$, meaning that even when $n$ is large there is non-vanishing uncertainty on the value of $c_i$. 
In other words, while the posterior distribution of $\bm{w}$ and $\bm{\theta}$ 
concentrates as $n\to\infty$, the one of $c$ does not (not even at the level of marginals, meaning that, for every fixed $i$, $\pi(c_i)$ does not converge to a point mass as $n\to\infty$);
see e.g.\ \cite{nguyen2013convergence, guha2021posterior} and references therein for asymptotic results for mixture model posteriors.
Intuitively, lack of concentration occurs because the information about each individual $c_i$ does not grow with $n$ (since each $c_i$ is associated to a single datapoint). This also tends to make posteriors flatter, i.e.\ moving one observation from one cluster to another usually leads to a small change in the target distribution. By contrast, many models with discrete parameters lead to posteriors that become increasingly more concentrated and rough as $n\to\infty$, which has a major impacts on the convergence properties of MCMC algorithms targeting them, including making standard MCMC converge faster (see e.g.\ \citealp{yang2016computational, zhou2022dimension, zhou2023complexity}) and lifting techniques less effective (as discussed in Section \ref{sec:lit}).

\vspace{0.3cm}
\noindent
\textbf{Cancellations in the acceptance ratio. }
For $\pi(c)$ as in \eqref{eq:marg_c}, the MH ratio $r(c,i,k_-, k_{+})$ reads
\begin{equation}\label{eq:acc_ratio_mixtures}
r(c,i,k_-, k_{+}) = \left(\frac{\alpha_{k_+} + n_{k_+}(c)}{n_{k_+}(c) + 1}\right)\left(\frac{n_{k_-}(c)}{\alpha_{k_-} + n_{k_-}(c) - 1}\right)\frac{p(Y_i \mid Y_{-i}, c_{-i}, c_i = k_+)}{p(Y_i \mid Y_{-i}, c_{-i}, c_i = k_-)}.
\end{equation}
Interestingly, the proposal probability almost matches the term $\alpha_k + n_k(c)$ arising from the prior. 
In particular, by writing
\[
\left(\frac{\alpha_{k_+} + n_{k_+}(c)}{n_{k_+}(c) + 1}\right)\left(\frac{n_{k_-}(c)}{\alpha_{k_-} + n_{k_-}(c) - 1}\right) = \left(1+\frac{\alpha_{k_+}-1}{n_{k_+}(c)+1} \right)\left(1+\frac{\alpha_{k_-}-1}{n_{k_-}(c)} \right)^{-1},
\]
we see that the first part of \eqref{eq:acc_ratio_mixtures} goes to $1$ as $n_{k_+}(c)$ and $n_{k_-}(c)$ increase, for every fixed value of $\boldsymbol{\alpha}$. Notice that with $\alpha_k = 1$ for every $k$ this ratio is always equal to $1$. 
This cancellation contributes to make \eqref{eq:acc_ratio_mixtures} closer to $1$ and thus to make excursions of $\PNR$ longer. 

\vspace{0.3cm}
\noindent
\textbf{Flatness in the tails and behavior out of stationarity. }
Interestingly, also the ratio of predictive distributions in \eqref{eq:acc_ratio_mixtures} tends to get close to $1$ for partitions that do not correspond to well-identified and separate clusters, meaning that mixture model posteriors $\pi(c)$ become increasingly flatter in the tails.
To illustrate this, consider the common situation when labels are initialized uniformly at random, i.e.\ $c_i^{(0)} \simiid \text{Unif}([K])$. 
In this situation, by construction, clusters are similar to each other under $c^{(0)}$, resulting in ratios of predictive distributions  
that are close to $1$ (and converge to $1$ as $n\to\infty$). As a consequence, the non-reversible chain will proceed almost deterministically without flipping directions until clusters start to differentiate significantly.
This is indeed the behavior observed in the middle panel of Figure \ref{fig:example}, as well as in Section \ref{sec:posterior_sim} and \ref{sec:poisson_app} of the supplement with different values of $K$ and likelihood kernels.
More generally, in mixture model contexts, non-reversibility is particularly helpful during the transient phase, where the algorithm has not yet reached the high-probability region under $\pi$ and has to explore large flat regions of the state space\footnote{This is, again, in contrast with typical Bayesian discrete models that lead to posteriors with large ``discrete gradients'' in the tails providing strong enough signal for reversible schemes to converge fast in the first part of the transient phase \citep{yang2016computational,zhou2022dimension, zhou2023complexity}.}.

\vspace{0.3cm}
\noindent
\textbf{Overlapping components and the overfitted case. }
Another situation that makes ratios of predictive distributions close to $1$ is when the actual clusters in the data have a considerable overlap. 
An extreme case of this situation is 
when the true number of components $K^*$ (assuming data were actually generated by a well-specified mixture model) is strictly smaller than $K$, which amounts to assuming that a plausible upper bound on $K^*$ is known and $K$ is set to such value (instead of the less plausible scenario where $K^*$ itself is known). This is often called the overfitted case, see e.g.\ \citet{rousseau2011asymptotic} for a theoretical exploration of its asymptotic properties, and it is a common situation since in many context (e.g.\ density estimation) it is preferable to overshoot rather than undershoot the value of $K$ and thus practitioners often set $K$ to some conservative, moderately large value. See Section \ref{sec:overfitting} for more discussion on the overfitted case and empirical evidence that in this setting the improvement of $\PNR$ over the latter is particularly apparent.

\section{Asymptotic variance comparison results}\label{sec:as_var}
In this section we compare the various kernels discussed above in terms of asymptotic variances. 
Given a Markov chain $\{X_t\}_t$ with a $\pi$-invariant Markov kernel $P$ started in stationarity, the asymptotic variance of the associated MCMC estimator is given by
\[
\Var(g, P) = \lim_{T \to \infty} T\Var\left(\frac{1}{T}\sum_{t = 1}^Tg(X_t) \right) = \Var_\pi(g)+2\sum_{t = 1}^\infty \Cov \left(g(X_0), g(X_t) \right),
\]
for every function $g$ such that $\Var_\pi(g)$ is well-defined.

\subsection{Ordering of reversible and non-reversible schemes}

To fairly compare $\PMG$, $\Prev$ and $\PNR$ we adjust for their cost per iteration, which is driven by the number of predictive density evaluations\footnote{We consider the additional cost of sampling directions to be negligible, since by Algorithm \ref{alg:sampling_clusters} in Section \ref{sec:sampling_weights} of the supplement it requires only two samples from a uniform distribution.} required
to construct either the full conditional distribution in \eqref{eq:predictive} or the Metropolis-Hastings ratio in \eqref{eq:acc_ratio}. Thus we consider their \textit{lazy} counterparts, defined as $Q_{\text{MG}}(c,c') = K^{-1}P_{\text{MG}}(c,c') +(1-K^{-1})\mathbbm{1}(c=c')$, $Q_{\text{R}}(c,c') = 0.5P_{\text{R}}(c,c') + 0.5\mathbbm{1}(c=c')$ and
 $Q_{\text{NR}}(c,c') = 0.5P_{\text{NR}}(c,c') + 0.5\mathbbm{1}(c=c')$. These Markov kernels propose a movement with probability inversely proportional to the number of predictive density evaluation they require at each iteration; otherwise, they remain at the current state. 
 
The next theorem provides ordering results for the asymptotic variances of $Q_{\text{MG}}$, $Q_R$ and $Q_{\text{NR}}$.
In principle, these kernels are not directly comparable, since $Q_{\text{MG}}$ and $Q_R$ are defined 
on $[K]^n$ while $Q_{\text{NR}}$ is defined on $[K]^n \times \{-1,+1\}^{K(K-1)/2}$. Nonetheless, we are only interested in estimating expectations of test functions that depend on $c$ alone, so that we can restrict attention to those, as usually done in non-reversible MCMC literature \citep{andrieu2021peskun, gagnon2024theoretical}.
Given a test function $g \,: \, [K]^n \, \to \R$, with a slight abuse of notation, we also use $g$ in $\Var(g, Q_{\text{NR}})$ to denote the function defined as $g(c,v)=g(c)$ for all $(c,v) \in [K]^n \times \{-1,+1\}^{K(K-1)/2}$. 
Below $\Var_\pi(g)$ denotes $\Var(g(X_0))$ for $X_0 \sim \pi$.
\begin{theorem} \label{thm:asymp_variances}
Let $\pi$ be a probability distribution on $[K]^n$ and $ g : [K]^n \to \mathbb{R}$. Then 
\begin{equation} \label{eq:asymp_variances}
\text{Var}(  g,Q_{\text{NR}}) \leq \text{Var}(  g,Q_R) \leq 4\text{Var}(  g,Q_{\text{MG}}) + 3 \text{Var}_\pi( g)\,.
\end{equation}
\end{theorem}

Since in most realistic applications $\Var(g, Q_{\text{MG}})$ is much larger than $\Var_\pi(g)$, the inequality in \eqref{eq:asymp_variances} implies that $Q_{\text{NR}}$ can be worse than $Q_{\text{MG}}$, in terms of variance of the associated estimators, only by a factor of $4$. The first inequality in \eqref{eq:asymp_variances} relies on extending classical asymptotic variance comparison results for lifted kernels to the case of state-dependent mixtures such as in $\PNR$, as shown in 
Section \ref{sec:as_var_general} of the supplement.
We stress that the results of Theorem \ref{thm:asymp_variances} hold uniformly, in the sense that no assumptions on $\pi$ are needed. 
In the next sections, we will see that $\PNR$ can lead to significant improvements (e.g.\ by a factor of $n$) relative to $\PMG$.

\begin{remark}\label{rmk:latent_parameters}
Theorem \ref{thm:asymp_variances} can be extended in several directions, which we explore in the supplement. First, as discussed in 
Section \ref{sec:marginal}, samples from $\pi(c)$ can be used to perform full posterior inference on $\pi(c, \bm{\theta}, \bm{w})$ defined in \eqref{eq:augmented_posterior}, simply drawing samples from $\pi(\bm{\theta}, \bm{w} \mid c)$. The ordering of \eqref{eq:asymp_variances} carries over to this context, as shown in Corollary \ref{crl:extension_asymptotic_variance} of Section \ref{sec:extension_parameters}. Second, Proposition \ref{prop:asymp_variances_cond} in Section \ref{sec:conditional} shows that $\text{Var}(  g,P_{\text{MG}}) \leq \text{Var}(  g,P_{\text{CD}})$, where $\PCD$ is the kernel of the conditional sampler. Finally, Theorem \ref{thm:random_partitions} in Section \ref{sec:infinite} extends \eqref{eq:asymp_variances} to general random partition models.
\end{remark}

\section{Scaling limit analysis}\label{sec:optimal_scaling}

In this section we derive scaling limit results for $\PMG$ and $\PNR$ as $n\to\infty$, in the case of finite mixtures.
In general, given a sequence of discrete-time Markov chains $\{X^{(n)}_t\}_{t \in \N}$, scaling limit results \citep{gelman1997weak, roberts2001optimal} consist in showing that a time-changed process $\{Z^{(n)}_{t}\}_{t \in \R}$ defined as $Z^{(n)}_{t} = X^{(n)}_{\lceil h(n)t\rceil}$, with $h(n) \to \infty$ and $\lceil \cdot\rceil$ denoting the ceiling function, converges in a suitable sense to a non-degenerate process $\{Z_{t}\}_{t \in \R_+}$ as $n\to \infty$. 
Provided the limiting process is non-singular and ergodic, this is usually interpreted as suggesting that $\sO(h(n))$ iterations of the discrete-time Markov chain are needed to mix.  
In other words, the time rescaling required to obtain a non-trivial limit is a measure of how the process speed scales as $n$ grows.
In order to derive such results we restrict to the prior case, where the likelihood is uninformative and the posterior distribution of $c$ coincides with the prior  \eqref{eq:augmented_model}. 
This can be formally described by setting $f_\theta(y) = f(y)$, with $f$ probability density on $\sY$, so that the data do not provide any information on the labels. Given $K \in \N$, the joint distribution and full conditionals become
\begin{equation}\label{eq:marg_priori}
\pi(c) = \frac{\Gamma(|\bm{\alpha}|)}{\Gamma(|\bm{\alpha}| + n)} \prod_{k = 1}^K\frac{\Gamma(\alpha_k + n_k(c))}{\Gamma(\alpha_k)}, \quad \pi(c_i = k \mid c_{-i}) = \frac{\alpha_k + n_k(c_{-i})}{|\bm{\alpha}|+n-1},
\end{equation}
with $|\bm{\alpha}| = \sum_{k = 1}^K\alpha_k$. This is clearly a simplified setting, which allows an explicit mathematical treatment and it can be considered as an extreme case of un-identifiability and overlapping components (which are indeed all the same). 

\subsection{Marginal Gibbs sampler} 
Consider a Markov chain $\{c^{(t)}\}_{t \in \N}$ with kernel $\PMG$ and invariant distribution \eqref{eq:marg_priori}, where we suppress the dependence on $n$ for simplicity. Let
\[
X_k(c) = \frac{n_k(c)}{n} = \frac{1}{n}\sum_{i = 1}^n\mathbbm{1}(c_i = k), \quad c\in [K]^n,
\]
be the multiplicity of component $k$ and
\begin{equation}\label{eq:transf_c}
\begin{aligned}
&\bm{X}_t = (X_{t, 1}, \dots, X_{t, K}) = \left(X_1\left(c^{(t)}\right), \dots, X_K\left(c^{(t)}\right) \right)\,.
\end{aligned}
\end{equation}
Crucially, since $\pi(c_i = k\mid c_{-i})$ defined in \eqref{eq:marg_priori} only depends on the multiplicities, i.e.\ $(c_1,\dots,c_n)$ are exchangeable a priori, it follows that $\{\bm{X}_t\}_{t \in \N}$ is itself a Markov chain. 
Moreover, $\{\bm{X}_t\}_{t \in \N}$ is de-initializing for $\{c^{(t)}\}_{t \in \N}$ in the sense of \cite{roberts2001markov}, so that the convergence properties of the former are equivalent to the one of the latter (by e.g.\ Corollary $2$ therein). With an abuse of notation, we denote the kernel of $\{\bm{X}_t\}_{t \in \mathbb{N}}$ also as $\PMG$. 
Let $\{\bm{Z}_t\}_{t \in \R_+}$ be the continuous-time process with generator
\begin{equation}\label{eq:generator_WF}
\A g(\bm{x}) = \frac{1}{2}\sum_{k = 1}^K\left(\alpha_k -|\bm{\alpha}|x_k \right)\frac{\partial}{\partial x_k}g(\bm{x})+\frac{1}{2}\sum_{k, k' = 1}^Kx_k(\delta_{kk'}-x_{k'})\frac{\partial^2}{\partial x_k \partial x_{k'}}g(\bm{x}),
\end{equation}
for every $g \, : \, \Delta_{K-1} \, \to \, \R$ twice differentiable and where $\bm{x} = (x_1, \dots, x_K)$. Such process exists \citep{ethier1976class} and is called Wright-Fisher with mutation rates given by $\bm{\alpha}$. In particular, $\{ \bm{Z}_t\}_{t \in \R_+}$ is a diffusion taking values in $\Delta_{K-1}$ whose stationary density is exactly $\pi$ in \eqref{eq:marg_priori}. The next theorem shows that, choosing $h(n) = n^2/2$, the continuous-time rescaling of $\{\bm{X}_t\}_{t \in \N}$ converges to $\{ \bm{Z}_t\}_{t \in \R_+}$.
\begin{theorem}\label{thm:diffusion_rev}
Let $\{\bm{Z}^{(n)}_t\}_{t \in \R_+}$ such that $\bm{Z}_t^{(n)} = \bm{X}_{\lceil\frac{n^2}{2}t\rceil}$, with $\{\bm{X}_t\}_{t \in \N}$ as in \eqref{eq:transf_c} with kernel $\PMG$. Let $\{\bm{Z}_t\}_{t \in \R_+}$ be a diffusion with generator as in \eqref{eq:generator_WF}. 
Then if $\bm{Z}^{(n)}_{0} \to \bm{Z}_0$ weakly as $n \to \infty$, we have that $\{\bm{Z}^{(n)}_{t}\}_{t \in \R_+} \to \{\bm{Z}_t\}_{t \in \R_+}$ weakly as $n \to \infty$, according to the Skorokhod topology.
\end{theorem}
The proof relies on convergence of generators, which is a standard technique when dealing with sequences of stochastic processes  \cite[Chapter 4]{ethier2009markov}. 
While this approach is common in the MCMC literature (see e.g.\ \citet{gelman1997weak,roberts2001optimal} and related works), we are not aware of applications of it to mixture model contexts. 
\begin{remark}
Theorem \ref{thm:diffusion_rev} suggests that $\sO(n^2)$ iterations are needed for $\PMG$ to converge. This is coherent with \citet[Prop.14.10.1]{khare2009rates} where, albeit motivated by a different problem, the authors show that, when targeting $\pi$ in \eqref{eq:marg_priori}, the second largest eigenvalue of $\PMG$ is $1-|\bm{\alpha}|/(n^2+|\bm{\alpha}|n-n)$. This means that $\sO(n^2)$ iterations are required to mix; see e.g.\ \citet[Thm.12.5]{levin2017markov}.
\end{remark}

In order to see why an $\mathcal{O}(n^2)$ convergence is slower than desired, consider the case $K = 2$. Then $\{X_{t, 1}\}_{t \in \N}$ is a Markov chain on $\{0, 1/n, \dots, 1\}$ and thus $\PMG$ requires $n^2$ iterations to sample from a distribution on a state space with cardinality $n$. Moreover, $\{X_{t, 1}\}_{t \in \N}$ can be seen as a random walk with transition probabilities
\begin{align*}
    \P\left(X_{t+1, 1} = x_1 + \frac{1}{n} \mid X_{t, 1} = x_1\right) &= (1-x_1)\frac{\alpha_1 + nx_1}{\alpha_1 + \alpha_2 + n-1} \approx x_1(1-x_1)\,,
    \\
    \P\left(X_{t+1, 1} = x_1 - \frac{1}{n} \mid X_{t, 1} = x_1\right)& = x_1\frac{\alpha_2 + n(1-x_1)}{\alpha_1 + \alpha_2 + n-1} \approx x_1(1-x_1)\,,
\end{align*}
when $n$ is large. Thus the probability of going up and down is almost the same, leading to the observed random-walk behaviour. This is reminiscent of classical examples from the non-reversible MCMC literature \citep{diaconis2000analysis}, where faster algorithms are devised by lifting the standard random walk. 

\subsection{Non-reversible sampler $\PNR$}
Consider now a Markov chain  $\{c^{(t)}, v^{(t)}\}_{t \in \N}$ with kernel $\PNR$. Define $\bm{X}_t$ as in \eqref{eq:transf_c} and $\bm{V}_t = \left\{V_{t, k, k'} \right\}_{(k, k')\in [K]^2}$ as
$V_{t, k, k'} =0$
 if $k = k'$,
 $V_{t, k, k'} =v^{(t)}_{k, k'}$
 if $k < k'$ and
 $V_{t, k, k'} =-v^{(t)}_{k', k}$
 if $k > k'$. 
This means that $V_{t, k', k} = +1$ implies that we are proposing from cluster $k'$ to $k$, for every pair $(k, k')$.

By exchangeability arguments as above, it is simple to see that $\{(\bm{X}_t, \bm{V}_t)\}_{t \in \N}$ is de-initializing for $\{c^{(t)}, v^{(t)}\}_{t \in \N}$ and thus it has the same convergence properties. 
A technical issue for the scaling limit analysis is that, when $X_{t, k} = 0$ for some $k$ then one of the velocities jumps deterministically to $V_{t,k', k} = +1$ with $k'\neq k$. 
To avoid complications related to such boundary effects, we study the scaling of the process in the set
\begin{equation}\label{eq:domain_PDMP}
E_M = \left\{ \bm{x} \in \Delta_{K-1} \, \mid \, x_k \geq \frac{1}{M} \text{ for every $k$}\right\}.
\end{equation}
Therefore, we define $\left\{\bm{X}^{(M)}_t, \bm{V}^{(M)}_t\right\}_{t \in \N}$ as a modification of $\{\bm{X}_t, \bm{V}_t\}_{t \in \N}$, which remains constant as soon as $X^{(M)}_{t, k} \leq 1/M$ for some $k$. 
The next theorem shows that, choosing $h(n) = n$, the  rescaling of $\left\{\bm{X}^{(M)}_t, \bm{V}^{(M)}_t\right\}_{t \in \N}$ converges to a piecewise deterministic Markov process $\left\{ \bm{Z}^{(M)}_t\right\}_{t \in \R_+}$, defined in Section \ref{sec:definition_PDMP} of the supplement.

\begin{theorem}\label{thm:diffusion_nonrev}
Fix $M > K$ arbitrary large and let $\left\{\bm{Z}^{(M, n)}_t\right\}_{t \in \R_+}$ such that $\bm{Z}_t^{(M, n)} = \left(\bm{X}^{(M)}_{\lceil nt\rceil}, \bm{V}^{(M)}_{\lceil nt\rceil}\right)$, where $\{\bm{X}_t, \bm{V}_t\}_{t \in \N}$ is a Markov chain with kernel $\PNR$. Let $\left\{\bm{Z}^{(M)}_t\right\}_{t \in \R_+}$ be a Markov process with generator \eqref{eq:generator_PDMP} in Section \ref{sec:definition_PDMP}.  
Then if $\bm{Z}^{(M, n)}_{0} \to \bm{Z}^{(M)}_0$ weakly as $n \to \infty$, we have that $\left\{\bm{Z}^{(M, n)}_{t}\right\}_{t \in \R_+} \to \left\{\bm{Z}^{(M)}_t\right\}_{t \in \R_+}$ weakly as $n \to \infty$, according to the Skorokhod topology.
\end{theorem}
Theorem \ref{thm:diffusion_nonrev} suggests that the overall computational cost of Algorithm \ref{alg:PNR} is $\sO(n)$ and, combined with Theorem \ref{thm:diffusion_rev}, it suggests a $\mathcal{O}(n)$ speed-up relative to $\PMG$.

\section{Simulations}\label{sec:simulations}

\subsection{Well-specified setting}\label{sec:posterior_sim} 
Consider model \eqref{eq:original_model} with $\sY=\Theta=\R$, $K = 3$,
\begin{equation}\label{eq:normal_kernel}
f_\theta(y) = N(y \mid \theta, \sigma^2), \quad p_0(\theta) = N(\theta \mid \mu_0, \sigma^2_0).
\end{equation}
and hyperparameters set to $\mu_0 = 0$ and $\sigma^2 = \sigma^2_0 = 1$. We then generate $300$ independent data sets of size $n = 1000$, each generated from the model as follows:
\begin{enumerate}
    \item Sample $\bm{w} \sim \text{Dirichlet}(\bm{\alpha})$ and $\theta_k \simiid p_0$ for $k = 1, \dots, K$.
    \item Sample $Y_i \simiid \sum_{k = 1}^Kw_kf_{\theta_k}(y)$ for $i = 1, \dots, n$.
\end{enumerate}
We set the concentration parameter first with $\bm{\alpha}=(1,1,1)$ and then with $\bm{\alpha}=(0.1,0.1,0.1)$ and use some independently generated initial configurations, so that $c_i^{(0)} \simiid \text{Unif}\left([K]\right)$.
For each dataset we target the associated posterior using $\PMG$ and $\PNR$ and store the value of the chains after $T=100\times n$ iterations. In Figure \ref{fig:posterior} we plot the corresponding proportion of labels of the first two components, i.e.\ $(n_1(c^{(T)})/n,n_2(c^{(T)})/n)$.
Since the data are generated from the (Bayesian) model, the resulting distribution of the proportions within each component should be close to the prior one, i.e.\ a Dirichlet-multinomial with parameter $\bm{\alpha}$ (since $n$ is large, this is visually close to drawing samples directly from a $\text{Dir}(\bm{\alpha})$ distribution). This test for convergence, discussed for example in \cite{geweke2004getting}, relies on the fact that sampling from the \emph{prior} distribution is equivalent to sampling from the \emph{posterior}, given data generated according to the marginal distribution of the model.

From Figure \ref{fig:posterior} it is clear that the non-reversible scheme (second column) leads to faster convergence: this is particularly manifest in the second row (corresponding to $\bm{\alpha}=(0.1,0.1,0.1)$), where the mass should be concentrated around the borders of the simplex. Indeed, both chains associated to $\PMG$ remain stuck close to the initial configuration, where the proportion within each group is close to $1/3$. This is also clear from the last column of Figure \ref{fig:posterior}, which shows that the marginal distribution of $\PNR$ (in black) converges to the stationary one after fewer iterations. 
A similar behavior is also observed in Sections \ref{sec:sim_prior_case}, \ref{sec:high_dimensional} and \ref{sec:poisson_app} of the supplement, featuring the prior $f_\theta = f$, a high dimensional Gaussian and the Poisson kernel cases.

\begin{figure}[h]
\centering
\includegraphics[width=.31\textwidth]{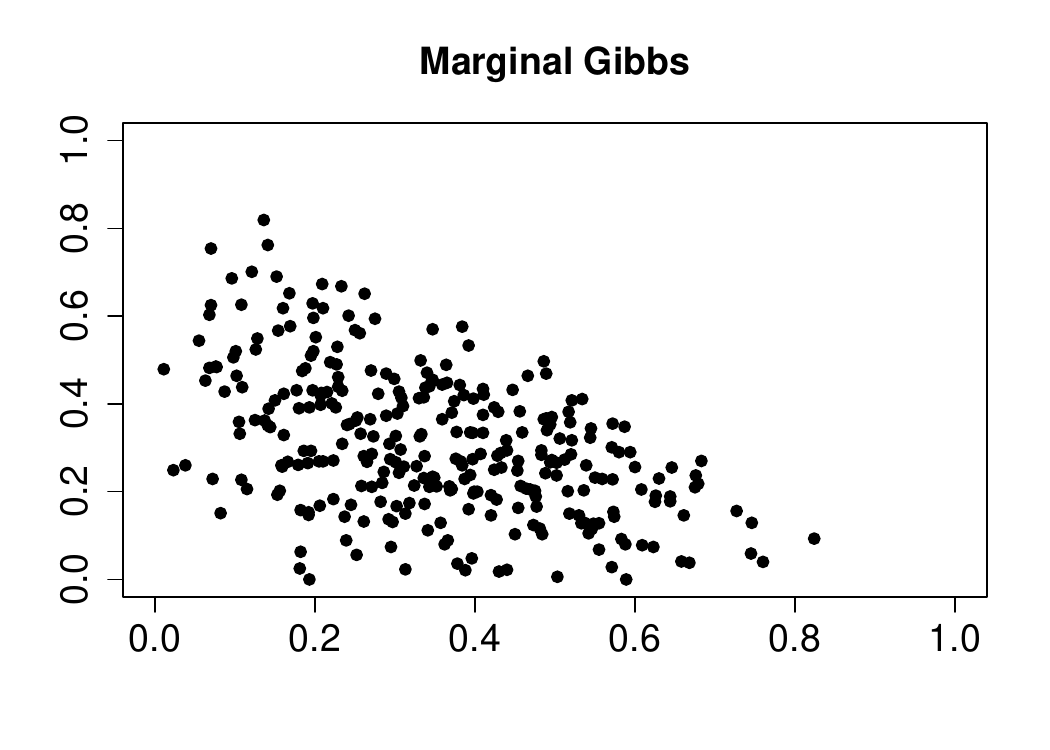} \,
\includegraphics[width=.31\textwidth]{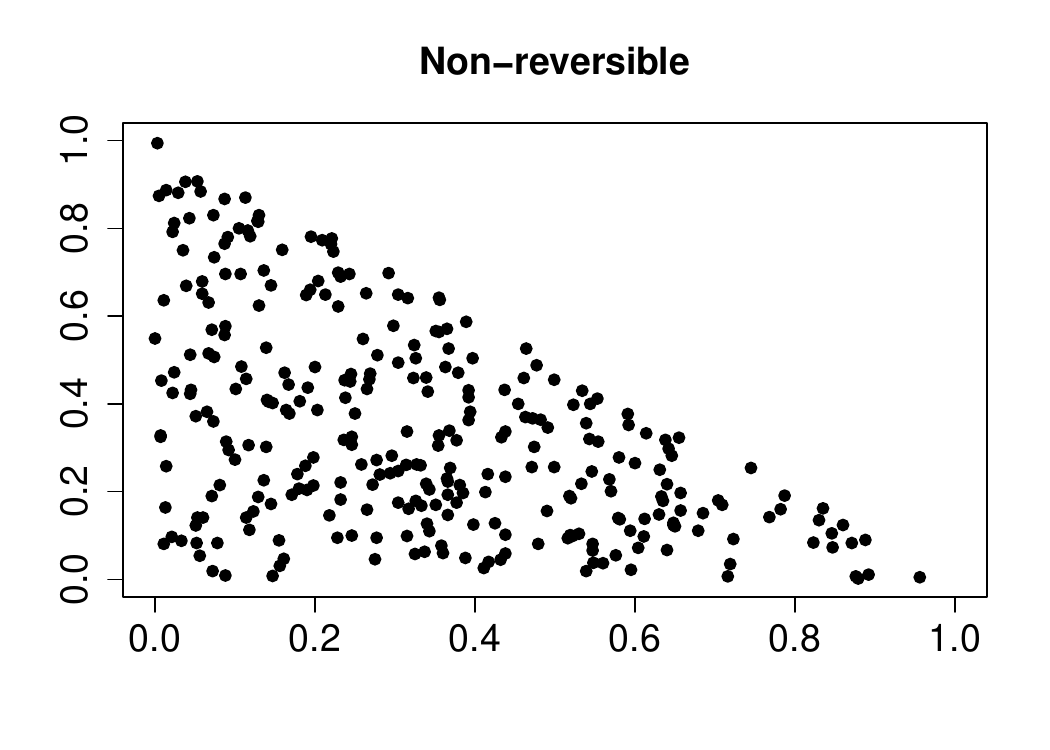} \,
\includegraphics[width=.31\textwidth]{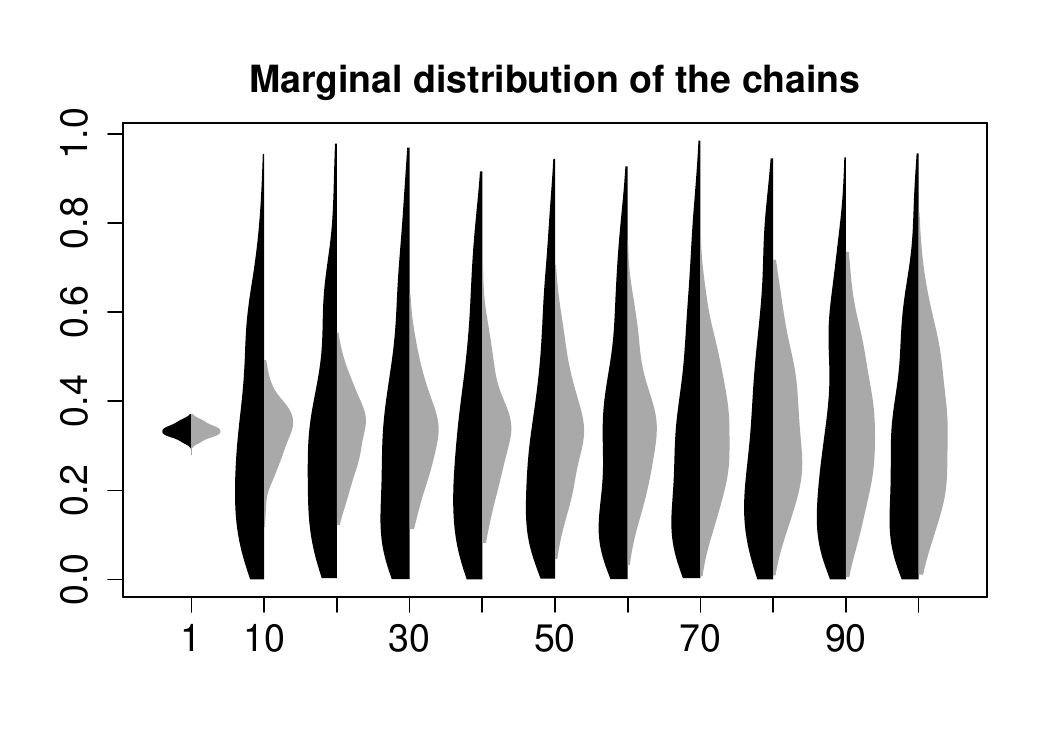}\\
\includegraphics[width=.31\textwidth]{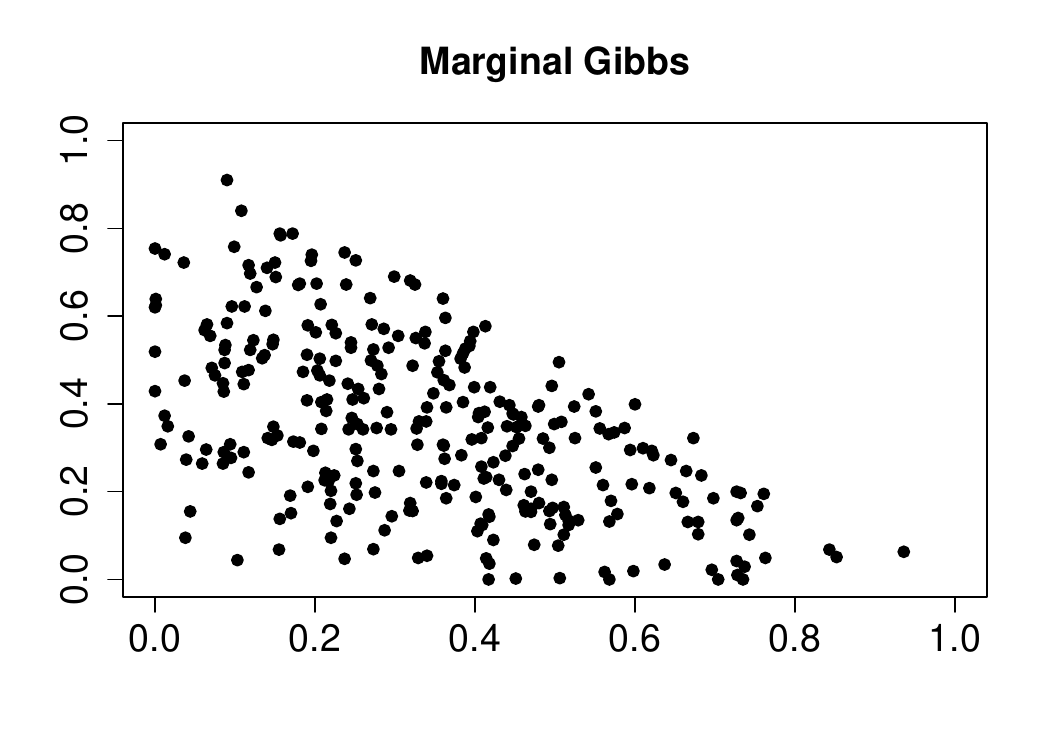} \,
\includegraphics[width=.31\textwidth]{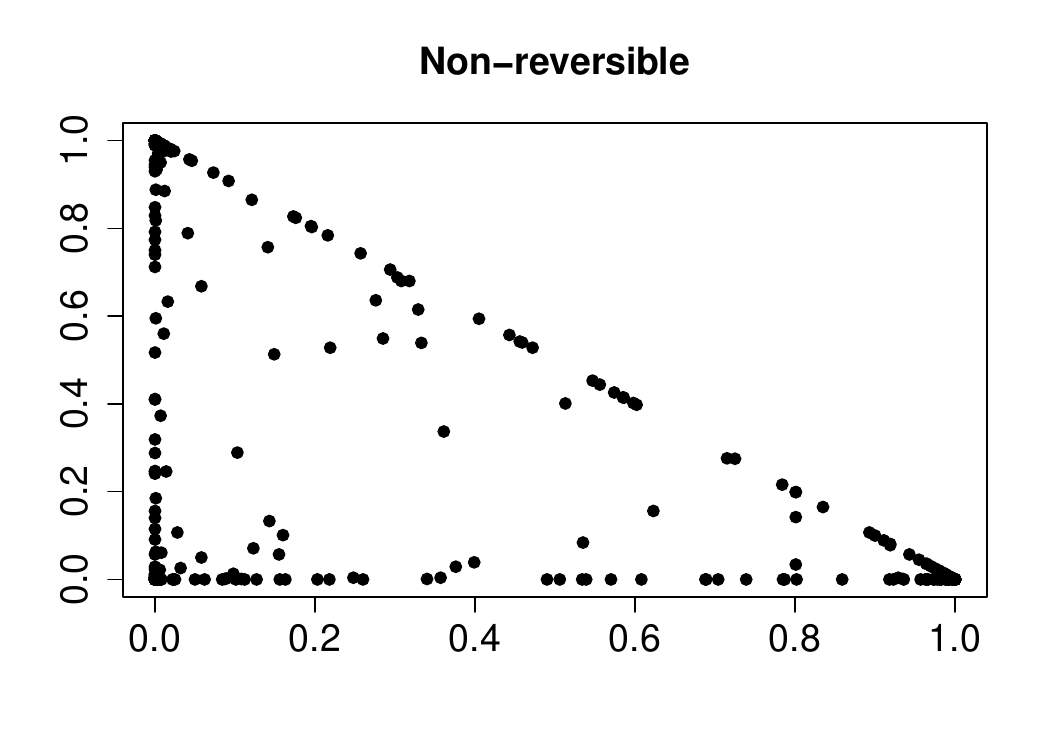} \,
\includegraphics[width=.31\textwidth]{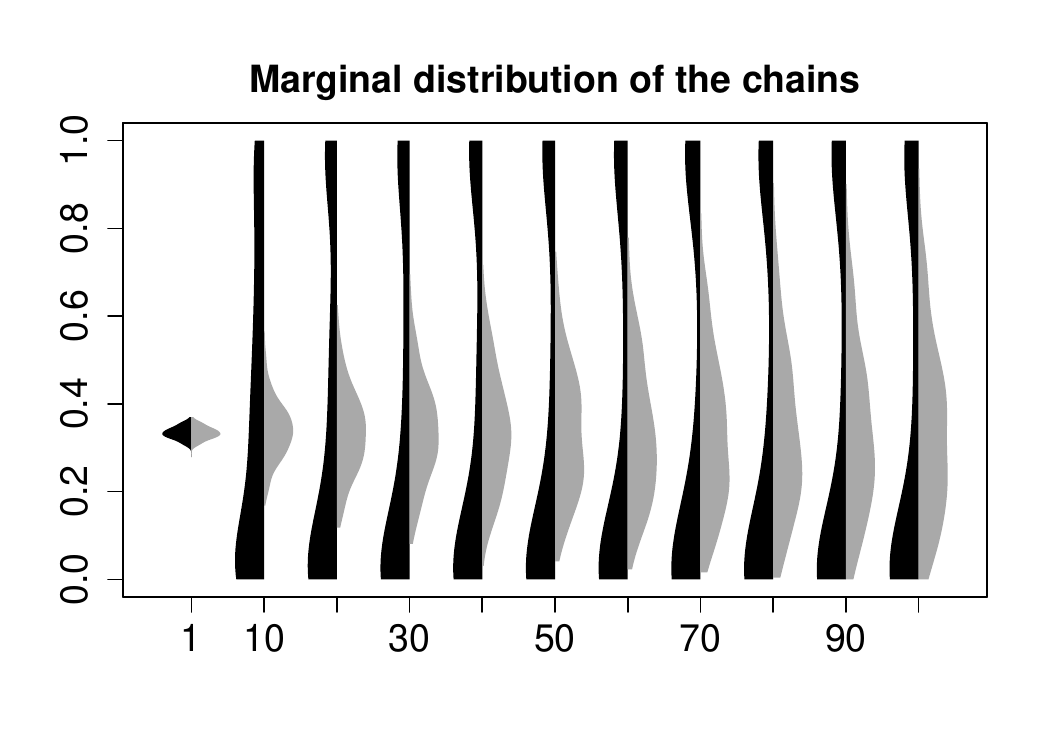}
 \caption{\scriptsize{
Left and center column: plot of the proportions of the first two components in the last of $100$ iterations (after a thinning of size $n$) over $300$ independent runs for $\PMG$ (left) and $\PNR$ (center). Right column: plot of the marginal distribution of the proportion of the first component at every $10$ iterations (after thinning) for $\PMG$ (gray) and $\PNR$ (black). The rows refer to $\alpha = 1$ and $\alpha = 0.1$ and the target distribution is given by the posterior of model \eqref{eq:original_model}, with $f_\theta(y)$ as in \eqref{eq:normal_kernel}, $\mu_0 = 0$ and $\sigma^2 = \sigma^2_0 = 1$.
  }}
 \label{fig:posterior}
\end{figure}

\subsection{Overfitted setting}\label{sec:overfitting}

We consider now the overfitted case, previously discussed in Section \ref{sec:specificities}.
We take a one-dimensional Gaussian kernel as in \eqref{eq:normal_kernel} and take $\alpha_k = \alpha$ for all $k\in\{1,\dots,K\}$. 
In this setting, using the notation of Section \ref{sec:specificities}, \citet[Thm.1]{rousseau2011asymptotic} implies that if $\alpha > 1/2$ more than $K^*$ atoms have non-negligible mass (i.e.\, multiple atoms are associated to the same ``true'' component), while if $\alpha \leq 1/2$, then the posterior concentrates on configurations with exactly $K^*$ components, up to $n^{-1/2}$ posterior mass.
We take $K = 2$ and $K^* = 1$, with $Y_i \simiid N(y \mid 2, 1)$ and $n = 1000$. The first two columns of Figure \ref{fig:overfitted} plot the histogram of the proportion of the first component after $T=100\times n$ iterations (and thinning of size $n$) for $\alpha = 1$ (top row) and $\alpha = 0.1$ (bottom row). The two algorithms are initialized according to the ``incorrect'' scenario, i.e.\ all the observations in the first component in the first row and uniformly at random in the bottom row. The figure illustrates that only $\PNR$ is able to reach the high probability region: this means that, despite its locality, the persistence of $\PNR$ allows for significantly faster traveling across the space. On the contrary, $\PMG$ remain stuck in the initial configuration (which yields a similar likelihood) for both the scenarios. This is also confirmed by the right column, which depicts the marginal distribution of the chains: after few iterations, the distribution associated to $\PNR$ stabilizes and yields the correct behaviour.

\begin{figure}[h]
\centering
\includegraphics[width=.30\textwidth]{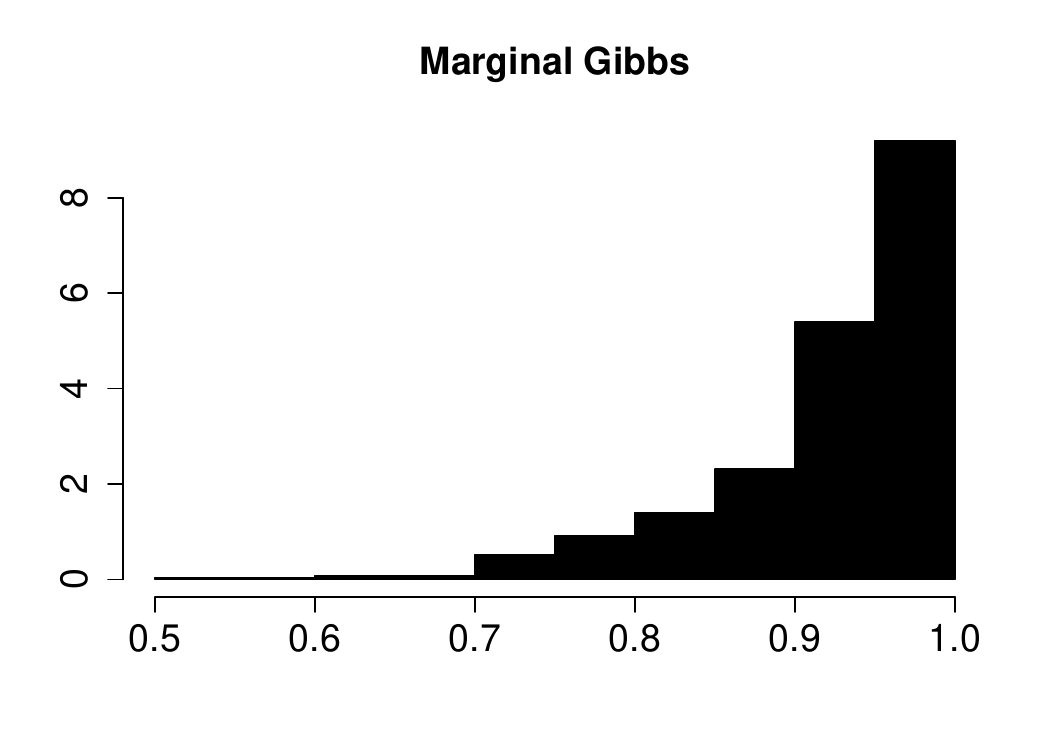} \quad
\includegraphics[width=.30\textwidth]{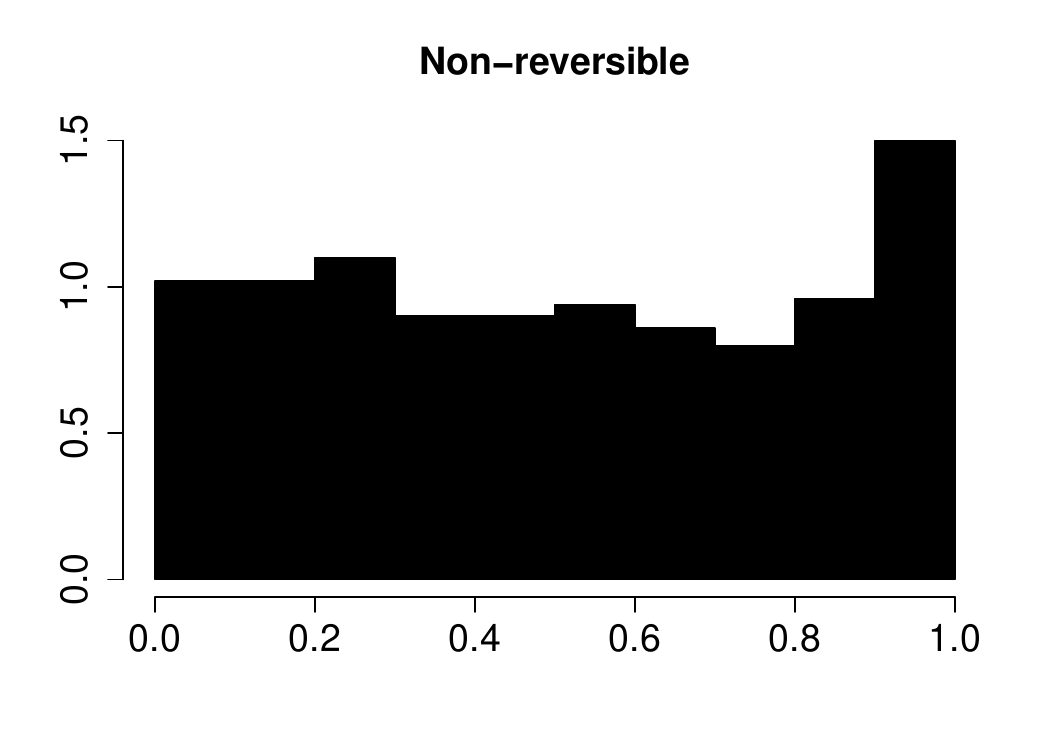} \quad
\includegraphics[width=.30\textwidth]{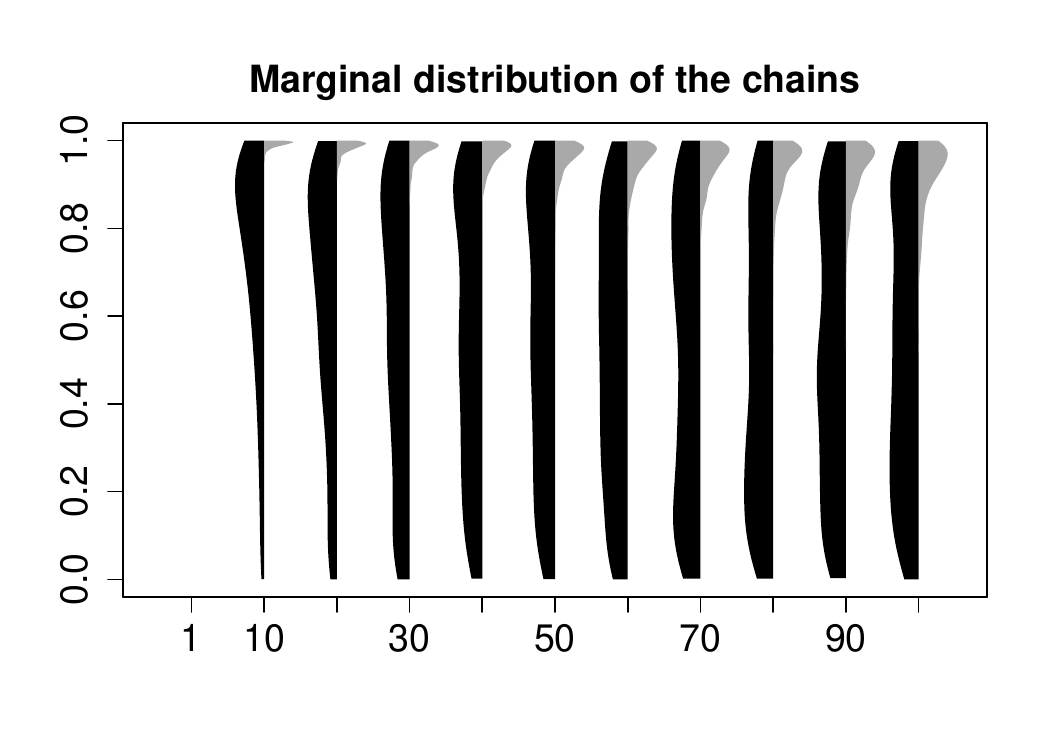}\\
\includegraphics[width=.30\textwidth]{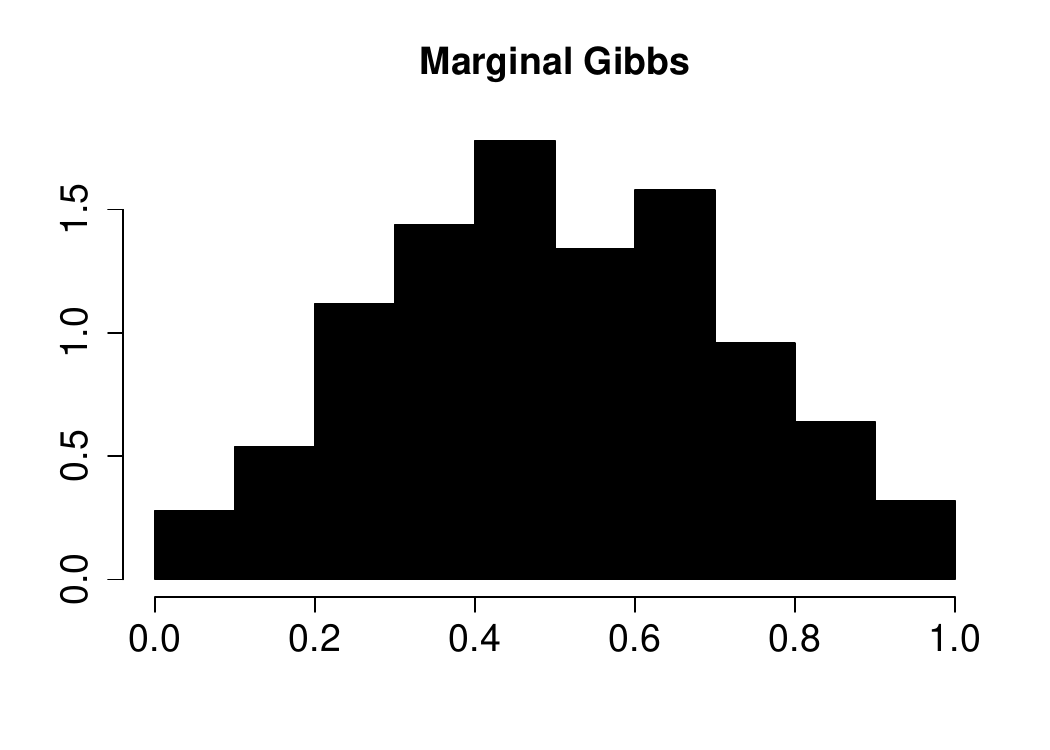} \quad
\includegraphics[width=.30\textwidth]{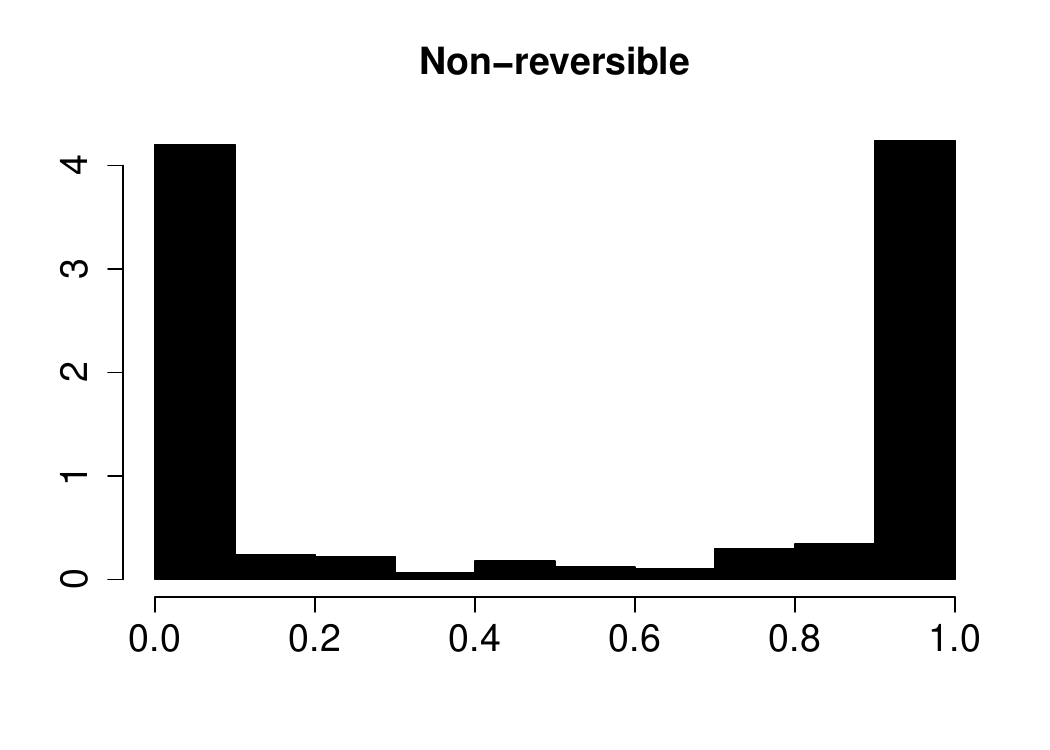} \quad
\includegraphics[width=.30\textwidth]{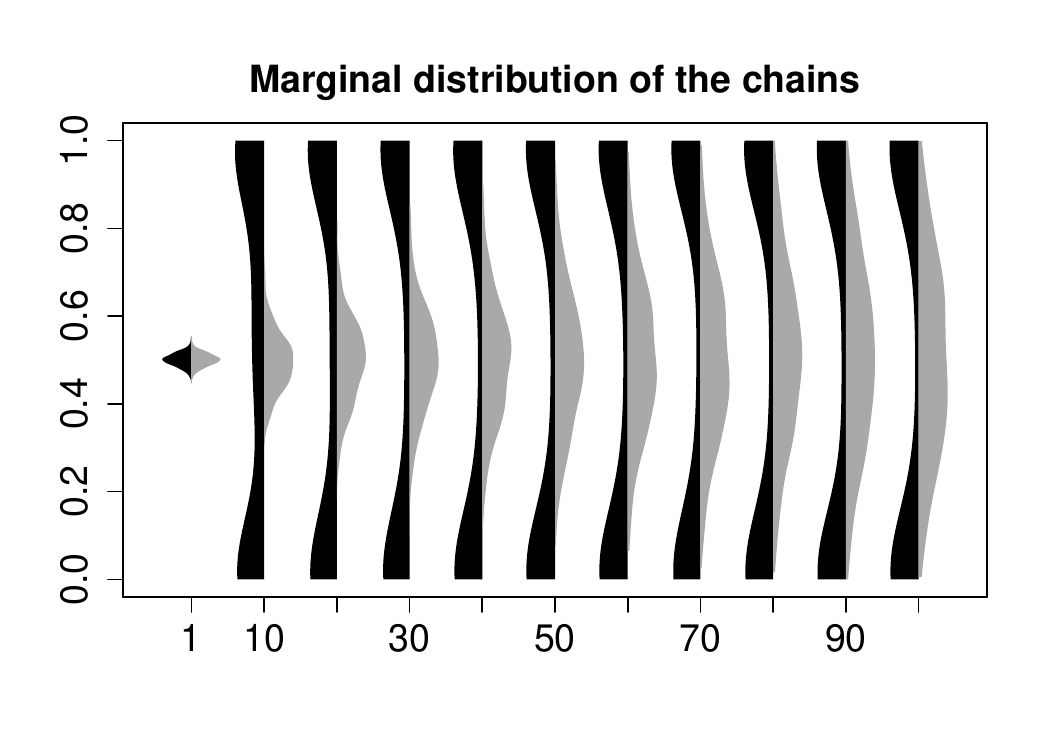}
 \caption{\scriptsize{
Left and center column: histogram of the proportion of the first component in the last of $100$ iterations (after a thinning of size $n$) over $300$ independent runs for $\PMG$ (left) and $\PNR$ (center). Right column: plot of the marginal distribution of the proportion of the first component at every $10$ iterations (after thinning) for $\PMG$ (gray) and $\PNR$ (black). First row: $\alpha = 1$ and initialization uniformly at random. Second row: $\alpha = 0.1$ and $c_i^{(0)} = 1$ for every $i$. The target distribution is given by the posterior of model \eqref{eq:original_model}, with $Y_i \simiid N(y \mid 2, 1)$ and $f_\theta(y)$ as in \eqref{eq:normal_kernel}, $\mu_0 = 0$ and $\sigma^2 = \sigma^2_0 = 1$.
  }}
 \label{fig:overfitted}
\end{figure}

\subsection{Comparison with split-and-merge type algorithms} \label{sec:simDP}

Both $\PMG$ and $\PNR$ are examples of local methods, where at most a single allocation variable per iteration is updated. In this section we make a comparison with global methods, which try to move large subsets of variables at each iteration. We consider in particular split-and-merge schemes \citep{jain2004split, dahl2022sequentially}, which at each iteration sample a pair $(i,j)$ and try to merge two clusters (if $c_i \neq c_j$) or to split one (if $c_i = c_j$).
Following the review in \cite{dahl2022sequentially}, we consider the Non Uniform Sequentially Allocated Merge and Split Sampler (NUSAMS), which uses an approximation of the posterior pairwise similarities, $\pi(c_i = c_j)$, to  sample the initial pair $(i,j)$. In particular, \citet[Section $5.1$]{dahl2022sequentially} empirically show that NUSAMS$(t,k,g)$ outperforms other popular split-and-merge schemes. Here $t$ tunes the complexity of the splitting strategy, while $k$ and $g$ are the number of (respectively) global and local updates to be alternated: following the empirical results in \cite{dahl2022sequentially}, we take $(t,k,g) = (0,5,1)$. For the local update we run $n$ iterations of either $\PMG$ (as done in the aforementioned paper, denoted simply with ``NUSAMS") or $\PNR$ (denoted with ``NUSAMS +NR").
We consider various configurations of sample size and dimension of the data, ranging respectively from $1000$ to $10000$ and from $1$ to $5$. For each configuration, using the same strategy of Section \ref{sec:posterior_sim}, we generate $500$ datasets from a Dirichlet process mixture with Gaussian kernel, i.e. model \eqref{eq:model_random} where $\rho$ is as in \eqref{eq:prior_PY} with $\delta = 0$, $\alpha = 0.1$ and
\begin{equation} \label{eq:generative_model}
f_\theta(y) = \mathrm{N}(y \mid \mu,I_p), \quad p_0(\theta) = \mathrm{N}\left(\theta \mid 0,10I_p\right).
\end{equation}
For each dataset $\PMG$, $\PNR$, $\text{NUSAMS}(0,5,1)$ and $\text{NUSAMS}(0,5,1)+\text{NR}$ are run until $10000\times n$ local updates are performed: each allocation variable is initialized uniformly at random over $10$ possible clusters. Again following \cite{dahl2022sequentially}, we focus on the entropy $E(c) = -\sum_{k = 1}^{K(c)}\frac{n_k(c)}{n}\log \left(\frac{n_k(c)}{n} \right)$
and measure the convergence time as $t_{conv}(\epsilon;P) = \min\left\{ t \geq 0 \,:\,\| P^t_E\ - \nu_E \|_{TV} \leq \epsilon\right\}$, where $\nu_E$ and $P^t_E$ are the distribution of $E(c)$ under respectively the prior $\nu$ and the Markov chain with kernel $P$ after $t$ iterations: $\|\cdot\|_{TV}$ denotes the total variation distance and it is approximated with numerical integration, based on the realizations of the chains and Monte Carlo samples from the prior. We select the threshold $\epsilon = 0.1$,  to account for Monte Carlo and integration error.
To fairly compare the different procedures, Figure \ref{fig:DPsim} (top row) plots the average number of predictive density evaluations required to reach $t_{\mathrm{conv}}(\epsilon;P)$, for each kernel: this measure is implementation independent and allows to compare algorithms with different cost per iteration.
After convergence, each chain is run for additional $1000\times n$ local updates and the effective sample size (ESS) of the entropy is computed using the \texttt{R} \texttt{coda} package \citep{Plummer2006}: in Figure \ref{fig:DPsim} (bottom row) the average ESS, rescaled by the number of predictive density evaluations, is illustrated. 
\begin{figure}[h]
\centering
\includegraphics[
    width=.9\textwidth,
    height=0.7\textheight,
    keepaspectratio
]{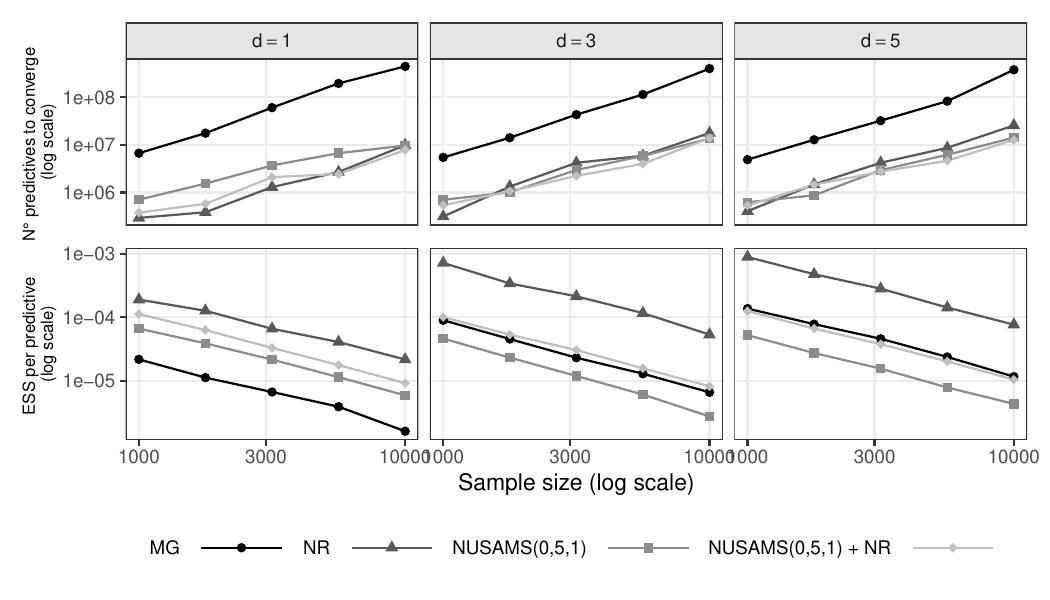}
\caption{\scriptsize{
Convergence diagnostics for the Dirichlet Process mixtures in \eqref{eq:generative_model}. The upper panel reports, for each transition kernel $P$, the average number of predictive density evaluations required to reach the estimated convergence time $t_{conv}(0.1;P)$. The lower panel reports, for each $P$, the average effective sample size (ESS), averaged over $500$ simulated datasets, computed from $1000\times n$ post-convergence local updates, normalized by the corresponding number of predictive density evaluations. Both axes are shown on a logarithmic scale.
}}
\label{fig:DPsim}
\end{figure}
The plots show that our non-reversible local method dominates $\PMG$ and is comparable to split-and-merge global methods in terms of convergence speed, while achieving a much lager ESS in all scenarios. Moreover, incorporating non-reversible local updates into global methods further improves their performance.
We also note that Figure \ref{fig:DPsim} is designed to be favourable to global schemes, since we do not account for the cost of approximating the posterior similiarities in NUSAMS (which can be of order $\sO(n^2)$, see \citet[Section 3.5]{dahl2022sequentially} for details).

\section{Application to astronomical data}\label{sec:data}
We consider now a real data application based on the COMBO-17 galaxy dataset from the Chandra Deep Field South (CDFS) survey field, obtained from the Pennsylvania State Center for Astrostatistics\footnote{\url{https://sites.psu.edu/astrostatistics/datasets-combo-17}}, which contains measurements on $n = 3462$ galaxies. The scientific goal is to study galaxy recession velocities through the variable \texttt{Mcz}, a photometric estimate of the cosmological \textit{redshift}, an indicator of when the light source is moving away from the observer \citep{geller1989center}. 

We model the data through a Pitman-Yor mixture with Gaussian kernel. More formally, if $y_1,\dotsc,y_n$ are the galaxies redshifts, we consider model \eqref{eq:model_random} with $\rho$ in \eqref{eq:prior_PY}  and
\begin{equation} \label{eq:model_PY}
f_\theta(y) = \mathrm{N}(y \mid \mu,\sigma^2), \quad p_0(\theta) = \text{N}\left(\mu|\mu_0,\frac{\sigma^2}{\lambda_0}\right) \times \text{IG}\big(\sigma^2| \alpha_0, \beta_0\big),
\end{equation}
where $\text{IG}$ denotes the density of the Inverse-Gamma distribution. We consider a fully Bayesian approach and place suitable priors on all the hyperparameters: see Section \ref{sec:more_data} in the supplement for the full model and additional details. 

Analogously to Section \ref{sec:simDP}, we run $500$ independent chains with kernel $\PMG$, $\PNR$, $\text{NUSAMS}(0,5,1)$ and $\text{NUSAMS}(0,5,1)+\text{NR}$ for $100000 \times n$ local updates, where the hyperparameters are updated every $n$ iterations through univariate slice sampling. Convergence time is again estimated based on the entropy, where a very long run of $\PNR$ is used as a proxy of the exact posterior distribution. After convergence, each chain is run for an additional $1000\times n$ local updates to estimate the ESS. Table \ref{tab:results} provides the number of predictive evaluations to reach convergence (first column) and the ESS per predictive evaluation (second column) for the four samplers, relative to $\PMG$. Here $\PNR$ yields the faster convergence time and the highest relative ESS, constituting a significant improvement over the reversible counterparts. Notice that the split-and-merge algorithm alone provides only a modest improvement in efficiency over $\PMG$; however, when its local updates are replaced with $\PNR$, the resulting sampler becomes highly competitive. 

\begin{table}[ht]
\centering
\caption{Monte Carlo estimates (with Monte Carlo standard errors) of efficiency metrics relative to MG, for model \eqref{eq:model_PY}.}
\label{tab:results}
\begin{tabular}{|l|c|c|}
\hline
  & N° predictives to converge & ESS per predictive\\
    & (relative to MG) & (relative to MG)\\
\hline
NR & \textbf{0.166} (0.011) & \textbf{5.317} (0.111)\\
\hline
NUSAMS(0,5,1) & 0.778 (0.062) & 2.175 (0.030)\\
\hline
NUSAMS(0,5,1)+NR & 0.280 (0.019) & 5.301 (0.097)\\
\hline
\end{tabular}
\end{table}

\section{Discussion}\label{sec:discussion}
In this work we introduced a novel, simple and effective non-reversible MCMC sampler for mixture models, which enjoys three favourable features: (i) it is a simple modification of the original marginal scheme of Algorithm \ref{alg:PMG}, (ii) its performance cannot be worse than the reversible chain by more than a factor of four (Theorem \ref{thm:asymp_variances}), (iii) it can drastically speed-up convergence in various scenarios of interest.

Both the theory and methodology presented in this work could be extended in many interesting directions.
First, in the current formulation of Algorithm \ref{alg:PNR}, the pair of clusters to update and the observation to move are selected with probabilities that do not depend on the actual observations within the clusters (except for their sizes). A natural extension would be to consider informed proposal distributions, as in e.g.\ \cite{zanella2020informed, power2019accelerated, gagnon2024theoretical}: we expect this to lead to a potentially large decrease of the number of iterations needed for mixing, but with an additional cost per iteration. We leave the discussion and exploration of this tradeoff to future work. 
From a methodological point of view, in this paper we stick to the case of mixture models, which are an instance of the broader framework of latent class models \citep{goodman1974exploratory} and would be interesting to explore the effectiveness of the methodology developed here in such broader settings. These could include, for example, hierarchical mixtures, where more clustering layers can be updated using non-reversible dynamics with shared or independent directional variables, or hidden Markov models, where the predictive density for a given unit requires to account also for the likelihood terms of the allocation variables that shares a Markovian dependence with it. It would be of interest to derive sharper bounds that depend on specific features of the model under consideration.

\bibliographystyle{abbrvnat}
\bibliography{reference}

\begin{appendix}

\section{Sampling from $p_c(k,k')$}\label{sec:sampling_weights}

\begin{algorithm}[htbp]
\begin{algorithmic}
\State Sample $k_1$ from $\{1, \dots, K\}$ with probabilities $(n_{1}(c)/n,\dots,n_K(c)/n)$
\State Sample $k_2$ uniformly at random from $\{1,\dots,K\}\backslash \{k_1\}$
\State Set $k = \min\{k_1, k_2\}$ and $k' = \max\{k_1, k_2\}$
\end{algorithmic}
\caption{Sampling $(k, k') \sim p_c$ defined in \eqref{eq:prob_components}
\label{alg:sampling_clusters}}
\end{algorithm}

Sampling from $p_c$ in \eqref{eq:prob_components} can be performed efficiently using Algorithm \ref{alg:sampling_clusters}, where one cluster is selected with probability proportional to its size and the other uniformly at random from the remaining ones: in terms of implementation it requires only two samples from the uniform distribution. Validity is proved in the next lemma.
\begin{lemma}\label{lm:sampling_clusters}
For each $c \in [K]^n$, Algorithm \ref{alg:sampling_clusters} produces a sample $(k, k')$ from the probability distribution $p_c$ defined in \eqref{eq:prob_components}.
\end{lemma}
\begin{proof}
    Fix $(k, k') \in \mathcal{K}$ and let $(k_1, k_2)$ be the pair sampled in the first two lines of Algorithm \ref{alg:sampling_clusters}. Then a draw from the latter will have $(k, k')$ as realization if and only if $(k_1, k_2) = (k, k')$ or $(k_1, k_2) = (k', k)$. By construction
    \[
    \mathbb{P}(k_1 = k, k_2 = k') = \frac{n_k(c)}{(K-1)n}, \quad \mathbb{P}(k_1 = k', k_2 = k) = \frac{n_{k'}(c)}{(K-1)n},
    \]
    and thus $p_c(k, k') = \mathbb{P}(k_1 = k, k_2 = k') + \mathbb{P}(k_1 = k', k_2 = k)$, as desired.
\end{proof}

\section{General results about lifting and mixtures}\label{sec:general_results}
In order to prove results below, especially Theorems \ref{thm:asymp_variances} and \ref{thm:random_partitions},  we first need to generalize some classical results about lifting of Markov chains (see e.g.\ \citealp{chen2013accelerating, bierkens2016non, andrieu2021peskun}) to our mixture case, which can be seen as a way to construct `multi-dimensional' lifted chains.
We will make use of the following classical lemma, which for example follows by results in \cite{chen2013accelerating} as detailed below.
\begin{lemma}\label{lemma:symmetrization}
Let $\mu$ be a probability distribution on a finite space $\sX$, $P$ a $\mu$-invariant and irreducible Markov transition matrix, $P^*$ the $\mu$-adjoint of $P$ and $K=(P+P^*)/2$. Then 
$\Var(g, P)\leq \Var(g, K)$ for all $g:\sX\to\R$.
\end{lemma}
\begin{proof}
Consider the decomposition
$P = K + Q$, $Q = \frac{1}{2}P-\frac{1}{2}P^*$.
By construction, $K$ is a $\mu$-reversible transition matrix. Moreover, by definition of adjoint we have that $Q$ is antisymmetric with respect to $\mu$, 
which means that 
$\mu(x)Q(x,y) = -\mu(y)Q(y,x)$ for all $x,y\in\sX$. Finally, for every $x\in \sX$ we have that
\[
\sum_{y \in \sX}Q(x, y) = \frac{1}{2}\sum_{y \in \sX}P(x, y)-\frac{1}{2}\sum_{y \in \sX}P^*(x, y) = 0
\]
and thus each row of $Q$ sums up to zero. Therefore by Lemma $2$ in \cite{chen2013accelerating} we have that $\Var(g, P) = \Var(g, K + Q) \leq \Var(g, K)$ for all $g:\sX\to\R$.
\end{proof}

\subsection{Result with general notation}\label{sec:as_var_general}
Let $\pi$ be a probability distribution on a finite space $\C$. 
Let $D\in\N$, $(K_{d,+1})_{d\in[D]}$ and $(K_{d,-1})_{d\in[D]}$ be Markov transition kernels on $\C$ and $\{p_c(d)\}_d$ such that $p_c(d) \geq 0$, $\sum_{d=1}^D p_c(d) = 1$ for all $c \in \mathcal{C}$ and
\begin{equation}\label{eq:single_coordinate_}
\pi(c)p_c(d)K_{d,+}(c,c') = \pi(c')p_{c'}(d)K_{d,-}(c',c) \quad \text{ for all } c \neq c' \text{ and } d\in[D]. 
\end{equation} 
Define the Markov transition kernel on $\C$ 
\begin{equation}\label{eq:symm_mixture_operator}
K_{R}(c, c') = \sum_{d = 1}^{D} p_c(d)
K_d\left(c, c'\right)\,,
\end{equation}
where $K_d=(K_{d,+}+K_{d,-})/2$. 
Define the Markov transition kernel on $\mathcal{C} \times \{-1,+1\}^D$
\begin{equation}\label{eq:lifted_mixture_operator}
K_{NR}((c,v), (c',v')) = \left( F \bar{K}_{NR} F \right)((c,v), (c',v')) ,
\end{equation}
where 
\begin{equation} \label{eq:flipping_operator}
F((c,v), (c',v'))  = F^{(c)}(v,v')\mathbbm{1}(c' = c), 
\end{equation}
with $F^{(c)}$ being a $\text{Unif}(\{-1,+1\}^D)$-reversible kernel for every $c$, and
\begin{equation} \label{eq:periodic_kernel}
\bar{K}_{NR}((c,v),(c',v')) = \sum_{d=1}^D p_c(d)K_{\text{lift},d}((c,v),(c',v')),
\end{equation}
with
\begin{equation} \label{eq:lifted_kernels}
K_{\text{lift},d}((c,v),(c',v')) = K_{d,v_d}(c,c')Q_{d,c,c'}(v,v')
\end{equation}
and 
\begin{equation} \label{eq:lifted_mixture_operator_}
Q_{d,c,c'}(v,v') = \mathbbm{1}(v_{-d} = v_{-d}')\left[ \mathbbm{1}(c \neq c')\mathbbm{1}(v_d = v_d') + \mathbbm{1}(c  = c')\mathbbm{1}(v_d = -v_{d}') \right]\,.
\end{equation}

\begin{lemma}\label{lemma:lifting}
Under \eqref{eq:single_coordinate_}-\eqref{eq:lifted_mixture_operator_}, we have that
\begin{itemize}
    \item[(a)] $K_R$ is $\pi$-reversible.
    \item[(b)] $K_{NR}$ is $\tilde{\pi}$-invariant, with $\tilde{\pi}(c,v)=\pi(c)2^{-D}$.
    \item[(c)] If $K_R$ is irreducible, then
    $\Var(\tilde{g}, K_{NR})\leq \Var(g, K_{R})$ for all $g:  \mathcal{C} \times \{-1,+1\}^D \to\R$ and $\tilde{g}:\C\to\R$ such that $g(c,v)=\tilde{g}(c)$ for all $(c,v)\in \mathcal{C} \times \{-1,+1\}^D$.
\end{itemize}
\end{lemma}

\begin{proof}
Consider first part (a). 
By \eqref{eq:single_coordinate_}, for every $c\neq c'$ we have
\[
\begin{aligned}
2\pi(c)p_c(d)K_d(c, c') &= \pi(c)p_c(d)K_{d, +}(c, c') + \pi(c)p_c(d)K_{d, -}(c, c')\\
&= \pi(c')p_{c'}(d)K_{d, -}(c', c) + \pi(c')p_{c'}(d)K_{d, +}(c', c) \\
&= 2\pi(c')p_{c'}(d)K_{d}(c', c).
\end{aligned}
\]
and thus by \eqref{eq:symm_mixture_operator} we have
\[
\begin{aligned}
\pi(c)K_R(c, c') &= \sum_{d = 1}^{D} \pi(c)p_c(d)K_d(c, c')\\
& = \sum_{d' =1}^{D} \pi(c')p_{c'}(d)K_{d}(c', c) = \pi(c')K_{R}(c', c)\,,
\end{aligned}
\]
meaning that $K_R$ is $\pi$-reversible.

Consider now point (b). Let $(c',v') \in \mathcal{C}\times \{-1,+1\}^D$.
If $v = v'$, by \eqref{eq:lifted_mixture_operator}, $Q_{d,c,c'}(v,v') = \mathbbm{1}(c\neq c')$ and \eqref{eq:single_coordinate_} we have
\begin{align*}
\sum_{c\in\C}\tilde{\pi}(c, v)p_{c}(d)K_{\text{lift},d}((c,v), (c',v')) &=
\sum_{c \in\C}\pi(c)p_{c}(d)K_{d, v_d}(c, c')Q_{d,c,c'}(v,v')2^{-D}
\\
&= 
\sum_{c \neq c'}\pi(c)p_{c}(d)K_{d, v'_{d}}(c, c')2^{-D}\\
&= \sum_{c \neq c'}\pi(c')p_{c'}(d)K_{d, -v_{d}'}(c', c)2^{-D} \\
&= \tilde{\pi}(c', v')p_{c'}(d)\left[1-K_{d, -v_{d}'}(c', c') \right].
\end{align*}
Similarly, if $v_{-d} = v'_{-d}$ and $v_d = -v'_d$, by $Q_{d,c,c'}(v,v')=
\mathbbm{1}(c= c')$
we have that
\begin{align*}
\sum_{c\in\C}\tilde{\pi}(c, v)p_{c}(d)K_{\text{lift},d}((c,v), (c',v')) &= \sum_{c \in\C}\pi(c)p_{c}(d)K_{d, v_d}(c, c')Q_{d,c,c'}(v,v')2^{-D(c)}
\\
&= \pi(c')p_{c'}(d)K_{d, -v'_d}(c', c')2^{-D(c')}\\
&= \tilde{\pi}(c', v')p_{c'}(d)K_{d, -v_{d}'}(c', c').
\end{align*}
Summing the two expressions above, and using the fact that 
 $K_{\text{lift},d}((c,v), (c',v'))=0$ if $v_{-d} \neq v_{-d}'$, we have
\[
\sum_{c, v}\tilde{\pi}(c, v)p_{c}(d)K_{\text{lift},d}((c,v), (c',v')) = p_{c'}(d)\tilde{\pi}(c', v').
\]
Finally, we have
\[
\begin{aligned}
\sum_{c, v}\tilde{\pi}(c, v)\bar{K}_{NR}((c,v), (c',v')) & = \sum_{c,v} \tilde{\pi}(c, v) \sum_{d = 1}^{D} p_c(d) K_{\text{lift},d} ((c,v), (c',v')) \\
&=  \sum_{d = 1}^{D} \sum_{c,v} \tilde\pi(c, v) p_c(d) K_{\text{lift},d} ((c,v), (c',v')) \\
&=  \sum_{d = 1}^{D} p_{c'}(d) \tilde{\pi}(c',v') = \tilde{\pi}(c',v'),
\end{aligned}
\]
and therefore $\bar{K}_{NR}$ is $\tilde \pi$-invariant. Since $F$ is trivially $\tilde{\pi}$-invariant and composition of invariant kernels remains invariant, then $K_{NR} = F \bar{K}_{NR}F$ is $\tilde{\pi}$-invariant. 

Consider now point (c). Let $\bar{K}_R=(K_{NR}+K^*_{NR})/2$, where $K^*_{NR}$ is the $\tilde{\pi}$-adjoint of $K_{NR}$.  
Since $F^* = F$ by \eqref{eq:flipping_operator}, we have that $K_{NR}^* = \left(F\bar{K}_{NR} F\right)^* = F^* \bar{K}_{NR}^* F^* = F \bar{K}_{NR}^* F$. Moreover
\begin{align*}
    \bar{K}_{NR}^*((c,v),(c',v')) &= \frac{\tilde \pi(c',v')}{\tilde \pi(c,v)} \bar{K}_{NR}((c',v'),(c,v)) \\
    &= \sum_{d = 1}^D p_{c'}(d)\frac{\tilde \pi(c',v')}{\tilde \pi(c,v)} K_{\text{lift},d}((c',v'),(c,v))\\
    &= \sum_{d = 1}^D p_{c'}(d) K_{\text{lift},d}^*((c,v),(c',v')) 
\end{align*}
and thus $\bar{K}_R = F \hat{K}_R F$, where
\[
\hat{K}_R((c,v), (c',v')) = \frac{1}{2}\sum_{d = 1}^{D}p_c(d) K_{\text{lift},d} ((c,v), (c',v')) +\frac{1}{2}\sum_{d = 1}^{D} p_{c'}(d) K^*_{\text{lift},d}((c,v), (c',v')).
\]
By \eqref{eq:single_coordinate_} we have that for $c' \neq c$
\begin{align*}
p_{c'}(d)K^*_{\text{lift}, d}((c,v), (c',v')) &= 
p_{c'}(d)\frac{\tilde\pi(c',v')}{\tilde\pi(c,v)}
K_{\text{lift}, d}((c',v'), (c,v)) 
\\
&= 
p_{c'}(d)\frac{\pi(c')}{\pi(c)}
K_{d,v'_{d}}\left(c', c\right)Q_{d,c',c}(v',v) 
\\
& = p_{c}(d)K_{d,-v_{d}'}\left(c, c'\right)Q_{d,c',c}(v',v) \\
&= p_{c}(d)K_{d,-v_d}\left(c, c'\right)Q_{d,c,c'}(v,v'),
\end{align*}
where we used the definition of $Q_{d,c,c'}(v,v')$.
For $c' = c$ we have that
\begin{align*}
K^*_{\text{lift}, d}((c,v), (c,v')) &= K_{\text{lift}, d}((c,v'), (c,v)) 
\\
&= K_{d,v_{d}'}\left(c, c\right)Q_{d,c,c}(v',v)
= K_{d,-v_d}\left(c, c\right)Q_{d,c,c'}(v,v'),
\end{align*}
where we used that $Q_{d',c,c}(v',v)>0$ implies $v'_d=-v_d$.
Thus 
\[
\begin{aligned}
\frac{1}{2} p_{c}(d) K_{\text{lift},d}\left(c, c'\right) &+ \frac{1}{2}p_{c'}(d) K_{\text{lift},d}^* ((c,v),(c',v')) \\
&= \frac{1}{2}p_{c}(d) K_{d,v_d}(c,c') Q_{d,c,c'}(v,v') + \frac{1}{2}p_{c}(d) K_{d,-v_d}(c,c') Q_{d,c,c'}(v,v') \\
&= p_{c}(d) \left[ \frac{1}{2} K_{d,+}(c,c') + \frac{1}{2}K_{d,-}(c,c') \right]Q_{d,c,c'}(v,v') \\
&= p_c(d)K_d(c,c')Q_{d,c,c'}(v,v'),
\end{aligned}
\]
which implies
\[
\hat{K}_R((c,v),(c',v')) = \sum_{d = 1}^D p_{c}(d) K_d(c,c')Q_{d,c,c'}(v,v').
\]
Let now $g: \mathcal{C} \times \{-1,+1\}^D \to\R$ and $\tilde{g}:\C\to\R$ such that $g(c,v)=\tilde{g}(c)$. Then, since $F$ leaves the first coordinate invariate  by \eqref{eq:flipping_operator} and $K_d$ does not depend on $v$ for every $d = 1,\dotsc,D$, we have that
\[
\begin{aligned}
    \bar{K}_{R} g(c,v) &= \left( F \hat{K}_{R} F g\right) (c, v) = \sum_{c', v'} \left( F \hat{K}_{R} F \right) ((c, v),(c',v'))\tilde{g}(c')\\
    & = \sum_{c'}\sum_{d = 1}^D p_c(d)K_d(c,c')\tilde g(c') = K_R \tilde g(c).
\end{aligned}
\]
By simple induction then $\bar{K}^t_{R}g(c, v) = K_R^t\tilde{g}(c)$ for every $t$. It thus follows that 
$\Var(g, \bar{K}_{R})= \Var(\tilde{g}, K_{R})$. Point (c) then follows from $\Var(g, K_{NR})\leq \Var(g, \bar{K}_{R})$, which is a consequence of Lemma \ref{lemma:symmetrization}.
\end{proof}

\section{Proofs}

\subsection{Proof of Lemma \ref{lemma:erg_Prev}}
\begin{proof}
Reversibility follows by Lemma \ref{lemma:lifting} (point (a)), with $(k, k') \in \mathcal{K}$ in place of $d \in [D]$, $p_c(k, k')$ defined as in \eqref{eq:prob_components} and
\[
P_{(k,k'), v}(c, c') = \tilde{P}_{k, k'}\left((c,v\bm{1}), c' \times \{-1, +1\}^{K(K-1)/2}\right), \quad v \in \{+, -\}
\]
where $\bm{1} = (1, \dots, 1)$, so that $K_R = \Prev$.
The only delicate condition to verify is \eqref{eq:single_coordinate_}, which follows by
\[
\pi(c)P_{(k,k'), +}(c, c') = \pi(c')P_{(k,k'), -}(c', c)
\]
for every $c \neq c'$ and the fact that $P_{k, k'}(c, c') > 0$ implies that $n_k(c) + n_{k'}(c) = n_k(c') + n_{k'}(c')$ and therefore $p_c(k, k') = p_{c'}(k, k')$.

Since $\pi(c) > 0$ for every $c\in[K]^n$, we have $\pi(c_i = k \mid c_{-i}) > 0$ for every $c_{-i} \in [K]^{n-1}$. Combining this with the fact that $p_c(k, k') > 0$ for every $(k, k') \in \mathcal{K}$ such that $n_k(c) + n_{k'}(c) > 0$, we get that for every pair $c \neq c' \in [K]^n$ there exists a $T\in\N$ and a sequence $c = c^{(0)}, c^{(1)}, \dots, c^{(T)} = c'$ such that $\Prev\left(c^{(t-1)}, c^{(t)}\right) > 0$ for every $t = 1, \dots, T-1$. Thus, $\Prev$ is irreducible. It is also easy to see that $\Prev$ is aperiodic. Uniform ergodicity then follows from \citet[Theorem 4.9]{levin2017markov}.

\end{proof}

\subsection{Proof of Lemma \ref{lemma:ergodicity}}
\begin{proof}
Invariance follows by Lemma \ref{lemma:lifting} (point (b)), with $(k, k') \in \mathcal{K}$ in place of $d \in [D]$, $p_c(k, k')$ defined as in \eqref{eq:prob_components} and
\[
P_{(k,k'), v}(c, c') = \tilde{P}_{k, k'}\left((c,v\bm{1}), c' \times \{-1, +1\}^{K(K-1)/2}\right), \quad v \in \{+, -\}
\]
where $\bm{1} = (1, \dots, 1)$,
so that $K_{NR} = \PNR$.
Condition \eqref{eq:single_coordinate_} is satisfied as shown in the proof of Lemma \ref{lemma:erg_Prev}.

Consider then irreducibility. For ease of notation, we denote $\sX = [K]^n\times \{-1, +1\}^{K(K-1)/2}$ and $x = (c, v) \in \sX$. 
If $\pi(c) > 0$, this implies that $\pi(c_i = k \mid c_{-i}) > 0$ for every $c_{-i} \in [K]^{n-1}$. Combining this with the fact that $p_c(k, k') > 0$ for every $(k, k') \in \mathcal{K}$ such that $n_k(c) + n_{k'}(c) > 0$, we get that for every pair $x \neq x' \in \sX$ there exists a $T\in\N$ and a sequence $x = x^{(0)}, x^{(1)}, \dots, x^{(T)} = x'$ such that $\PNR\left(x^{(t-1)}, x^{(t)}\right) > 0$ for every $t = 1, \dots, T-1$. Thus, $\PNR$ is irreducible. Moreover, if $\xi > 0$ it is immediate to deduce that $\PNR$ is aperiodic. Uniform ergodicity then follows from \citet[Theorem 4.9]{levin2017markov}.

\end{proof}

\subsection{Proof of Theorem \ref{thm:asymp_variances}}
We first need a well-known preliminary lemma, whose proof we provide for self-completeness.
\begin{lemma} \label{lemma:lazy_as_var}
Let $P$ be a $\pi$-invariant, irreducible and aperiodic Markov kernel on a finite space $\C$ and let $Q = (P + I)/2$, where $I$ denotes the identity operator. Then 
\[
\Var(g,Q) = 2\Var(g,P)+\Var_\pi(g),
\] 
for every $g : \C \to \R$.
\end{lemma}
\begin{proof}
Let $g_0 = g - \E_\pi[g]$. Then the asymptotic variance can be equivalently written (see e.g.\ \cite{andrieu2021peskun}) as
\[
\Var(g, P) = \langle g_0, \left(2(I-P)^{-1}-I\right)g_0\rangle_\pi,
\]
where $\langle h, g \rangle_\pi = \sum_{c \in \C}h(c)g(c)\pi(c)$ and the operator $(I-P)^{-1}$ is well-defined for zero mean functions since the second largest modulus eigenvalue of $P$ is strictly smaller than $1$ by construction. Then, since $(I-Q)^{-1} = 2(I-P)^{-1}$, we have that
\[
\begin{aligned}
\Var(g, Q) &= \langle g_0, \left(4(I-P)^{-1}-I\right)g_0\rangle_\pi\\
& = 2\langle g_0, \left(2(I-P)^{-1}-I\right)g_0\rangle_\pi + \Var_\pi(g) = 2\Var(g, P) + \Var_\pi(g),
\end{aligned}
\]
as desired.
\end{proof}
\begin{proof}[Proof of Theorem \ref{thm:asymp_variances}]

By point (c) of Lemma \ref{lemma:lifting}, applied with $K_R = \Prev$ and $K_{NR} = \PNR$, we have that $\Var(g, \PNR) \leq \Var(g, \Prev)$. By Lemma \ref{lemma:lazy_as_var} this implies the inequality  $\Var(g, Q_{\text{NR}}) \leq \Var(g, Q_{\text{R}})$ in \eqref{eq:asymp_variances}.

In order to prove the second inequality in \eqref{eq:asymp_variances}, it suffices to show that
\begin{equation}\label{eq:inequality_operators}
Q_R(c, c') \geq \frac{1}{4}Q_{\text{MG}}(c, c'), \quad c \neq c' \in [K]^n.
\end{equation}
The result then follows by, e.g., Theorem $2$ in \cite{zanella2020informed}. 
Fix $c$ and $c'$ such that $c = (c_{-i}, k)$ and $c' = (c_{-i}, k')$ with $i \in [n]$ and $(k, k') \in \mathcal{K}$. Indeed for every other pair $(c, c')$ such that $c \neq c'$ we have that $Q_R(c, c') = Q_{\text{MG}}(c, c') = 0$. 
By definition of $Q_{\text{MG}}$ and $Q_R$ we have
\[
Q_{\text{MG}}(c, c') = \frac{1}{nK}\pi(c_i = k'\mid c_{-i}).
\]
and
\[
Q_R(c, c') = \frac{1}{4n_k}\frac{n_k+n_{k'}}{(K-1)n}\min \left\{1, \frac{n_{k}}{n_{k'} + 1}\frac{\pi(c_i = k'\mid c_{-i})}{\pi(c_i = k\mid c_{-i})} \right\},
\]
where $n_j = n_j(c)$ for every $j \in [K]$ and $n_k\geq 1$ by definition of $c$. Thus
\[
\begin{aligned}
Q_R(c, c') &= \frac{1}{4(K-1)n}\min\left\{\frac{n_k+n_{k'}}{n_{k}},  \frac{n_{k}+n_{k'}}{n_{k'} + 1}\frac{\pi(c_i = k'\mid c_{-i})}{\pi(c_i = k\mid c_{-i})}\right\}\\
&\geq \frac{1}{4(K-1)n}\min\left\{1,  \frac{\pi(c_i = k'\mid c_{-i})}{\pi(c_i = k\mid c_{-i})}\right\}\\
&= \frac{\pi(c_i = k'\mid c_{-i})}{4(K-1)n}\min\left\{\frac{1}{\pi(c_i = k'\mid c_{-i})},  \frac{1}{\pi(c_i = k\mid c_{-i})}\right\}\\
&\geq \frac{1}{4(K-1)}\frac{\pi(c_i = k'\mid c_{-i})}{n} \\
&\geq \frac{1}{4}\frac{1}{nK} \pi(c_i = k'\mid c_{-i}) = \frac{1}{4}Q_{\text{MG}}(c, c'),
\end{aligned}
\]
which is exactly \eqref{eq:inequality_operators}.
\end{proof}

\subsection{Proof of Theorem \ref{thm:diffusion_rev}}
\begin{proof}
By \eqref{eq:marg_priori} for every $\bm{x} \in \Delta_{K-1}$ we have that, as $n\to\infty$,
\[
\begin{aligned}
\E\left[X_{t+1, k} -x_k \mid \bm{X}_t = \bm{x} \right] &= \frac{1-x_k}{n}\frac{\alpha_k + nx_k}{|\bm{\alpha}|+n-1}-\frac{x_k}{n}\frac{|\bm{\alpha}|-\alpha_k + n(1-x_k)}{|\bm{\alpha}|+n-1}\\
& = \frac{\alpha_k -|\bm{\alpha}|x_k}{n(|\bm{\alpha}|+n-1)} = \frac{2}{n^2}\left[\frac{\alpha_k}{2} -|\bm{\alpha}|\frac{x_k}{2} + o(1) \right]
\end{aligned}
\]
and
\[
\begin{aligned}
\E\left[\left(X_{t+1, k} -x_k\right)^2 \mid \bm{X}_t = \bm{x} \right] &= \frac{1-x_k}{n^2}\frac{\alpha_k + nx_k}{|\bm{\alpha}|+n-1}+\frac{x_k}{n^2}\frac{|\bm{\alpha}|-\alpha_k + n(1-x_k)}{|\bm{\alpha}|+n-1}\\
& = \frac{2}{n^2}\left[x_k(1-x_k) + o(1) \right]
\end{aligned}
\]
and
\[
\begin{aligned}
\E\left[\left(X_{t+1, k} -x_k\right)\left(X_{t+1, k'} -x_{k'}\right) \mid \bm{X}_t = \bm{x} \right] &= \frac{-x_k}{n^2}\frac{\alpha_{k'} + nx_{k'}}{|\bm{\alpha}|+n-1}+\frac{x_{k'}}{n^2}\frac{\alpha_k + nx_k}{|\bm{\alpha}|+n-1}\\
& = \frac{2}{n^2}\left[-x_kx_{k'} + o(1) \right],
\end{aligned}
\]
 and $n^2\E\left[\left(X_{t+1, k} -x_k\right)^3 \mid \bm{X}_t = \bm{x} \right] = o(1)$ for $k \neq k' \in [K]$. By a second-order Taylor expansion, this means that
\[
\sup_{\bm{x} \in \Delta_{K+1}} \, \left \lvert \frac{n^2}{2}\left(\E\left[ g(\bm{X}_{t+1})\mid \bm{X}_t = \bm{x}\right] - g(\bm{x})\right) - \A g(\bm{x})\right\rvert \to 0,
\]
as $n \to \infty$ for every $g$ twice differentiable real-valued function. The result then follows by Corollary $8.9$ in \cite[Chapter 4]{ethier2009markov}.
\end{proof}

\subsection{Proof of Theorem \ref{thm:diffusion_nonrev}}

\begin{proof}
Fix $\bm{z} = (\bm{x}, \bm{v}) \in E_M \times V$. Notice that
\[
\begin{aligned}
\E\biggl[&X^{(M)}_{t+1, k}-x_k \mid \bm{X}_t^{(M)} = \bm{x}, \bm{V}^{(M)}_t = \bm{v} \biggr] =\left(1-\frac{\xi}{n} \right)\times  \\
&\times \frac{1}{n}\left[\sum_{k' \, : \, v_{k', k} = +1}\frac{x_k + x_{k'}}{K-1}\alpha(\bm{x}, k, k') - \sum_{k' \, : \, v_{k', k} = -1}\frac{x_k + x_{k'}}{K-1}\alpha(\bm{x}, k', k)\right]+ o\left(\frac{1}{n}\right),
\end{aligned}
\]
where
\[
\begin{aligned}
\alpha(\bm{x}, k, k') &= \min \left\{1, \left(\frac{\alpha_{k} + n x_k}{nx_k + 1}\right)\left(\frac{nx_{k'}}{\alpha_{k'} + nx_{k'} - 1}\right) \right\}\\
& = 1-\frac{\beta(x_k, x_{k'}, v_{k', k})}{n}+ o\left(\frac{1}{n}\right),
\end{aligned}
\]
from which we deduce that
\begin{equation}\label{eq:nonrev_scaling1}
\E\biggl[X^{(M)}_{t+1, k}-x_k \mid \bm{X}^{(M)}_t = \bm{x}, \bm{V}^{(M)}_t = \bm{v} \biggr] = \frac{\Phi_k(\bm{z})}{n} + o\left(\frac{1}{n}\right).
\end{equation}
Similarly we get that
\begin{equation}\label{eq:nonrev_scaling2}
\E\biggl[\left(X^{(M)}_{t+1, k}-x_k\right)\left(X^{(M)}_{t+1, k'}-x_{k'}\right) \mid \bm{X}^{(M)}_t = \bm{x}, \bm{V}^{(M)}_t = \bm{v} \biggr] = o\left(\frac{1}{n}\right), 
\end{equation}
for every $(k, k') \in [K]^2$. Moreover
\begin{equation}\label{eq:nonrev_scaling3}
\begin{aligned}
\E\biggl[g&\left( \bm{x}, \bm{V}^{(M)}_{t+1}\right)- g(\bm{x}, \bm{v})\mid \bm{X}^{(M)}_t = \bm{x}, \bm{V}^{(M)}_t = \bm{v} \biggr] \\
&
\begin{aligned}
    = \sum_{k \neq k'}&\left[g(\bm{z_{(k, k')}})-g(\bm{z}) \right]\frac{x_k+x_{k'}}{2(K-1)}\left(1-\alpha(\bm{x}, k, k')\right)\left(1-\frac{\xi}{n} \right)^2\\
    &+\frac{2\xi}{n}\sum_{v \in V}\left(\frac{1}{2} \right)^{K(K-1)/2}\left[g(\bm{z}_1, v)-g(\bm{z}) \right] + o\left(\frac{1}{n} \right)
\end{aligned}\\
&
\begin{aligned}
    = \frac{\lambda(\bm{z})}{n}\sum_{k\neq k'}&q(k, k')\left[g(\bm{z_{(k, k')}})-g(\bm{z}) \right]\\
    & +\frac{2\xi}{n}\sum_{v \in V}\left(\frac{1}{2} \right)^{K(K-1)/2}\left[g(\bm{z}_1, v)-g(\bm{z}) \right] + o\left(\frac{1}{n} \right).
\end{aligned}
\end{aligned}
\end{equation}
By a Taylor expansion we have that
\[
\begin{aligned}
\E\left[ g(\bm{Z}^{(M}_{t+1})\mid \bm{Z}^{(M)}_t = \bm{z}\right] &- g(\bm{z}) = \E\left[ g(\bm{x}, \bm{V}^{(M)}_{t+1})\mid \bm{Z}^{(M)}_t = \bm{z}\right] - g(\bm{z})\\
& + \sum_{k = 1}^K\E\left[ \left(X^{(M)}_{t+1, k} - x_k \right) \mid \bm{Z}^{(M)}_t = \bm{z}\right]\frac{\partial}{\partial z_{1, k}}g(\bm{z}) + o\left(\frac{1}{n} \right),
\end{aligned}
\]
that, combined with \eqref{eq:nonrev_scaling1}, \eqref{eq:nonrev_scaling2} and \eqref{eq:nonrev_scaling3}, implies
\[
\sup_{\bm{z} \in E_M \times [K]^2} \, \left \lvert n\left(\E\left[ g(\bm{Z}^{(M)}_{t+1})\mid \bm{Z}^{(M)}_t = \bm{z}\right] - g(\bm{z})\right) - \B g(\bm{z})\right\rvert \to 0,
\]
as $n \to \infty$ for every $g \, : \, \Delta_{K-1}\times [K]^2 \, \to \, \R$ twice continuously differentiable in the first argument. The result then follows by Corollary $8.9$ in \citet[Chapter 4]{ethier2009markov}.
\end{proof}

\section{Extension to random partition models} \label{sec:infinite}

\subsection{From partitions to labels}
Let $\mu \in \sP(\Pi)$ be an arbitrary probability distribution on $\Pi$, say the posterior distribution of model \eqref{eq:model_random}. We can always represent a generic partition $\rho \in \Pi$ with a vector of labels $c \in [n]^n$. Let $\rho(c) \in \Pi$ denote the partition induced by $c \in [n]^n$ and $K(c)$ the number of its unique values. The new target distribution then reads
\begin{equation} \label{eq:target_varying}
\begin{aligned}
	&\pi(c) = \frac{\mu(\rho(c))}{[n]_{K(c)}}, & c \in [n]^n,\quad [n]_K = n(n-1)\cdots(n-K+1).
\end{aligned}
\end{equation}
Sampling $c \sim \pi$ naturally implies that $\rho(c) \sim \mu$, while the labels of the clusters in $c$ given $\rho(c)$ are chosen by sampling (without replacement) from $[n]$.
An explicit representation with labels allows us to introduce the notion of directions in analogy with Section \ref{sec:non_reversible}: in particular we consider the extended target distribution
\begin{align}\label{eq:aug_invariant_beyond}
\tilde{\pi}(c,v )&:=  \pi(c)\left(\frac{1}{2}\right)^{\frac{n(n-1)}{2}} , \qquad (c,v) \in [n]^n \times \{-1,+1\}^{\frac{n(n-1)}{2}}.
\end{align}
We stress that, in the practical implementation of the non-reversible scheme, it is not necessary to keep track of all $n(n-1)/2$ velocities at each iteration, as we discuss later.
In other words, the redundant representation with $c\in[n]^n$ is only used for conceptual cleanness in the presentation.
For future reference let
\begin{equation}\label{eq:defin_directions}
\mathcal{K} = \{(k,k') \in [n]^2 : k < k'\}, \qquad \hat{\sK} = \sK \cup \{\emptyset\}.
\end{equation}

\subsection{Definition of $\PMG$, $\Prev$ and $\PNR$}

Analogously to \eqref{eq:marginal}, $\PMG$ is defined as a the random scan Gibbs sampler
\begin{equation}\label{eq:marginal_inf}
\begin{aligned}
    \PMG(c, c') &= \frac{1}{n}\sum_{i = 1}^n\PMGi(c, c'), 
    &c,c'\in [n]^n
\end{aligned}
\end{equation}
with $P_{\text{MG},i}(c,c') = \mathbbm{1}(c_{-i} = c_{-i}')\pi(c'_i\,\mid\,c'_{-i})$. Its pseudocode is then the same of Algorithm \ref{alg:PMG}. It is easy to see that if $\{c^{(t)}\}_{t \in \N}$ is a Markov chain with kernel $\PMG$, then $\{\rho^{(t)}\}_{t \in \N}$, with $\rho^{(t)} = \rho\left(c^{(t)} \right)$, is itself a Markov chain and it is $\mu$-invariant. In fact, if $i \in [n]$ is selected then
\[
\P\left(\rho^{(t)} = \rho \mid \rho^{(t-1)},i\right) \propto \mu(\rho)\mathbbm{1}\left(\rho_{-i} = \rho^{(t-1)}_{-i}\right).
\]
This is exactly the usual definition of marginal Gibbs sampler on the space of partitions, see e.g.\ \cite{dahl2022sequentially}.

Analogously to \eqref{eq:P_rev}, $\Prev$ is defined as the mixture
\begin{equation} \label{eq:directional_inf_rev}
\Prev(c,c')= \sum_{x \in \hat{\sK}} p_c(x)P_{x}(c,c'), \quad c,c' \in [n]^n,
\end{equation}
with $\hat{\sK}$ as in \eqref{eq:defin_directions} and where
\begin{equation}\label{eq:proposal}
p_c(x) = \begin{cases} 
\frac{1}{K(c) + 1} \frac{n_k(c)}{n} \frac{1}{n - K(c)} & \text{ if } x = (k,k') \text{ and } n_{k}(c) > 0, n_{k'}(c) = 0 \\
\frac{1}{K(c) + 1} \frac{n_{k'}(c)}{n} \frac{1}{n - K(c)} & \text{ if } x = (k,k') \text{ and } n_{k}(c) = 0, n_{k'}(c) > 0 \\
\frac{1}{K(c)} \frac{n_k(c) + n_{k'}(c)}{n} & \text{ if } x = (k,k') \text{ and } n_{k}(c) > 0, n_{k'}(c) > 0 \\
\frac{1 + K(c)\mathbbm{1}(K(c) = n)}{K(c)(K(c) + 1)} & \text{ if } x = \emptyset
\end{cases}
\end{equation}
For $x = (k, k')$ with $\min\{n_{k}(c),n_{k'}(c)\} > 0$, $P_x$ corresponds to the $\pi$-invariant Metropolis-Hastings kernel $P_{k, k'}$ that proposes to move a uniformly drawn point from cluster $k$ to cluster $k'$ or vice-versa with probability $1/2$. Instead, when either $n_{k}(c) = 0$ or $n_{k'}(c) = 0$, $P_{k,k'}$ is the $\pi$-invariant Metropolis-Hastings kernel that proposes to move a label in a new cluster labeled $k$ or $k'$, respectively. Finally, $x = \emptyset$ implies that no change is attempted (more formally $P_\emptyset = I$, with $I$ being the identity operator). The latter term is necessary to preserve invariance with respect to $\pi$ with a varying number of clusters. The pseudocode for $\Prev$ is given in Algorithm \ref{alg:Prev_inf}. 
\begin{algorithm}[htbp]
\begin{algorithmic}
\State Sample $x \sim p_c$ using Algorithm \ref{alg:Pkk_inf}.
    \State Set $c'=c$
    \If{$x \neq \emptyset$}    
    \State Set $(k_-,k_+)=(k,k')$ with probability $1/2$ and $(k_-,k_+)=(k',k)$ otherwise.
    \If{$n_{k_-}(c)> 0$}
    \State Sample $i \sim \text{Unif}\left(\left\{i' \in [n]  \, \mid \, c_{i'} = k_-\right\}\right)$
    \State Set $c' = (c_{-i}, k_{+})$ with probability
    \begin{equation} \label{eq:acceptance_inf}
\alpha(c,i,k_-,k_+) = \min\left\{ 1, \left(\frac{n_{k_-}(c)}{n_{k_+}(c) + 1}\right) \frac{\mu\left(\rho(c')\right)}{\mu(\rho(c))} \right\}.
\end{equation}
    \EndIf
    \EndIf
\end{algorithmic}
\caption{
Generating a sample $c'\sim \Prev(c,\cdot)$  
\label{alg:Prev_inf}}
\end{algorithm}
As established by the following lemma, sampling from \eqref{eq:proposal} can be performed efficiently using Algorithm \ref{alg:Pkk_inf}, requiring at most four draws from the uniform distribution.

\begin{algorithm}[htbp]
\begin{algorithmic}
\State Set $u = \{c_1^*, \dots, c_{K(c)}^*\}$ the vector of unique values in $c$.
\State With probability $p = \frac{1 + K(c)\mathbbm{1}(K(c) = n)}{K(c)\{K(c)+1\}}$ set $x = \emptyset$
\If{$x \neq \emptyset$}
    \State Sample $i \sim \text{Unif}([n])$ and set $k_1 = c_i$
	\State Set $S = u\setminus k_1$
	\If{$K(c) < n$}
    \State With probability $\frac{1}{(1-p)(K(c)+1)}$ set $S = [n] \setminus u$
    \EndIf
	\State Sample $k_2$ uniformly from $S$
	\State Set $x = (\min(k_1,k_2), \max(k_1,k_2))$
\EndIf
\end{algorithmic}
\caption{Sampling $x \sim p_c$ defined in \eqref{eq:proposal}}
\label{alg:Pkk_inf}
\end{algorithm}

\begin{lemma} \label{lemma:Pkk_inf}
For each $c \in [n]^n$, Algorithm \ref{alg:Pkk_inf} produces a sample $(k, k')$ from the probability distribution $p_c$ defined in \eqref{eq:proposal}.
\end{lemma}
\begin{proof}
By construction, Algorithm~\ref{alg:Pkk_inf} returns $\emptyset$ with probability $p=p_c(\emptyset)$. 
As regards the other weights, consider first the case $K(c)<n$. The probability that the algorithm returns $(k,k')$ with $n_{k}(c) > 0$ and $n_{k'}(c) = 0$ is
\[
(1-p)\frac{n_{k}(c)}{n}\frac{1}{(1-p){(K(c)+1)}}\frac{1}{n-K(c)}
=\frac{n_{k}(c)}{n{K(c)+1}}\frac{1}{n-K(c)} = p_c(k,k'),
\]
and a similar argument holds for $(k,k')$ with $n_{k}(c) = 0$ and $n_{k'}(c) > 0$. 
Now consider $k<k'$ with $n_{k}(c) > 0$ and $n_{k'}(c) > 0$. The algorithm selects this pair with probability
\begin{align*}
(1-p)
&\left[1-\frac{1}{(1-p)(K(c)+1)}\right]
\frac{n_{k}(c)+n_{k'}(c)}{n(K(c)-1)}\\
& = \left[\frac{K(c)-p(K(c)+1)}{K(c)+1}\right]
\frac{n_{k}(c)+n_{k'}(c)}{n(K(c)-1)} = \left[\frac{K(c)-1/K(c)}{K(c)+1}\right]
\frac{n_{k}(c)+n_{k'}(c)}{n(K(c)-1)}\\
&=
\frac{1}{K(c)}\frac{n_{k}(c)+n_{k'}(c)}{n} = p_c(k,k').
\end{align*}

Finally, suppose $K(c)=n$. In this case $p=\frac{1+n}{n(n+1)}=\frac{1}{n}$ and therefore the probability of selecting $(k,k')$ with $k< k'$ is
\[
2\left(1-\frac{1}{n}\right)\frac{1}{n(n-1)}
=\frac{2}{n^2} =p_c(k,k'),
\]
as desired.
\end{proof}
Finally, analogously to \eqref{eq:mixture_representation}, $\PNR$ is defined as $\PNR = F\bar{P}_{NR} F$, where
\begin{equation} \label{eq:directional_inf_nonrev}
 \bar{P}_{\text{NR}}((c,v),(c',v')) = \sum_{x\in \hat{\sK}} p_c(x)\tilde{P}_{x}((c,v),(c',v')),
\end{equation}
with $\tilde{P}_{x}$ lifted counterparts of $P_x$ in \eqref{eq:directional_inf_rev} and
\begin{equation}\label{eq:refresh_inf}
\begin{aligned}
F((c&,v),(c',v')) = \mathbbm{1}(c = c')\Bigg[\left(\frac{\xi}{n}\right) \left(\frac{1}{2}\right)^{\frac{n(n-1)}{2}} + \\
& \left(1-\frac{\xi}{n}\right)\left(\frac{1}{2}\right)^{\frac{\{n-K(c)\}\{n-K(c)-1\}}{2}}
\prod_{(k,k')\,:\,n_{k}(c)n_{k'}(c) > 0}\mathbbm{1}(v_{k,k'} = v'_{k,k'})\Bigg].
\end{aligned}
\end{equation}
The pseudocode for $\PNR$ is given in Algorithm \eqref{alg:PNR_inf}.

\begin{algorithm}[htbp]
\begin{algorithmic}
\State Update $(c,v) \sim F((c,v),(c,\cdot))$
\State Sample $x \sim p_c$ using Algorithm \ref{alg:Pkk_inf}
\State Set $c' = c$ and $v' = v$
\If{ $x \neq \emptyset$}
\State Set $(k_-,k_+)=(k,k')$ if $v'_{k, k'} = +1$, and $(k_-,k_+)=(k',k)$ if $v'_{k, k'} = -1$

\If{$n_{k_-}(c) > 0$}
\State Sample $i \sim \text{Unif}\left(\{i' \in [n]  \, \mid \, c_{i'} = k_{-}\}\right)$ 
\State With probability $\alpha(c,i,k_-,k_+)$ in \eqref{eq:acceptance_inf} set  $c' = (c_{-i}, k_+)$ and $v' = v$,
\State $\,\,$ otherwise set $(c',\bv') = (c,\bv^{(\textit{flip})})$, with $\bv^{(\textit{flip})}=(v_{-(k, k')},-v_{k, k'})$.
\Else 
\State Set $(c',\bv') = (c,\bv^{(\textit{flip})})$.
\EndIf
\EndIf
\State Update $(c',v') \sim F((c',v'),(c',\cdot))$
\end{algorithmic}
\caption{One step of the non-reversible kernel $(c',v')\sim\PNR((c,v),\cdot)$
\label{alg:PNR_inf}}
\end{algorithm}

Notice that, by definition of $F$ in \eqref{eq:refresh_inf}, each velocity component $v_{k,k'}$ such that $n_{k}(c)n_{k'}(c) = 0$ is resampled at each iteration from its stationary distribution. Therefore, in practical implementations, it suffices to store the $K(c)(K(c)-1)/2$ velocities corresponding to the allocated labels: anytime $x = (k,k')$ is selected with either $n_{k}(c) = 0$ or $n_{k'}(c) = 0$, the associated velocities can be generated independently and uniformly from $\{-1,+1\}$.

\subsection{Validity of the samplers and asymptotic variance}
The following theorem shows that most theoretical properties of Lemma \ref{lemma:erg_Prev}, Lemma \ref{lemma:ergodicity} and Theorem \ref{thm:asymp_variances} carry over to the setting of random partition models.
Analogously to Section \ref{sec:as_var}, we compare the asymptotic variance after adjusting for the number of predictive density evaluations per iteration. For $P_{\text{NR}}$ and $\Prev$ this means considering the lazy kernels $Q_{\text{R}}(c,c') = 0.5P_{\text{R}}(c,c') + 0.5\mathbbm{1}(c=c')$ and
$Q_{\text{NR}}(c,c') = 0.5P_{\text{NR}}(c,c') + 0.5\mathbbm{1}(c=c')$, as for the case with a finite number of components. For $\PMG$, the number of predictive density evaluations depends instead on the current state $c$ and therefore we consider the following mixture of lazy kernels
\[
Q_{\text{MG}}(c,c') = \frac{1}{n}\sum_{i = 1}^n\left\{ \frac{1}{K(c_{-i}) + 1} \PMGi(c,c') + \frac{K(c_{-i})}{K(c_{-i}) + 1} \mathbbm{1}(c = c')\right\},
\]
which requires on average one predictive evaluation per iteration.
\begin{theorem} \label{thm:random_partitions}
Let $\pi$ and $\tilde{\pi}$ be as in \eqref{eq:target_varying} and \eqref{eq:aug_invariant_beyond}. Assume that $\pi(c)>0$ for all $c\in [n]^n$. Let $\PMG$, $\Prev$ and $\PNR$ be as in respectively \eqref{eq:marginal_inf}, \eqref{eq:directional_inf_rev} and \eqref{eq:directional_inf_nonrev}. Then:
\begin{itemize}
	\item[(a)] $\Prev$ is $\pi$-reversible and uniformly ergodic. 
	\item[(b)] $\PNR$ is $\tilde{\pi}$-invariant and uniformly ergodic, provided that $\xi>0$.
	\item[(c)] For every function $g: [n]^n \to\mathbb{R}$ such that $g(c)=g(c')$ if $\rho(c)= \rho(c')$, it holds
	\[
	\Var(g,Q_{\text{NR}}) \leq \Var(g,Q_R)\leq 4\Var(g,Q_{\text{MG}}) + 3\Var_\pi(g).
	\]
\end{itemize}
\end{theorem}

\begin{proof}
Consider the setting of Lemma \ref{lemma:lifting}, with $x \in \hat{\sK}$ in place of $d \in [D]$, $p_c(x)$ defined as in \eqref{eq:proposal} and
\[
P_{(k,k'), v}(c, c') = \tilde{P}_{k, k'}\left((c,v\bm{1}), c' \times \{-1, +1\}^{n(n-1)/2}\right), \quad v \in \{+, -\},
\]
where $\bm{1} = (1, \dots, 1)$,
if $x = (k,k')$ and $P_{x,+}(c,c') = P_{x,-}(c,c') = \delta_{c}(c')$ if $x = \emptyset$, so that $K_R = \Prev$ and $K_{NR} = \PNR$ with $F$ in \eqref{eq:refresh_inf}. 

Then $\pi$-reversibility of $\Prev$ and $\tilde{\pi}$-invariance of $\PNR$ follow by Lemma \ref{lemma:lifting} provided assumptions \eqref{eq:single_coordinate_}-\eqref{eq:lifted_mixture_operator_} are satisfied. The only delicate condition to verify is \eqref{eq:single_coordinate_}, when $x \neq \emptyset$ (otherwise the condition holds trivially). Let $x = (k,k') \in \sK$ and notice that
\begin{align*}
\pi(c)p_{c}(k,k')&{P}_{x,+}(c,c') \\
&= \frac{\mu(\rho(c))}{[n]_{K(c)}} p_{c}(k,k') \frac{1}{n_{k}(c)} \min \left\{ 1, \left( \frac{n_k(c)}{n_{k'}(c) + 1} \right) \frac{\mu(\rho(c'))}{\mu(\rho(c))}\right\} \\
&= p_{c}(k,k') \frac{[n]_{K(c')}}{[n]_{K(c)}} \frac{\mu(c')}{[n]_{K(c')}} \frac{1}{n_{k'}(c')} \min\left\{ \left(\frac{n_{k}(c')}{n_{k}(c') + 1}\right) \frac{\mu(\rho(c))}{\mu(\rho(c'))} , 1 \right\} \\
&=  p_{c}(k,k') \frac{[n]_{K(c')}}{[n]_{K(c)}} \pi(c'){P}_{x,-}(c',c).
\end{align*}
Noting that
\[
\frac{p_{c'}(k,k')}{p_c(k,k')} = \frac{[n]_{K(c')}}{[n]_{K(c)}}  = 
\begin{cases} 
1 & \text{ if } K(c') = K(c) \\
n - K(c) & \text{ if } K(c') = K(c) + 1 \\
\frac{1}{n - K(c) + 1}& \text{ if } K(c') = K(c) - 1 
\end{cases},
\]
we obtain
\[
\pi(c)p_{c}(k,k'){P}_{(k,k'),+}(c,c') = \pi(c')p_{c'}(k,k')\hat{P}_{(k,k'),-}(c',c)
\]
which is exactly \eqref{eq:single_coordinate_}. Thus point (a) follows by Lemma \ref{lemma:lifting} (point (a)) and the same argument of \ref{lemma:erg_Prev}. Similarly, point (b) follows by Lemma \ref{lemma:lifting} (point (b)) and the same argument of \ref{lemma:ergodicity}.

Finally, we discuss point (c). Since $\Prev$ is clearly irreducible, by Lemma \ref{lemma:lifting} (point (c)) we have that $\Var(g,\PNR) \leq \Var(g,\Prev)$ which by Lemma \ref{lemma:lazy_as_var} implies $\Var(g,Q_{\text{NR}}) \leq \Var(g,Q_{\text{R}})$ for every $g : [n]^n \to \R$. Finally, the inequality $\Var(g,Q_R)\leq 4\Var(g,Q_{\text{MG}}) + 3\Var_\pi(g)$ 
follows from
\begin{equation}\label{eq:inequality_operators_inf}
Q_R(c, c') \geq \frac{1}{4}Q_{\text{MG}}(c, c'), \quad \rho(c) \neq \rho(c'),
\end{equation}
which is proven in Lemma \ref{thm:peskun_ordering_partitions} below.
In particular, \eqref{eq:inequality_operators_inf} implies that the Markov chains $\rho^{(t)}=\rho(c^{(t)})$ on $\Pi$ induced by, respectively, $Q_R$ and $Q_{\text{MG}}$ satisfy a Peskun ordering comparison (with constant $1/4$) and thus the result follows by, e.g., Theorem $2$ in \cite{zanella2020informed}. 
\end{proof}

\begin{lemma} \label{thm:peskun_ordering_partitions}
Let $\pi$, $\PMG$ and $\Prev$ be as in respectively \eqref{eq:target_varying}, \eqref{eq:marginal_inf} and \eqref{eq:directional_inf_rev}. Then \eqref{eq:inequality_operators_inf} holds.
\end{lemma}
\begin{proof}
Fix $c$ and $c'$ such that $c = (c_{-i}, k)$ and $c' = (c_{-i}, k')$ with $i \in [n]$ and $x = (k, k') \in \mathcal{K}$. Indeed for every other pair $(c, c')$ such that $c \neq c'$ we have that $Q_R(c, c') = Q_{\text{MG}}(c, c') = 0$. By definition of $Q_{\text{MG}}$ and $Q_R$ we have
\[
Q_{\text{MG}}(c, c') = \frac{\pi(c_i = k'\mid c_{-i})}{n(K(c_{-i})+1)}
\]
and
\[
Q_R(c, c') = \frac{1}{4n_k}p_c(x)\min \left\{1, \frac{n_{k}}{n_{k'} + 1}\frac{\mu(\rho(c'))}{\mu(\rho(c))} \right\},
\]
where $n_j = n_j(c)$ for every $j \in [K]$ and $n_k\geq 1$ by definition of $c$. Thus
\[
\begin{aligned}
Q_R(c, c') &= \frac{p_c(x)}{4(n_k+n_{k'})}\min \left\{\frac{n_k+n_{k'}}{n_k}, \frac{n_{k}+n_{k'}}{n_{k'} + 1}\frac{\mu(\rho(c'))}{\mu(\rho(c))} \right\}\\
& \geq \frac{p_c(x)}{4(n_k+n_{k'})}\min \left\{1, \frac{\mu(\rho(c'))}{\mu(\rho(c))} \right\}.
\end{aligned}
\]
Since $\rho(c) \neq \rho(c')$ we have three mutually exclusive cases:
\begin{enumerate}
\item If $n_k \geq 2$ and $n_{k'} \geq 1$, then $\mu(\rho(c'))/\mu(\rho(c)) = \pi(c')/\pi(c)$ and therefore
\[
Q_R(c, c') \geq \frac{p_c(x)}{4(n_k+n_{k'})}\min \left\{1, \frac{\pi(c')}{\pi(c)} \right\} \geq \frac{\pi(c_i = k' \mid c_{-i})}{4nK(c)} \geq \frac{1}{4}Q_{MG}(c, c').
\]

\item If $n_k \geq 2$ and $n_{k'} = 0$, then $K(c') = K(c) + 1$ and therefore
\[
\begin{aligned}
Q_R(c, c') &\geq \frac{p_c(x)}{4(n_k+n_{k'})}\min \left\{1, (n-K(c))\frac{\pi(c')}{\pi(c)} \right\}\\
&= \frac{p_c(x)}{4(n_k+n_{k'})}\pi(c_i = k' \mid c_{-i})\min \left\{ \frac{1}{\pi(c_i = k' \mid c_{-i})}, \frac{n-K(c)}{\pi(c_i = k \mid c_{-i})}\right\}\\
&\geq \frac{p_c(x)\left(n-K(c) \right)}{4(n_k+n_{k'})}\pi(c_i = k' \mid c_{-i}),
\end{aligned}
\]
where we used the fact that $\pi(c_i = \hat{k} \mid c_{-i}) \leq 1/(n-K(c))$ for every $\hat{k}$ such that $n_{\hat{k}}(c_{-i}) = 0$. This implies
\[
Q_R(c, c') \geq \frac{p_c(x)\left(n-K(c) \right)}{4(n_k+n_{k'})}\pi(c_i = k' \mid c_{-i}) \geq \frac{\pi(c_i = k' \mid c_{-i})}{n(K(c)+1)} = \frac{1}{4}Q_{MG}(c, c').
\]

\item If $n_k = 1$ and $n_{k'} \geq 1$, then $K(c') = K(c) - 1$ and therefore
\[
\begin{aligned}
Q_R(c, c') &\geq \frac{p_c(x)}{4(n_k+n_{k'})}\min \left\{1, \frac{1}{(n-K(c)+1)}\frac{\pi(c')}{\pi(c)} \right\}\\
&= \frac{p_c(x)}{4(n_k+n_{k'})}\pi(c_i = k' \mid c_{-i})\min \left\{ \frac{1}{\pi(c_i = k' \mid c_{-i})}, \frac{(n-K(c)+1)^{-1}}{\pi(c_i = k \mid c_{-i})}\right\}\\
&\geq \frac{p_c(x)}{4(n_k+n_{k'})}\pi(c_i = k' \mid c_{-i}),
\end{aligned}
\]
where we used again the fact that $\pi(c_i = \hat{k} \mid c_{-i}) \leq 1/(n-K(c))$ for every $\hat{k}$ such that $n_{\hat{k}}(c_{-i}) = 0$. This implies
\[
Q_R(c, c') \geq \frac{p_c(x)\left(n-K(c) \right)}{4(n_k+n_{k'})}\pi(c_i = k' \mid c_{-i}) \geq \frac{\pi(c_i = k' \mid c_{-i})}{4nK(c)} \geq \frac{1}{4}Q_{MG}(c, c').
\]
\end{enumerate}
\end{proof}

\begin{remark}
We note that the ordering in point $(c)$ of Theorem \ref{thm:random_partitions} could be extended to functions $g$ also depending on the latent parameters, analogously to Corollary \ref{crl:extension_asymptotic_variance} below.
\end{remark}

\section{Role of the refresh parameter $\xi$} \label{sec:xi}

As mentioned in Section \ref{sec:PNR_def} and proved in Lemma \ref{lemma:ergodicity}, keeping the refresh rate $\xi > 0$ in Algorithm \ref{alg:PNR} guarantees aperiodicity. Moreover $\xi = 0$ can actually lead to a periodic chain: for example, if $\pi$ is as in \eqref{eq:marg_priori} with $\bm{\alpha} = (1, \dots, 1)$ then it is easy to see that the acceptance ratio \eqref{eq:acc_ratio} is always equal to $1$ and $\PNR$ with $\xi = 0$ is periodic with period $2(n+1)$. At the same time, the larger $\xi$ the more frequent the refresh events, and therefore the shorter the excursions: as a consequence $\PNR$ will be closer to a reversible sampler and we expect its advantages over $\Prev$ (an $\PMG$) to disappear. As an extreme case, if $\{(c^{(t)}, v^{(t)}\}_{t \in \N}$ is a Markov chain with kernel $\PNR$ and $\xi = n$, then $\{(c^{(t)}\}_{t \in \N}$ is itself a reversible Markov chain with kernel $\Prev$. In this section we provide empirical evidence for choosing $\xi$ to be a small number, independent of $n$.

\subsection{Simulation study} \label{sec:sim_refresh}

We follow a similar strategy as in Section \ref{sec:posterior_sim}. We set $n = 1000$, $K = 3$, and $\alpha = (0.1, 0.1, 0.1)$, and generate $1000$ datasets from model \eqref{eq:original_model} with 
\[
f_\theta = \text{N}(y \mid \theta, 1), \quad p_0 = \text{N}(\theta \mid 0, 1).
\]
For each dataset we run $\PNR$ with different values of $\xi$, ranging from $0$ to $n$, for $2000 \times n$ iterations and we proceed as follows:
\begin{enumerate}
    \item For each $t = 1, \dots, 1000\times n$ and value of $\xi$ we estimate the Total Variation distance between the distribution of the first weight $w_1$ under $\PNR^t$ and the prior one (i.e.\ approximately a Beta$(0.1, 0.2)$).

    \item For each $\xi$ we compute the value of $t^*$, the smallest iteration for which the estimated distance is smaller than $0.1$ (to account for Monte Carlo and numerical error).

    \item Based on the second half of each chain, we compute the ESS of the entropy for each $\xi$: see Section \ref{sec:simDP} for the definition of entropy and its role for diagnostic purposes.
    
\end{enumerate}
Again as discussed in Section \ref{sec:posterior_sim}, the first two points rely on the fact that sampling from the prior is equivalent to sampling from the posterior, given data generated according to the (Bayesian) model: for more details on this type of diagnostic, see \cite{geweke2004getting}.
The results are reported in Figure \ref{fig:refresh_sim}, which depicts convergence time (left) and ESS (right) as a function of $\xi$. In both plots the substantial benefit of non-reversibility is clear when $\xi$ is small. As the latter increases toward $n$ (which corresponds to $\Prev$ and therefore reversibility), both metrics deteriorate: the number of iterations required to reach convergence increases, while ESS decreases. Moreover, the simulations suggest that fixing $\xi$ to a small number (say between $0$ and $1$) allows to retain most of the advantages of non-reversibility, with little sensitivity with respect to the exact choice.

\begin{figure}[h]
\centering
\includegraphics[width=.48\textwidth]{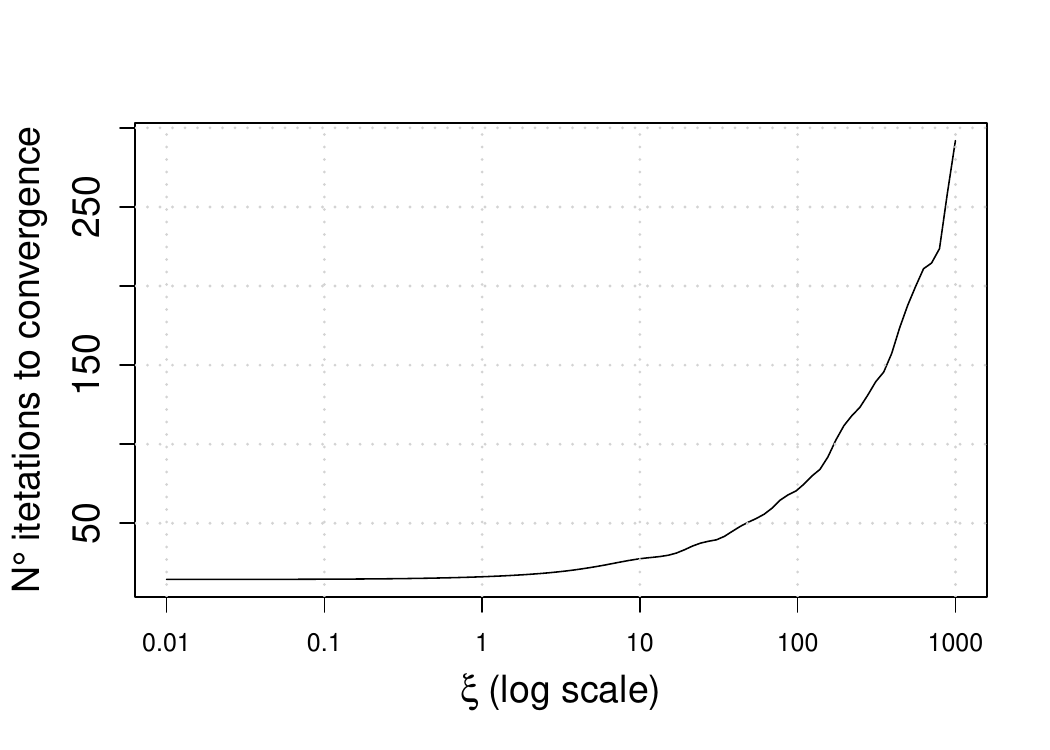} \,
\includegraphics[width=.48\textwidth]{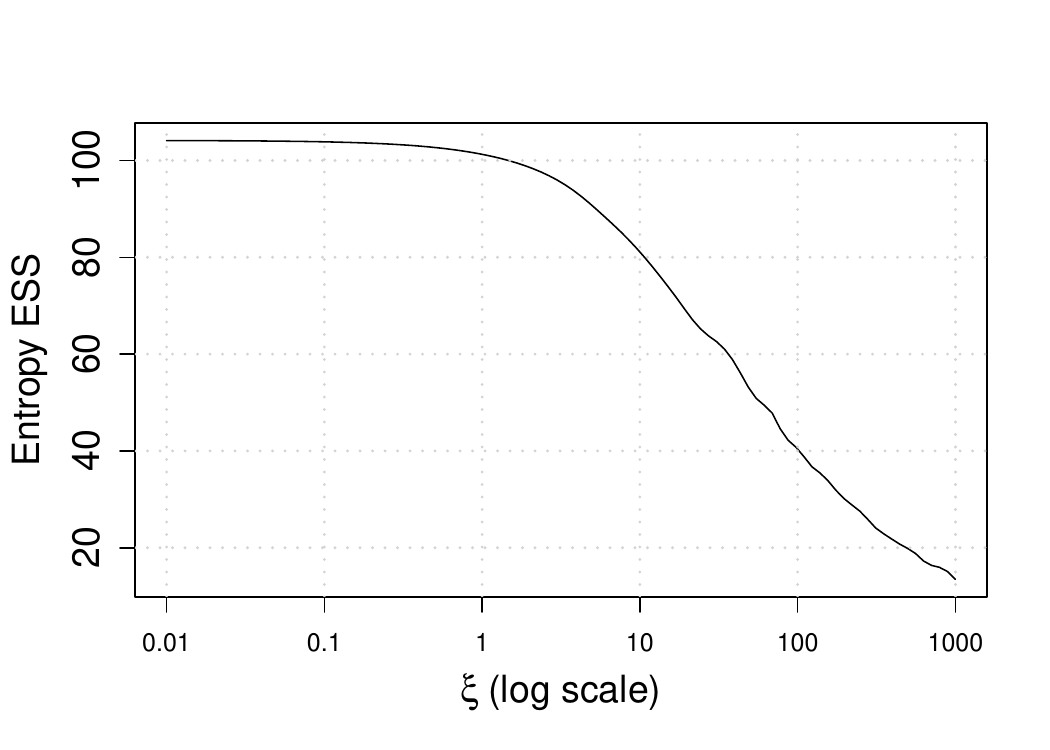} 
  \caption{\scriptsize{
Performance of $\PNR$ as a function of $\xi$ on datasets generated from \eqref{eq:original_model} with $f_\theta=\mathrm{N}(y\mid\theta,1)$, $p_0=\mathrm{N}(\theta\mid0,1)$, $K=3$, $n=1000$, and $\alpha=(0.1,0.1,0.1)$. Left: average number of iterations required to reach convergence, estimated as the first iteration at which the estimated distance between the distribution of the chain and the prior falls below $0.1$. Right: average effective sample size (ESS) of the partition entropy induced by $c$, computed over $1000$ draws after convergence. The $x$-axes is reported on a logarithmic scale.
  }}
 \label{fig:refresh_sim}
\end{figure}

\section{A variant of $\PNR$}\label{sec:variant}
The kernels $\Prev$ and $\PNR$ sample a new pair $(k,k')$ at every iteration. 
While this is natural and allows for direct theoretical comparisons with $\PMG$ (see Theorem \ref{thm:asymp_variances}), an alternative in the non-reversible case is to keep the same value of $(k,k')$ for multiple iterations. 
We thus define the following, non-reversible and $\tilde{\pi}$-invariant kernel
\begin{equation}\label{eq:tilPNR}
\QNR= \sum_{(k, k')\in \mathcal{K}} \frac{2}{K(K-1)}\sum_{t = 1}^\infty q_{m_c(k, k')}(t)\tilde{P}^{t}_{k, k'},
\end{equation}
with $m_c(k, k')=(n_k(c)+n_{k'}(c))/s$ for some fixed $s \in (0, 1)$ and $q_m(t)$ being the probability mass function of a geometric random variable with parameter $1/m$. 
The algorithm picks a couple $(k,k')$ uniformly at random and then takes a random number of steps of the lifted kernel $\tilde{P}_{kk'}$, with average number of steps proportional to the total number of points in the two clusters, i.e.\ $n_k(c)+n_{k'}(c)$. 
The associated pseudo-code is presented in Algorithm \ref{alg:PNR_variant}.
\begin{algorithm}[htbp]
\begin{algorithmic}
\State Sample $(k, k') \sim \text{Unif}(\mathcal{K})$
\State Sample $t\sim \text{Geom} \left(s/(n_k(c)+n_{k'}(c))\right)$ for some fixed $s \in (0,1)$
\State Sample $(c',v')\sim \tilde{P}^t_{k,k'}((c,v),\cdot)$
\end{algorithmic}
\caption{Modified non-reversible sampler $(c',\bv')\sim\QNR((c,\bv),\cdot)$
\label{alg:PNR_variant}}
\end{algorithm}
Reasoning as in Lemma \ref{lemma:ergodicity} it is easy to see that $\QNR$ is $\tilde{\pi}$-invariant and uniformly ergodic, as stated in the next lemma. 
\begin{lemma}\label{lemma:ergodicity_variant}
For any probability distribution $\pi$ on $[K]^n$, the Markov kernel $\QNR$ defined in Algorithm \ref{alg:PNR_variant} is $\tilde{\pi}$-invariant, with $\tilde{\pi}$ as in \eqref{eq:aug_invariant}.
Moreover, if $\pi(c) > 0$ for every $c \in [K]^n$, then $\QNR$ is irreducible, aperiodic and uniformly ergodic.
\end{lemma}

The distinction with the main algorithm is that $\PNR$ resamples the pair $(k, k')$ at each iteration with probability proportional to $n_k(c)+n_{k'}(c)$, while 
$\QNR$ keeps the same $(k,k')$ for $O(n_k(c)+n_{k'}(c))$ iterations and then resamples the pair $(k,k')$ uniformly from $\mathcal{K}$. Indeed we expect $\PNR$ and $\QNR$ to perform similarly for fixed values of $K$, but we empirically observe that $\QNR$ tends to yield slower mixing as $K$ increases: see Section \ref{sec:comparison_app} for a simulative comparison in the prior case. This motivated us to focus on $\PNR$ as the main scheme of interest in this paper.

\begin{remark}\label{rmk:scaling_variant}
Let
\[
V_1 =\left\{ v \in \{-1, 0, +1\}^{K(K-1)/2} \, : \, \sum_{k = 1}^Kv_k = 0, \, \sum_{k = 1}^Kv^2_k = 2\right\},
\]
and consider a Markov chain $\{(c^{(t)}, v^{(t)})\}_{t \in \N}$ on $[K]^n \times V_1$, such that if $(c^{(t)}, v^{(t)}) = (c, v)$, with $k_{-} = k'$ and $k_{+} = k$ if $v_{k} = -v_{k'} = +1$, then $\left(c^{(t+1)}, v^{(t+1)}\right) \sim \tilde{P}_{k, k'}((c, v), \cdot)$ with probability $1- s/(n_k(c) + n_{k'}(c))$, otherwise $c^{(t+1)} = c$ and $v^{(t+1)} \sim \text{Unif}(V_1)$. Then, by basic properties of the geometric distribution, $\{(c^{(t)}, v^{(t)})\}_{t \in \N}$ can be seen as the Markov chain keeping track of all the intermediate steps in Algorithm \ref{alg:PNR_variant}. 

Interestingly, in the prior case of Section \ref{sec:optimal_scaling}, where the invariant distribution is given by \eqref{eq:marg_priori}, it is possible to find a corresponding scaling limit. The proof is analogous to the case of $\PNR$ and we omit it for brevity, just limiting ourselves to identifying the candidate limit and discussing its implications. 
Given $\{(c^{(t)}, v^{(t)})\}_{t \in \N}$ as before, consider the Markov chain $\{(\bm{X}_t, \bm{V}_t)\}_{t \in \N}$ with $\bm{X}^{(t)}$ as in \eqref{eq:transf_c} and $\bm{V}^{(t)} = v^{(t)}$. 
With similar calculations as in Theorem \ref{thm:diffusion_nonrev}, the process $\left\{\bm{Z}^{(M, n)}_t\right\}_{t \in \R_+}$ defined as $\bm{Z}_t^{(M, n)} = \left(\bm{X}^{(M)}_{\lceil nt\rceil}, \bm{V}^{(M)}_{\lceil nt\rceil}\right)$ can be shown to converge to $\{\bm{Z}_t\}_{t \in \R_+}$ with generator
\[
\begin{aligned}
\mathcal{C}^{(M)} g(\bm{z}) = \mathbbm{1}(\bm{z}_1 \in E_M) \, \biggl\{&\frac{\partial}{\partial z_{1,k_+}}g(\bm{z})-\frac{\partial}{\partial z_{1,k_-}}g(\bm{z})\\
&+\max\left\{0, \frac{\alpha_{k_{-}}-1}{z_{1,k_-}}+\frac{1-\alpha_{k_{+}}}{z_{1,k_+}} \right\}\left[ g(\bm{z}_1, -\bm{z}_2) -g(\bm{z})\right]\\
&+\frac{s}{z_{1, k_-}+ z_{1, k_+}}\sum_{k\neq k'}\frac{2}{K(K-1)}\left[g\left(\bm{z}_1, \bm{z}_2^{(k, k')}\right)-g(\bm{z}) \right]\biggr\},
\end{aligned}
\]
with $k_{-} = k'$ and $k_{+} = k$ if $z_{2, k, k'} = +1$ and viceversa if $z_{2, k, k'} = -1$. Moreover $\bm{z}_2^{(k, k')}$ is the vector with $z_{2, k, k'} = -z_{2, k', k} = +1$ and zero otherwise. Interestingly, $\mathcal{C}^{(M)}$ coincides with the generator of the so-called Coordinate Sampler, introduced in \citet{wu2020coordinate}, with target distribution 
$\text{Dir}(\bm{\alpha})$.
\end{remark}

\subsection{The random projection sampler being approximated}\label{sec:variant_projection}
The main feature of $\QNR$ is that, after sampling a pair $(k, k') \in \mathcal{K}$, the operator $\tilde{P}_{k, k'}$ is applied for a random number of iterations. If $s \to 0$ the latter diverges almost surely, meaning that after selecting the pair the sampler will behave as $\tilde{P}^{t}_{k, k'}$ with $t\to\infty$.
By definition of $\tilde{P}^{t}_{k, k'}$ and ergodicity, this converges to the kernel $\Pi_{k, k'}$ that updates the sub-partition of points in clusters $k$ and $k'$ conditional on the rest, i.e.
\begin{align}\label{eq:proj_kernel} 
\lim_{t\to\infty}P^{t}_{k, k'}(c,c')&=
\tilde{\Pi}_{k, k'}(c,c')
\propto
\left(\prod_{i\,:\,c_i\notin\{k,k'\}}
\mathbbm{1}\left(c_i = c'_i\right)
\right)
\pi(c')
&c,c'\in[K]^n
\,.
\end{align}
Note that $\Pi_{k, k'}$ is a projection kernel, i.e.\ $\Pi^2_{k, k'}=\Pi_{k, k'}$. 
Analogously, again as $s \to 0$, $\QNR$ converges to the random projection kernel defined as
\begin{align}\label{eq:proj_kernel} 
\PRP(c,c')
&=
\frac{2}{K(K-1)}
\sum_{(k, k')\in \mathcal{K}} 
\Pi_{k, k'}(c,c')
&c,c'\in[K]^n
\,,
\end{align}
whose structure resembles the one of a random-scan Gibbs Sampler that updates the sub-partition of two randomly chosen pairs of clusters given the configuration of the other clusters.
In this perspective, 
$\QNR$ can be interpreted as a Metropolis-within-Gibbs sampler approximating $\PRP$. 

\begin{remark}
In the prior case, as $n \to \infty$, we expect $\PRP$ in turn to approximate a Gibbs sampler on the $(K-1)$-dimensional simplex, which at every iteration updates two coordinates chosen at random. In the special case of $\bm{\alpha} = (1, \dots, 1)$, the latter has been studied in \cite{smith2014gibbs} and shown to require $\sO(K \log (K))$ iterations for mixing. 
\end{remark}

\subsection{Comparison between $\PNR$ and $\QNR$}\label{sec:comparison_app}

In this section we consider the same setting of Section \ref{sec:sim_prior_case}, where the target distribution is given in \eqref{eq:marg_priori}. We run both $\PNR$ and $\QNR$ for $300$ independent trials with initialization uniformly at random: we chose $\xi^{-1} = K(K-1)/2$ in Algorithm \ref{alg:PNR} and $s = 1$ in Algorithm \ref{alg:PNR_variant} in order to match the expected number of refreshing events. We consider $n = 1000$, $K = 3, 10, 20 , 50$ and $\bm{\alpha} = (1, 1/(K-1), \dots, 1/(K-1))$, so that the marginal distribution on the proportion of the first component is a Dirichlet-Multinomial with parameters $(1,1)$ and thus close to a uniform distribution on $(0,1)$.

Figure \ref{fig:nonreversible comparison} plots the corresponding empirical marginal distribution obtained by the chains (black corresponds to $\PNR$ and gray to $\QNR$). Even if both schemes correctly reach stationarity, it seems that $\QNR$ yields slower mixing as $K$ increases: this is particularly evident in the case $K = 50$, where $\QNR$ remains close to the initial configuration.

\begin{figure}[h]
\centering
\includegraphics[width=.48\textwidth]{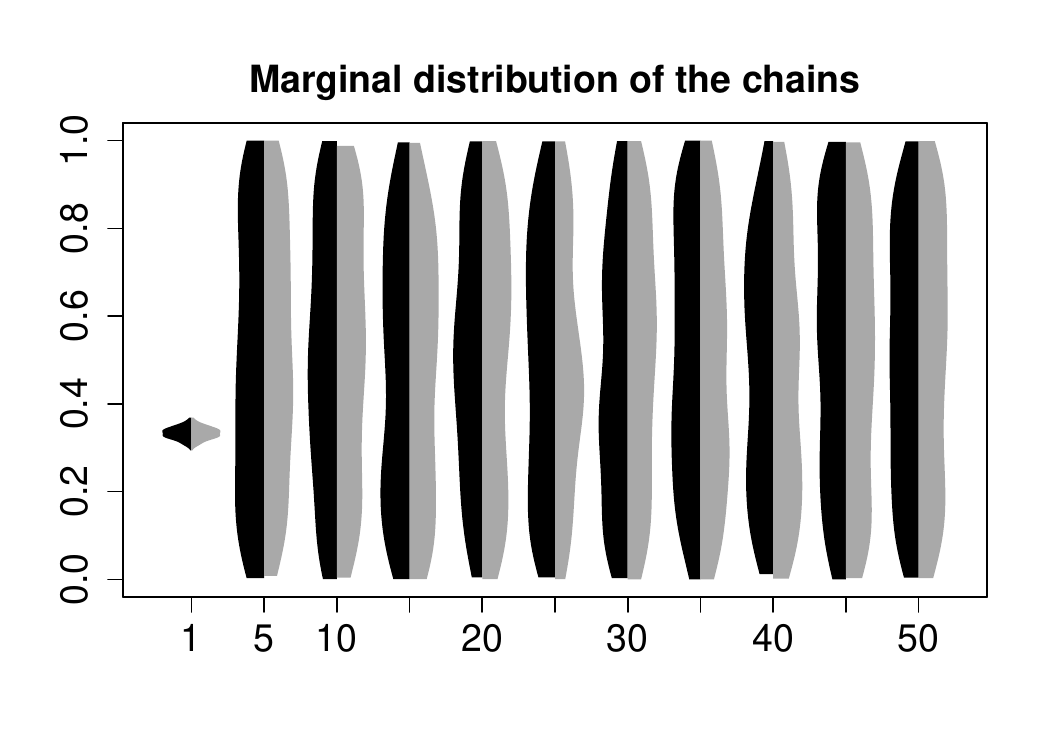} \,
\includegraphics[width=.48\textwidth]{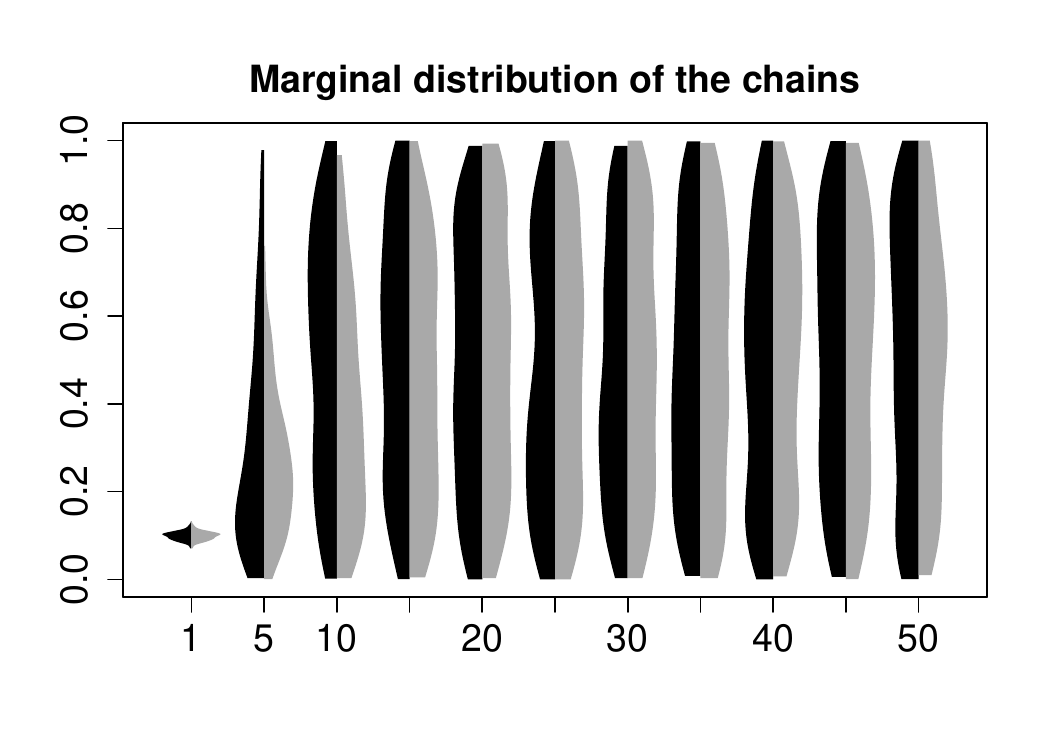}\\
\includegraphics[width=.48\textwidth]{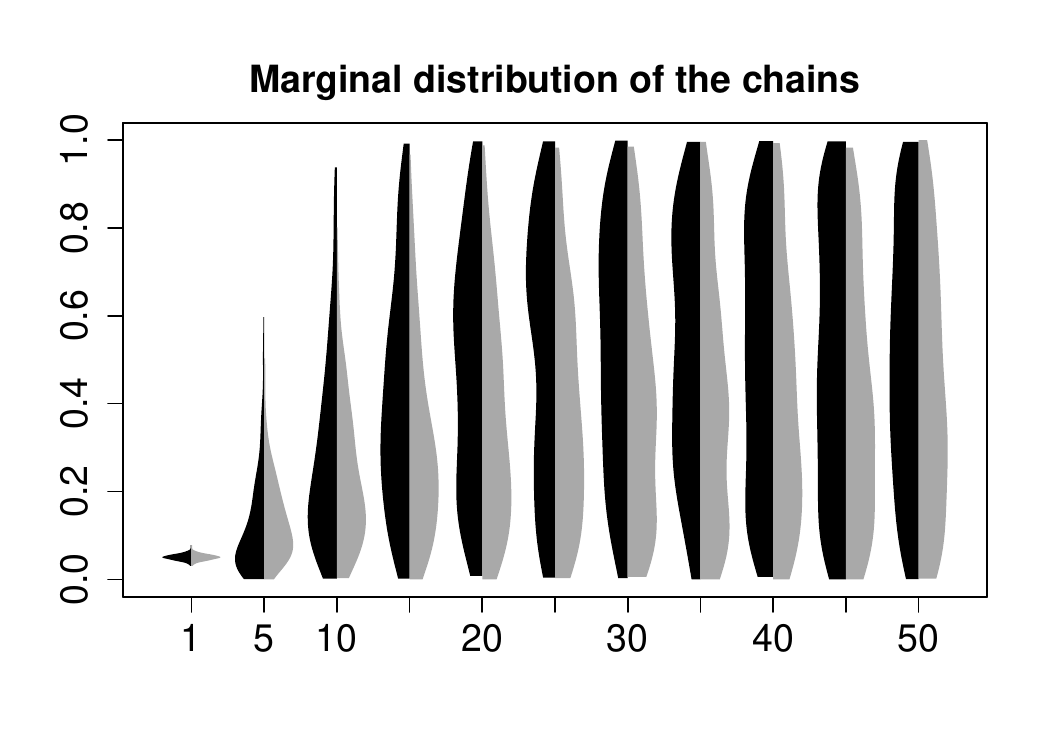} \,
\includegraphics[width=.48\textwidth]{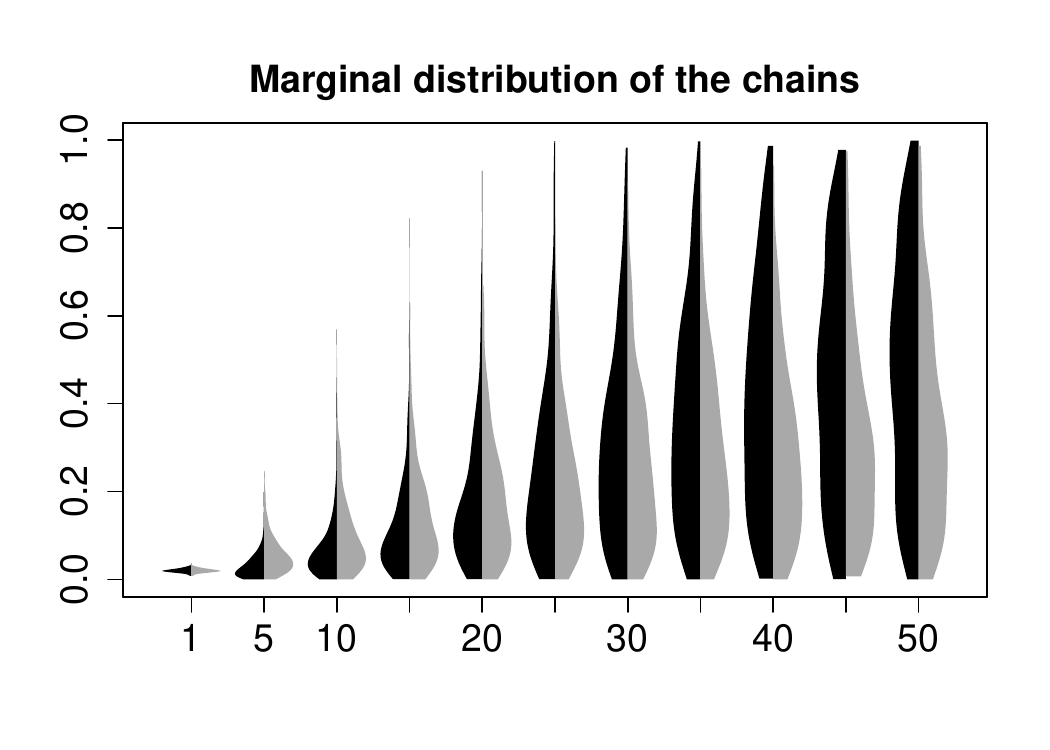}
 \caption{\scriptsize{
Plots of the marginal distribution of the proportion of the first component for the chains associated to $\PNR$ (black) and $\QNR$ (gray) at every. From top left to bottom right, the plots refer to $K = 3, 10, 20 , 50$, where the target distribution is as in \eqref{eq:marg_priori} with $\bm{\alpha} = (1, 1/(K-1), \dots, 1/(K-1))$ and $n = 1000$.
  }}
 \label{fig:nonreversible comparison}
\end{figure}

\section{Comparison with conditional sampler}\label{sec:conditional}

We now define the \emph{conditional} sampler targeting $\pi(c, \bm{\theta}, \bm{w})$ mentioned in Section \ref{sec:marginal}. 
The pseudocode is given in Algorithm \ref{alg:PCD} and we denote with $\PCD$ the associated Markov kernel on $[K]^n\times\Theta^K \times \Delta_{K-1}$. Also here we consider the random-scan case, which allows for an easier comparison with $\PMG$ and $\PNR$. We expect the main take-away messages to remain valid for the arguably more popular deterministic-scan scheme, even if theoretical results there are less neat (see e.g.\ \citet{roberts2015surprising, he2016scan, gaitonde2024comparison,ascolani2024entropy} and references therein).

\begin{algorithm}[htbp]
\begin{algorithmic}
\State Initialize $(c^{(0)},\bm{\theta}^{(0)}, \bm{w}^{(0)})\in [K]^n\times\Theta^K \times \Delta_{K-1}$
\For{$t \geq 1$}
   \State Sample $i \sim \text{Unif}\left(\{1, \dots, n+1\}\right)$.
   \If{$i \le n$}
    \State Sample $c^{(t)}_i \sim \pi(c_i \mid \bm{\theta}^{(t-1)}, \bm{w}^{(t-1)})$ with 
    \begin{equation} \label{eq:full_cond_cond}
    \pi(c_i = k \mid \bm{\theta}, \bm{w}) = \frac{w_kf_{\theta_k}(Y_i)}{\sum_{k' = 1}^Kw_{k'}f_{\theta_{k'}}(Y_i)}, \quad k = 1, \dots, K.
    \end{equation}
    \EndIf
    \If{$i = n+1$}
    \State Sample $\bm{w}^{(t)} \sim  \text{Dir}\left(\alpha_1 + n_1(c^{(t-1)}), \dots, \alpha_K + n_K(c^{(t-1)})\right)$.
    \State Sample $\theta_k^{(t)} \sim \pi(\theta_k \mid c^{(t-1)}) \, \propto \, \prod_{j \, : \, c_j^{(t-1)} = k}f_{\theta_{k}}(Y_j)p_0(\theta_k)$ for $k = 1, \dots, K$.
    \EndIf
   \EndFor
\end{algorithmic}
\caption{(Conditional sampler $\PCD$)
\label{alg:PCD}}
\end{algorithm}
    The next proposition, whose proof is inspired by the one of \cite[Thm.1]{liu1994collapsed}, shows that $\PMG$ always yields a smaller asymptotic variance than $\PCD$. Again with an abuse of notation we use $g$ to denote both $g \, : [K]^n \, \to \R$ and $g \, : [K]^n \times \Theta^K \times \Delta_{K-1} \, \to \R$ function of the first argument alone.
\begin{proposition}\label{prop:asymp_variances_cond}
Let $\pi$ be as in \eqref{eq:augmented_posterior} and $g \, : [K]^n \, \to \R$. Then for every $f_\theta$, $n$, $Y$ we have that $\Var(g, \PMG) \leq \Var(g, \PCD)$.
\end{proposition}
\begin{proof}
Let $\tPMG$ be the $\pi(c,\bm{\theta},\bm{w})$-reversible Markov kernel on $[K]^n \times \Theta^K \times \Delta_{K-1}$ that, given $(c^{(t)},\bm{\theta}^{(t)},\bm{w}^{(t)})$ generates $(c^{(t+1)},\bm{\theta}^{(t+1)},\bm{w}^{(t+1)})$  by
\[
c^{(t+1)}\sim \PMG\left(c^{(t)}, \cdot\right), \quad (\bm{\theta}^{(t+1)},\bm{w}^{(t+1)})\sim \pi\left(\bm{\theta},\bm{w}\mid c=c^{(t+1)}\right).
\]
By construction, $\Var(g, \PMG)=\Var(g,\tPMG)$ for any $g$ that is a function of $c$ alone, because the marginal process on $[K]^n$ induced by $\tPMG$ is a Markov chain with kernel $\PMG$.

We now compare $\tPMG$ and $\PCD$. Let
\[
\langle f,g\rangle_\pi = \int_{\sX}f(x) g(x) \pi(\d x), \quad \sX = [K]^n \times \Theta^K \times \Delta_{K-1},
\]
be the $L^2(\pi)$ inner product. Then for any 
$g\in L^2(\pi)$ and $(c,\bm{\theta},\bm{w})\sim \pi$ we have
\begin{equation}\label{eq:Dirichlet_form}
\begin{aligned}
&
\langle (I-\PCD)g,g\rangle_\pi
\\
&=\frac{1}{n+1}\sum_{i=1}^n
\E[\Var(g(c,\bm{\theta},\bm{w})\mid c_{-i},\bm{\theta},\bm{w})]
+
\frac{1}{n+1}
\E[\Var(g(c,\bm{\theta},\bm{w})\mid c)]
\\
&\leq 
\frac{1}{n+1}\sum_{i=1}^n
\E[\Var(g(c,\bm{\theta},\bm{w})\mid c_{-i})]
+
\frac{1}{n+1}
\sum_{i=1}^n
\frac{\E[\Var(g(c,\bm{\theta},\bm{w})\mid c_{-i})]}{n}
\\
&=
\frac{1}{n}\sum_{i=1}^n
\E[\Var(g(c,\bm{\theta},\bm{w})\mid c_{-i})]
=
\langle (I-\tPMG)g,g\rangle_\pi\,,
\end{aligned}
\end{equation}
where the middle inequality follows from the fact that
\[
\begin{aligned}
\E[\Var(g(c,\bm{\theta},\bm{w})\mid c)] & = \E\left[\E[\Var(g(c,\bm{\theta},\bm{w})\mid c)\mid c_{-i}]\right]\leq \E[\Var(g(c,\bm{\theta},\bm{w})\mid c_{-i})],
\end{aligned}
\]
for every $i = 1, \dots, n$ by the law of total variance. 
We thus have $\langle (I-\PCD)g,g\rangle_\pi\leq \langle (I-\tPMG)g,g\rangle_\pi $ for every $g\in L^2(\pi)$, which implies 
$\Var(g, \tPMG)\leq\Var(g,\PCD)$ for all $g$ (see e.g.\ the proof of \citealp[Theorem 4]{tierney1998note}). 
We thus have 
$\Var(g, \PMG)=\Var(g, \tPMG)\leq\Var(g,\PCD)$ for all $g$
for functions $g$ that depend only on $c$.
\end{proof}

\section{Extension of Theorem \ref{thm:asymp_variances} }\label{sec:extension_parameters}

As already discussed in Section \ref{sec:marginal}, even if $\PMG$, $\Prev$ and $\PNR$ give samples from $\pi(c)$, samples from $\pi(c, \bm{\theta}, \bm{w})$ can be obtained by drawing from $\pi(\bm{\theta}, \bm{w} \mid c)$. More formally, with the notation of Section \ref{sec:as_var}, we define
\begin{equation}\label{eq:def_extended1}
\check{Q}_j\left((c,  \bm{\theta}, \bm{w}), (\d c', \d \bm{\theta}', \d \bm{w}') \right) = Q_j(c, c')\pi(\d \bm{\theta}', \d \bm{w}' \mid c')
\end{equation}
for $j \in \{\text{MG, R}\}$ and
\begin{equation}\label{eq:def_extended2}
\check{Q}_{\text{NR}}\left((c,v,  \bm{\theta}, \bm{w}), (\d c', \d v', \d \bm{\theta}', \d \bm{w})' \right) = Q_{\text{NR}}((c, v), (c', v'))\pi(\d \bm{\theta}', \d \bm{w}' \mid c').
\end{equation}
We can then extend the ordering of Theorem \ref{thm:asymp_variances} to this more general setting, in Corollary \ref{crl:extension_asymptotic_variance} below. We first need a preliminary lemma.
\begin{lemma}\label{lm:asymp_var_extension}
Let $\pi$ be a probability distribution over $\sX \times \sY$ and let $\pi_1(\cdot) = \pi(\cdot \times \sY)$. Let $Q$ be a $\pi_1$-invariant Markov kernel on $\sX_1$ such that $\Var(g, Q)$ is well-defined for every $g : \sX \to \R$. If
\[
\check{Q}\left((x, y), (\d x', \d y') \right) = Q(x, \d x')\pi(\d y' \mid x'),
\]
then for every $h : \sX \times \sY \to \R$ such that $\Var_\pi(h) < \infty$ we have that
\[
\Var(h,\check{Q}) = \Var(g,Q)+\E_\pi[\Var_\pi(h \mid X)],
\]
where $g(x) = \E_\pi[h(x, Y)\mid x]$.
\end{lemma}
\begin{proof}
By definition we have that
\[
\Var(h,\check{Q}) = \Var_\pi(h) + 2\sum_{t\geq 1} \Cov(h(X_0,Y_0),h(X_t, Y_t)),
\]
where $(X_0, Y_0)\sim \pi$ and $(X_t, Y_t) \mid (X_0, Y_0) = (x_0, y_0) \sim \hat{Q}^t((x_0, y_0), \cdot)$. By the law of total variance we immediately have that
\begin{equation}\label{eq:equality_var}
\Var_\pi(h) = \Var_\pi(g)+\E_\pi[\Var_\pi(h \mid X)].
\end{equation}
Moreover, by the law of double expectation and the definition of $\check{Q}$ we have that
\[
\E\left[h(X_t, Y_t) \right] = \E\left[h(X_0, Y_0) \right] = \E_\pi[g]
\]
and
\[
\begin{aligned}
\E[h(X_0, Y_0)&h(X_t, Y_t) ] = \int_{\sX\times\sY}\int_{\sX\times\sY}h(x_0, y_0)h(x_t, y_t)\check{Q}^t\left((x_0, y_0), (\d x_t, \d y_t) \right)\pi(\d x_0, \d y_0)\\
&= \int_{\sX\times\sY}\int_{\sX\times\sY} h(x_0, y_0)h(x_t, y_t)Q^t\left(x_0, \d x_t \right)\pi(\d y_t \mid x_t)\pi(\d x_0, \d y_0)\\
&= \int_\sX\int_\sX \left[\int_\sY h(x_0, y_0) \pi(\d y_0 \mid x_0)\right]\left[\int_\sY h(x_t, y_t)\pi(\d y_t \mid x_t))\right]Q^t\left(x_0, \d x_t \right)\pi_1(\d x_0)\\
&= \int_\sX\int_\sX g(x_0)g(x_t)Q^t\left(x_0, \d x_t \right)\pi_1(\d x_0) = \E\left[g(X_0)g(X_t) \right],
\end{aligned}
\]
which imply
\begin{equation}\label{eq:equality_covar}
\Cov(h(X_0,Y_0),h(X_t, Y_t)) = \Cov(g(X_0),g(X_t)),
\end{equation}
for every $t \geq 1$. Then the statement follows by combining \eqref{eq:equality_var} and \eqref{eq:equality_covar}.
\end{proof}
\begin{corollary} \label{crl:extension_asymptotic_variance}
Let $\check{Q}_{\text{MG}}$, $\check{Q}_{\text{R}}$ and $\check{Q}_{\text{NR}}$ be as in \eqref{eq:def_extended1} and \eqref{eq:def_extended2}.
Then for every $h : [K]^n \times \Theta^K\times \Delta_{K-1} \to \R$ such that $\Var_\pi(h) < \infty$ we have that
\[
\Var (h,\check Q_{\text{NR}}) \leq \Var (h,\check Q_{\text{R}}) \leq 4 \Var (h,\check Q_{\text{MG}}) + 3\left[ \Var_\pi(h) - \mathbb{E}_\pi[ \Var_\pi (h|c) ] \right].
\]
\end{corollary}
\begin{proof}
Let $g(c) = \E_\pi\left[h \mid c \right]$. Then by Lemma \ref{lm:asymp_var_extension} and Theorem \ref{thm:asymp_variances} we have that
\[
\Var(h, \check{Q}_j) = \Var(g, Q_j) + \E_\pi[\Var_\pi(h \mid c)]
\]
for every $j \in \{\text{MG, R, NR}\}$ and
\[
\text{Var}(  g,Q_{\text{NR}}) \leq \text{Var}(  g,Q_R) \leq 4\text{Var}(  g,Q_{\text{MG}}) + 3 \text{Var}_\pi( g).
\]
Combining the two equalities above we get that
\[
\Var (h,\check Q_{\text{NR}}) \leq \Var (h,\check Q_{\text{R}}) \leq 4 \Var (h,\check Q_{\text{MG}}) + 3\text{Var}_\pi( g),
\]
and the statement follows by the law of total variance.
\end{proof}

\section{Definition of the piecewise deterministic Markov process in Section \ref{sec:optimal_scaling}}\label{sec:definition_PDMP}
Let
\[
V = \left\{ v \in \{-1, 0, +1\}^{K^2} \, : \, v_{k, k} = 0 \text{ for every } k \in [K] \right\}.
\]
Let $\left\{\bm{Z}^{(M)}_t\right\}_{t \in \R_+} = \left\{\bm{Z}^{(M)}_{1,t}, \bm{Z}^{(M)}_{2,t}\right\}_{t \in \R_+}$ be a piecewise deterministic Markov process \citep{davis1984piecewise} on $E_M \times V$ with $E_M$ in \eqref{eq:domain_PDMP}  and generator
\begin{equation}\label{eq:generator_PDMP}
\begin{aligned}
 \mathcal{B}^{(M)} g(\bm{z}) = \mathbbm{1}(\bm{z}_1 \in E_M) \, \biggl\{
     \sum_{k = 1}^K\Phi_k(\bm{z})&\frac{\partial}{\partial z_{1,k}}g(\bm{z})+\lambda(\bm{z})\sum_{k\neq k'}q(k, k')\left[g(\bm{z_{(k, k')}})-g(\bm{z}) \right]\\
     &+2 \xi\sum_{v \in V}\left(\frac{1}{2} \right)^{K(K-1)/2}\left[g(\bm{z}_1, v)-g(\bm{z}) \right] \biggr\}
     \end{aligned}
     ,
 \end{equation}
 for every $g \, : \, E_M\times V \, \to \, \R$ twice continuously differentiable in the first argument, where $\bm{z}^{(k, k')} \in E_M\times V$ is equal to $\bm{z}$ except for
 $
 \bm{z}^{(k, k')}_{2, k, k'} = - \bm{z}^{(k, k')}_{2, k', k} = -\bm{z}_{2, k, k'}
 $. We can think of $\left\{\bm{Z}^{(M)}_t\right\}_{t \in \R_+}$ as a process with an absorbing boundary, which remains constant as soon as  $Z^{(M)}_{1,t, k} \leq 1/M$ for some $k$.

In order to define $\Phi_k$, $\lambda$ and $q$ in \eqref{eq:generator_PDMP}, consider a inhomogeneous Poisson process $\Lambda_t$ with rate
\[
\lambda\left(\bm{Z}^{(M)}_t\right) = \frac{1}{2(K-1)}\sum_{k\neq k'}\left(Z^{(M)}_{1,t, k} + Z^{(M)}_{1,t, k'}\right)\beta\left(Z^{(M)}_{1, t, k}, Z^{(M)}_{1, t, k'}, Z^{(M)}_{2,t, k, k'} \right),
\]
where
\[
\beta(x_k, x_{k'}, v_{k', k}) = \max \left\{0, \frac{\alpha_{k_{-}}-1}{x_{k_-}}+\frac{1-\alpha_{k_{+}}}{x_{k_+}}\right\}
\]
with $k_{-} = k'$ and $k_{+} = k$ if $v_{k', k} = +1$ and viceversa if $v_{k', k} = -1$. In between events, $\{\bm{Z}^{(M)}_t\}_{t\in \R_+}$ evolves deterministically as
$\frac{\d Z^{(M)}_{1,t, k}}{\d t}= \Phi_k\left(\bm{Z}^{(M)}_t\right) $ with
\begin{equation}\label{eq:deterministic_PDMP}
\begin{aligned}
\Phi_k\left(\bm{Z}^{(M)}_t\right) 
& =\frac{1}{K-1}\left[\sum_{k' \, : \, Z^{(M)}_{2, t, k', k} = +1}\left(Z^{(M)}_{1,t, k} + Z^{(M)}_{1,t, k'}\right) - \sum_{k' \, : \, Z^{(M)}_{2, t, k', k} = -1}\left(Z^{(M)}_{1,t, k} + Z^{(M)}_{1,t, k'}\right)\right]
\end{aligned}
\end{equation}
and
$
\frac{\d Z^{(M)}_{2,t, k', k}}{\d t} = 0,
$ 
with $(k', k) \in [K]^2$. The system of differential equations in \eqref{eq:deterministic_PDMP} admits a unique solution by linearity in its arguments. Instead, at each event of $\Lambda_t$, say at $\tau>0$, a pair $(k, k') \in [K]^2$ is selected with probability
\[
q(k, k') \propto \frac{Z^{(M)}_{1,t, k} + Z^{(M)}_{1,t, k'}}{2(K-1)}\beta\left(Z^{(M)}_{1, \tau_-, k}, Z^{(M)}_{1, \tau_-, k'}, Z^{(M)}_{2,\tau_-, k', k} \right)\,\mathbbm{1}(k \neq k')
\]
and then the process jumps as follows:
\[
Z^{(M)}_{2, \tau, k', k} = - Z^{(M)}_{2, \tau-, k', k} \quad \text{and} \quad Z^{(M)}_{2, \tau, k, k'} = - Z^{(M)}_{2, \tau-, k, k'},
\]
where $\tau_{-}$ denotes the the left-limit at $\tau$. 
It follows that $\left\{\bm{Z}^{(M)}_t\right\}_{t \in \R_+}$ is a continuous-time process with generator as in \eqref{eq:generator_PDMP}. Such a process exists for every $M > 0$ since the rates $\lambda(\bm{z})$ are bounded \citep{davis1984piecewise}.

\section{Additional simulations for finite mixtures} \label{sec:additional_sim}

\subsection{Prior case}\label{sec:sim_prior_case}

Here we consider the prior case, where $f_\theta = f$ and the target distribution is given by \eqref{eq:marg_priori}. We let $K = 3$, $n = 1000$ and we run Algorithms \ref{alg:PMG} and \ref{alg:PNR} for $300$ independent runs, first with $\bm{\alpha}=(1,1,1)$ and then with $\bm{\alpha}=(0.1,0.1,0.1)$. 
Initial configurations are independently generated, so that $c_i^{(0)} \simiid \text{Unif}\left([K]\right)$.

For each run we store the value of the chains after $T=100\times n$ iterations and plot the corresponding proportion of labels of the first two components, i.e.\ $(n_1(c^{(T)})/n,n_2(c^{(T)})/n)$ in Figure \ref{fig:prior}. 
 If the chains had reached convergence by then, these should be $300$ independent samples approximately following 
a Dirichlet-Multinomial distribution with parameters $\bm{\alpha}$.
The resulting samples are displayed in Figure \ref{fig:prior}, with the same structure as in Figure \ref{fig:posterior}.
As expected (e.g. by Theorem \ref{thm:diffusion_rev} and \ref{thm:diffusion_nonrev}) the non-reversible scheme is much closer to the correct distribution, while $\PMG$ remains close to the initial configuration.

\begin{figure}[h]
\centering
\includegraphics[width=.31\textwidth]{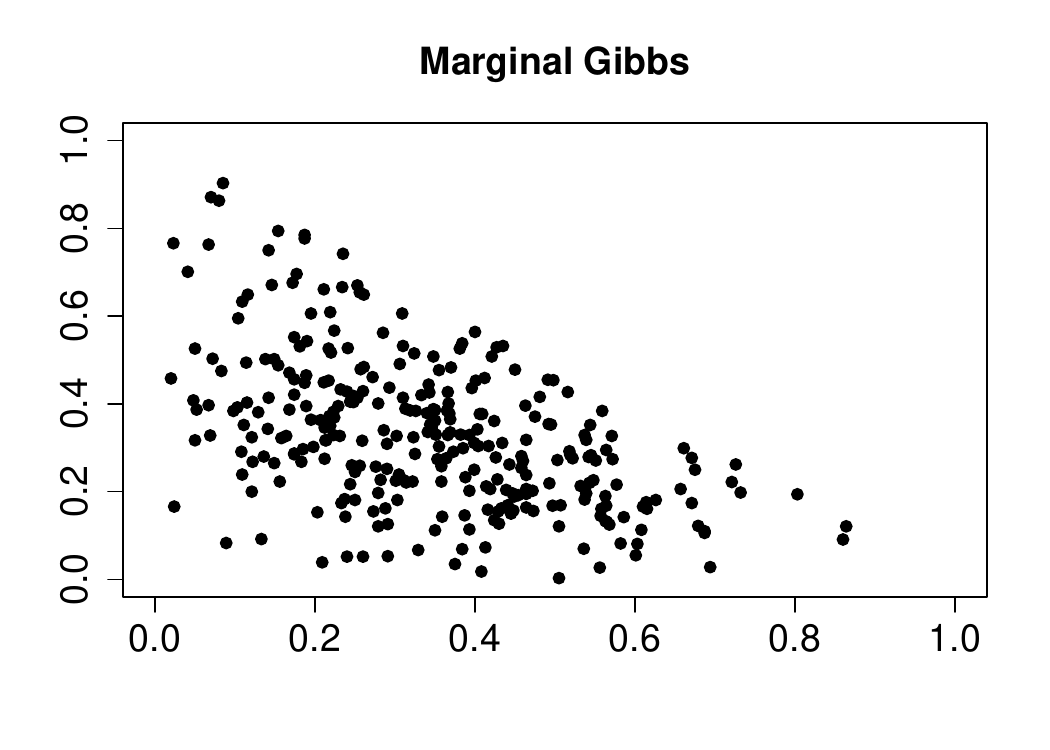} \,
\includegraphics[width=.31\textwidth]{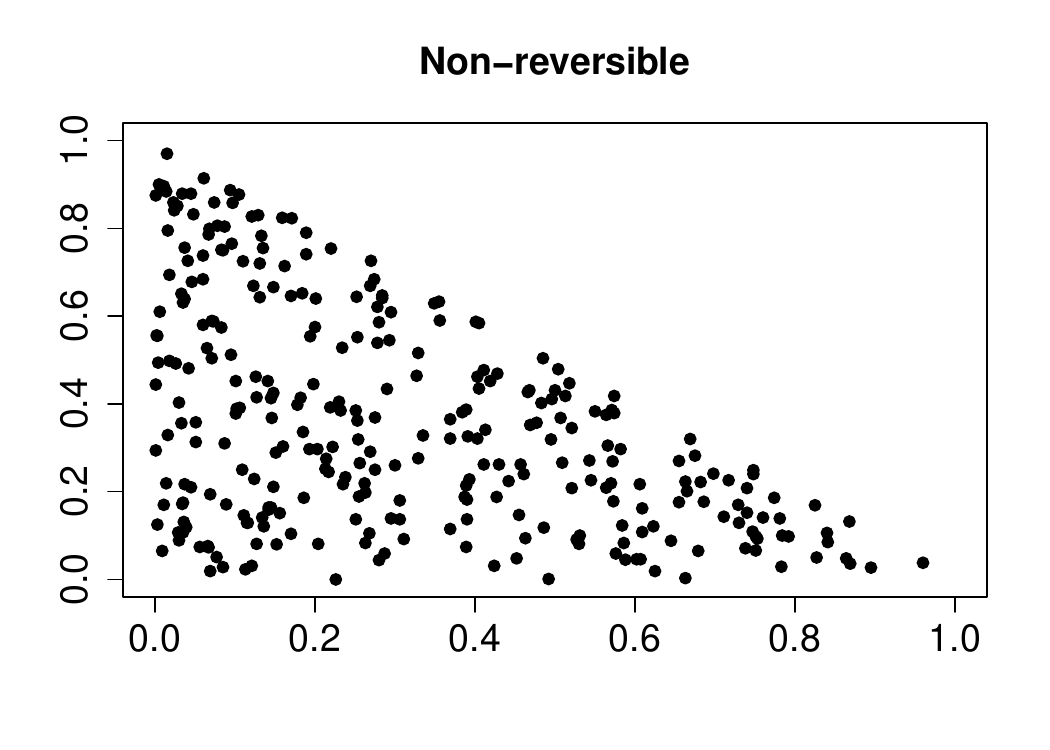} \,
\includegraphics[width=.31\textwidth]{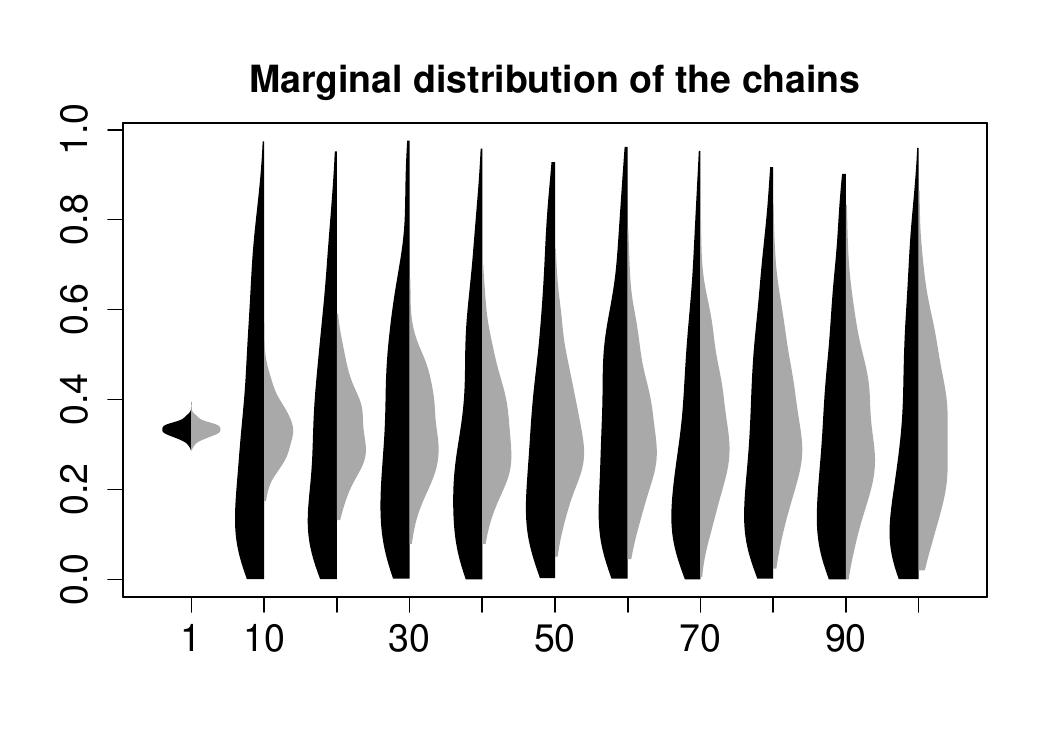}\\
\includegraphics[width=.31\textwidth]{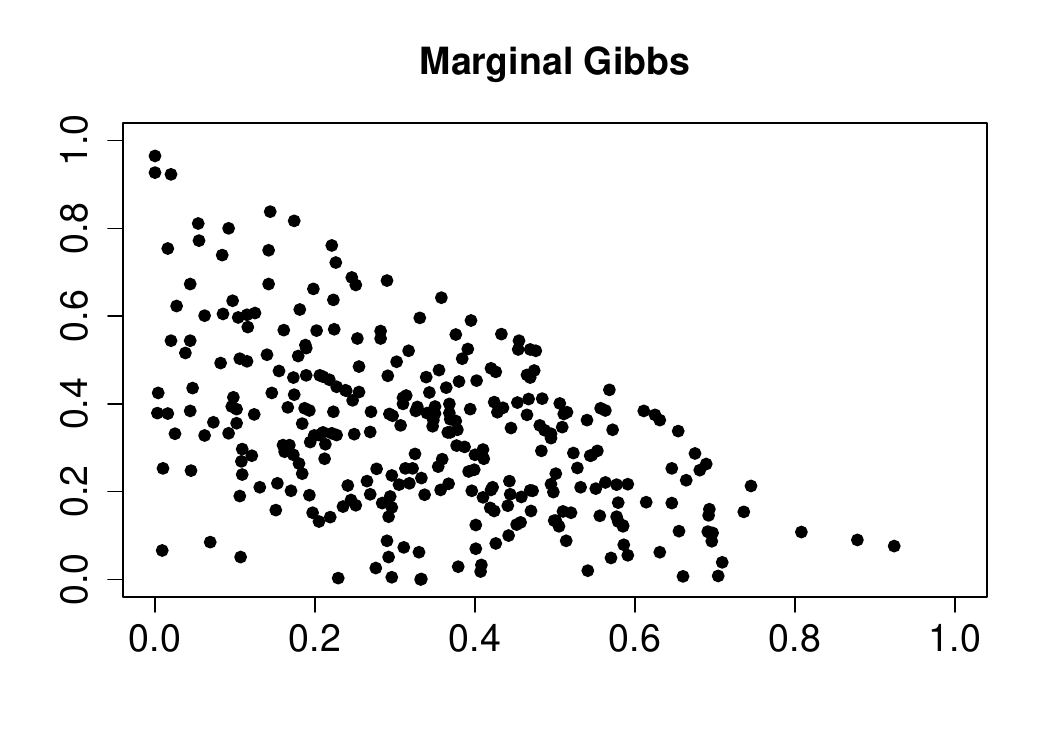} \,
\includegraphics[width=.31\textwidth]{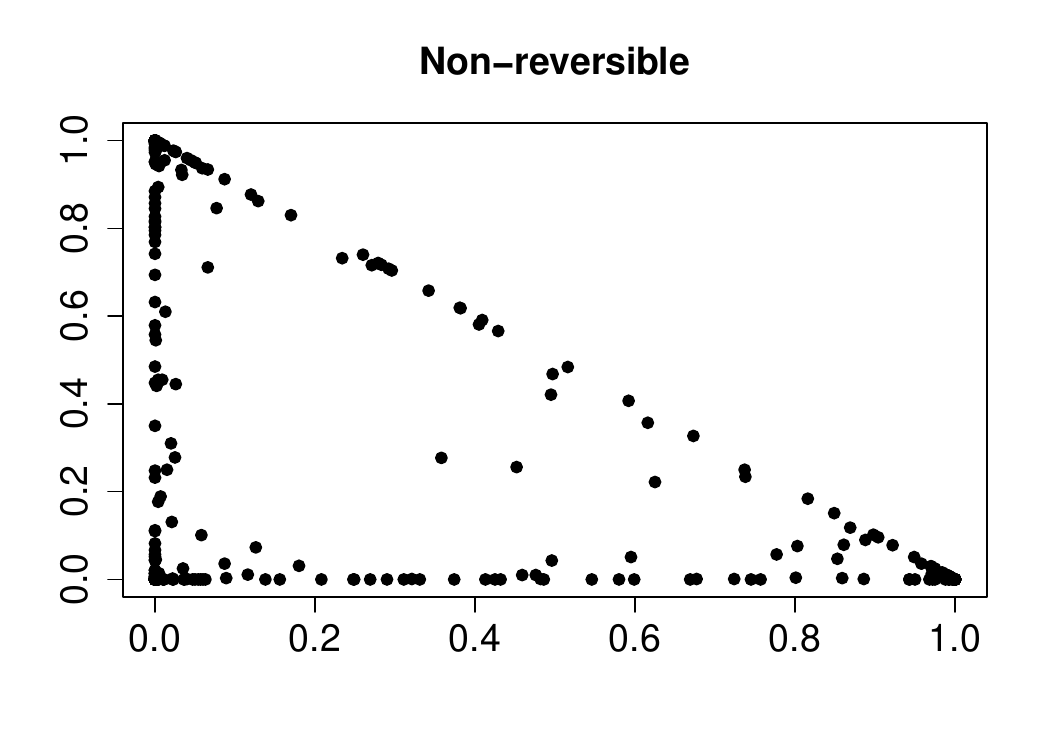} \,
\includegraphics[width=.31\textwidth]{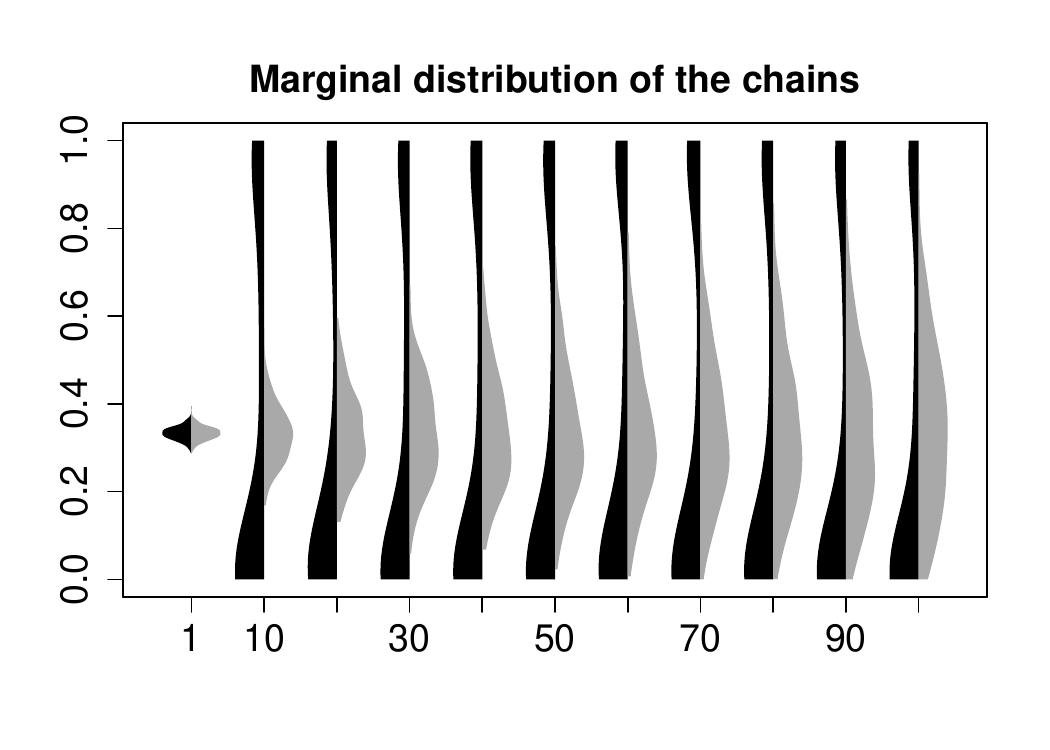}
 \caption{\small{
Left and center column: plot of the proportions of the first two components in the last of $100$ iterations (after a thinning of size $n$) over $300$ independent runs for $\PMG$ (left) and $\PNR$ (center). Right column: plot of the marginal distribution of the proportion of the first component at every $10$ iterations (after thinning) for $\PMG$ (gray) and $\PNR$ (black). The first and second rows refer to $\bm{\alpha}=(1,1,1)$ and $\bm{\alpha}=(0.1,0.1,0.1)$, respectively. The target distribution is given in \eqref{eq:marg_priori}.
  }}
 \label{fig:prior}
\end{figure}

\subsection{A high dimensional example} \label{sec:high_dimensional}
We now consider a higher dimensional version of the setting of Section \ref{sec:posterior_sim}, where
\begin{equation}\label{eq:normal_kernel_high_dim}
f_\theta(y) = N(y \mid \theta, \sigma^2_pI_p), \quad p_0(\theta) = N(\theta \mid \mu_0, \sigma^2_{0}I_p),
\end{equation}
where now $y \in \R^p$ and $\theta \in \R^p$ with $p\geq 1$. We rescale the likelihood variance as $\sigma_p^2 = cp$ which guarantees that
$
\frac{1}{\sigma^2_p}\sum_{j = 1}^p\left(\theta_{1j} - \theta_{2j} \right)^2 = \sO(1)$.
In other words, we ask that the distance across components, rescaled by the variance, does not diverge as $p$ grows: this implies that some overlap between components is retained and that the problem is statistically non-trivial (see e.g.\ \cite{chandra2023escaping} for more discussion of Bayesian mixture models with high-dimensional data).

We generate $500$ independent samples of size $n = 1000$ from model \eqref{eq:normal_kernel_high_dim} with $p = 18$, 
$K = 5$, $\mu_0 = 0$, $\sigma_0^2 = 0.5$, $c = 2$ and $\bm{\alpha} = (4, 1, \dots, 1)$. 
The data are generated as explained in the previous section and we run both $\PMG$ and $\PNR$, retaining only the last iteration for every chain: the initialization is again uniform at random. In Figure \ref{fig:high_dim} we plot the histograms of the last iteration for the proportion associated to the first component of $\PMG$ and $\PNR$ for $500$ independent runs. Comparing the latter with the prior density, given by a Dirichlet-Multinomial with parameters $(4,4)$ (approximately Beta$(4, 4)$), it is evident that the non-reversible scheme is able to forget the initialization while the reversible is not. Indeed, as also clear from the right plot of Figure \ref{fig:high_dim}, the marginal distribution of $\PMG$ significantly underestimates the size of the first cluster after $T=100\times n$ iterations.

\begin{figure}[h]
\centering
\includegraphics[width=.31\textwidth]{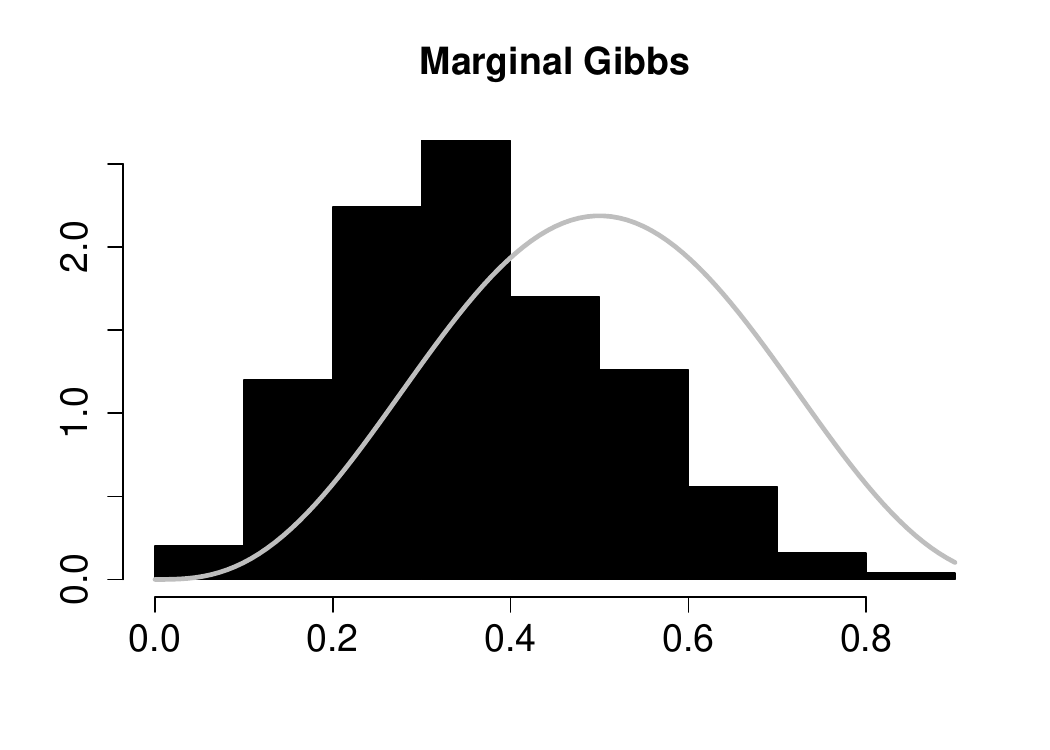} \,
\includegraphics[width=.31\textwidth]{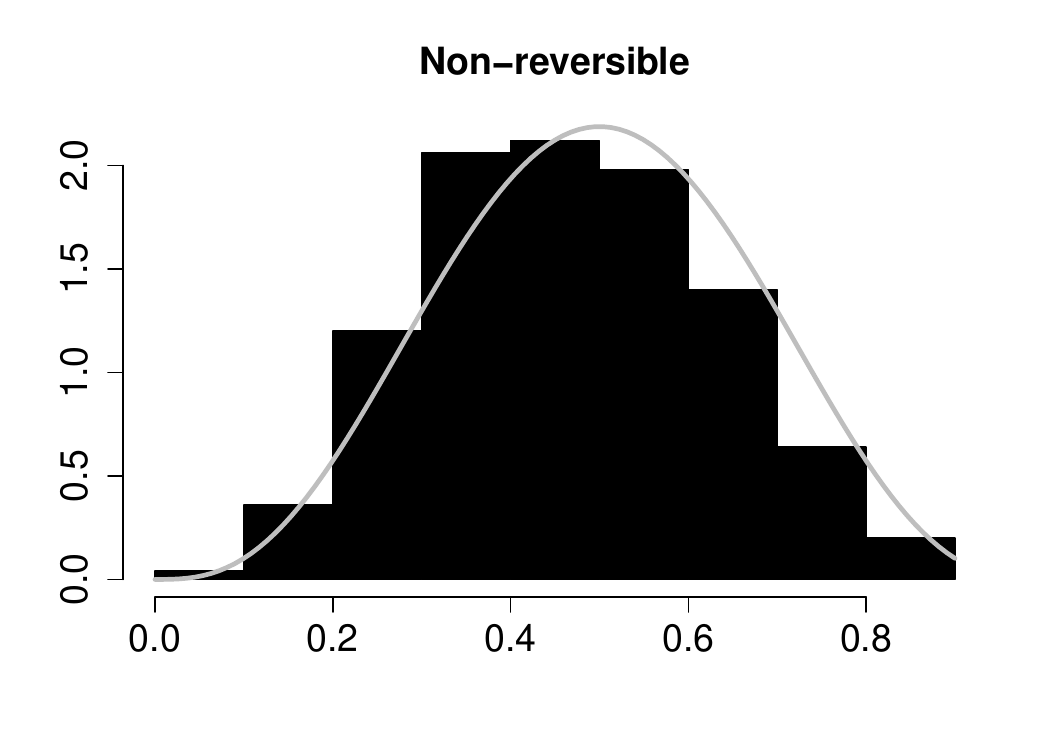} \,
\includegraphics[width=.31\textwidth]{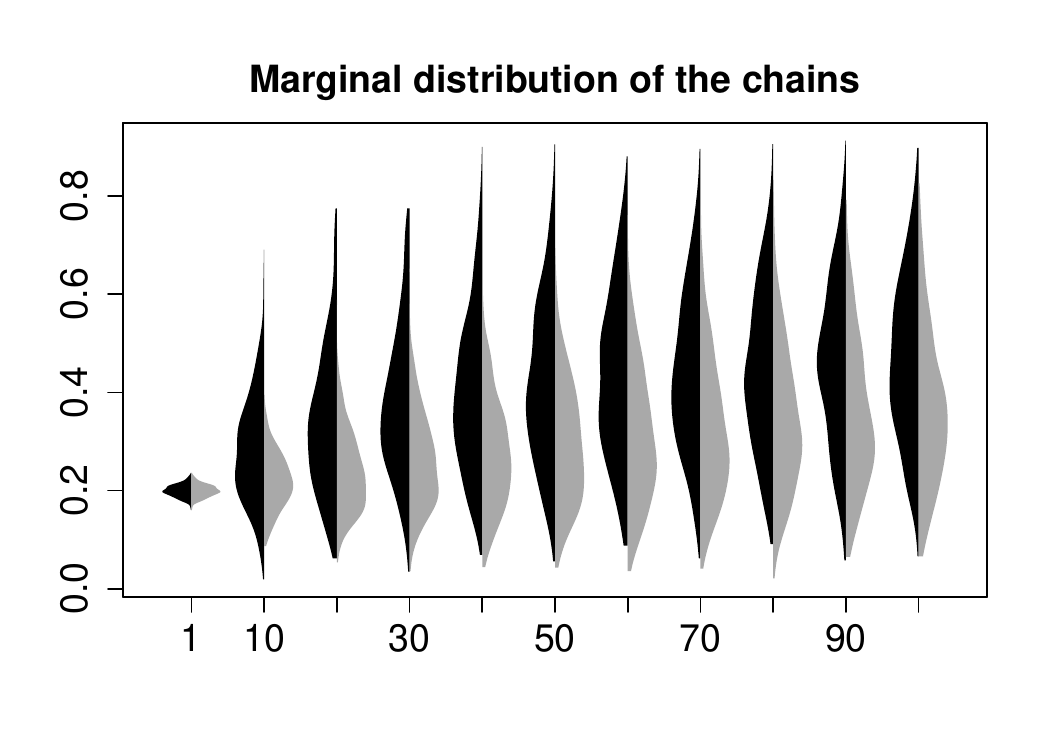}
  \caption{\scriptsize{
Left and center column: histogram of the proportion of the first component in the last of $100$ iterations (after a thinning of size $n$) over $500$ independent runs for $\PMG$ (left) and $\PNR$ (center). The gray line corresponds to the density of a \ Beta$(4, 4)$. Right column: plot of the marginal distribution of the chains at every $10$ iterations (after thinning) for $\PMG$ (gray) and $\PNR$ (black). The target distribution is given by the posterior of model \eqref{eq:original_model}, with $f_\theta(y)$ as in \eqref{eq:normal_kernel_high_dim}, $p = 18$, $K = 5$, $\mu_0 = 0$, $\sigma_0^2 = 0.5$, $c = 2$ and $\bm{\alpha} = (4, 1, \dots, 1)$.
  }}
 \label{fig:high_dim}
\end{figure}

\subsection{Simulations for the Poisson kernel}\label{sec:poisson_app}

Here we consider model \eqref{eq:original_model} with $K = 3$ and
\begin{equation}\label{eq:poisson_kernel}
f_\theta(y) = \text{Po}(y \mid \theta), \quad p_0(\theta) = \text{Gamma}(\theta \mid \beta_1, \beta_2).
\end{equation}
It is easy to show that the predictive distribution reads
\[
p(Y_{n+1} = y \mid Y) = \frac{\Gamma(\beta_1 + \sum_{i = 1}^{n}Y_i + y)}{\Gamma(\beta_1 + \sum_{i = 1}^{n}Y_i)\Gamma(y + 1)}\frac{(n+\beta_2)^{\beta_1 + \sum_{i = 1}^{n}Y_i}}{(n+\beta_2 + 1)^{\beta_1 + \sum_{i = 1}^{n}Y_i + y}}.
\]
We consider $\beta_1 = \beta_2 = 1$ and we draw $300$ independent samples from the model above with $n = 1000$, following the same procedure illustrated in Section \ref{sec:posterior_sim}. For each dataset we run Algorithms \ref{alg:PMG} and \ref{alg:PNR}, initialized uniformly at random, and we retain only the last iteration. 

The results of the simulations are similar to the ones of Section \ref{sec:posterior_sim}, as shown in Figure \ref{fig:posterior_poisson}: again the non-reversible scheme is much closer to the prior distribution, while $\PMG$ remains close to the initial configuration.

\begin{figure}[h]
\centering
\includegraphics[width=.31\textwidth]{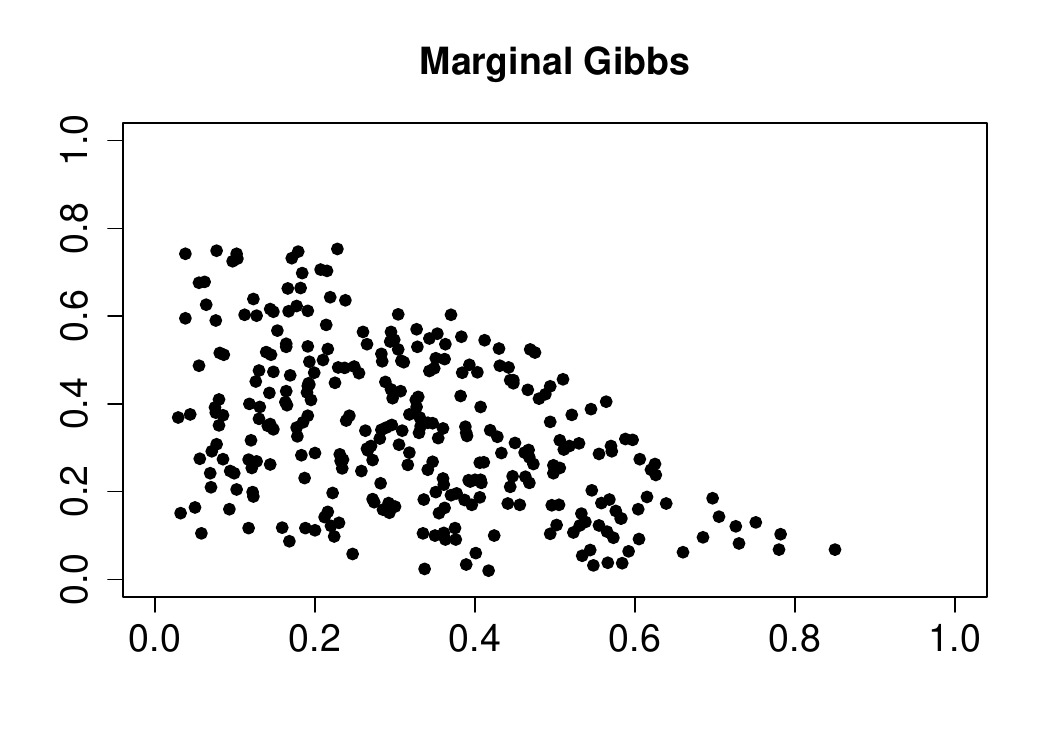} \,
\includegraphics[width=.31\textwidth]{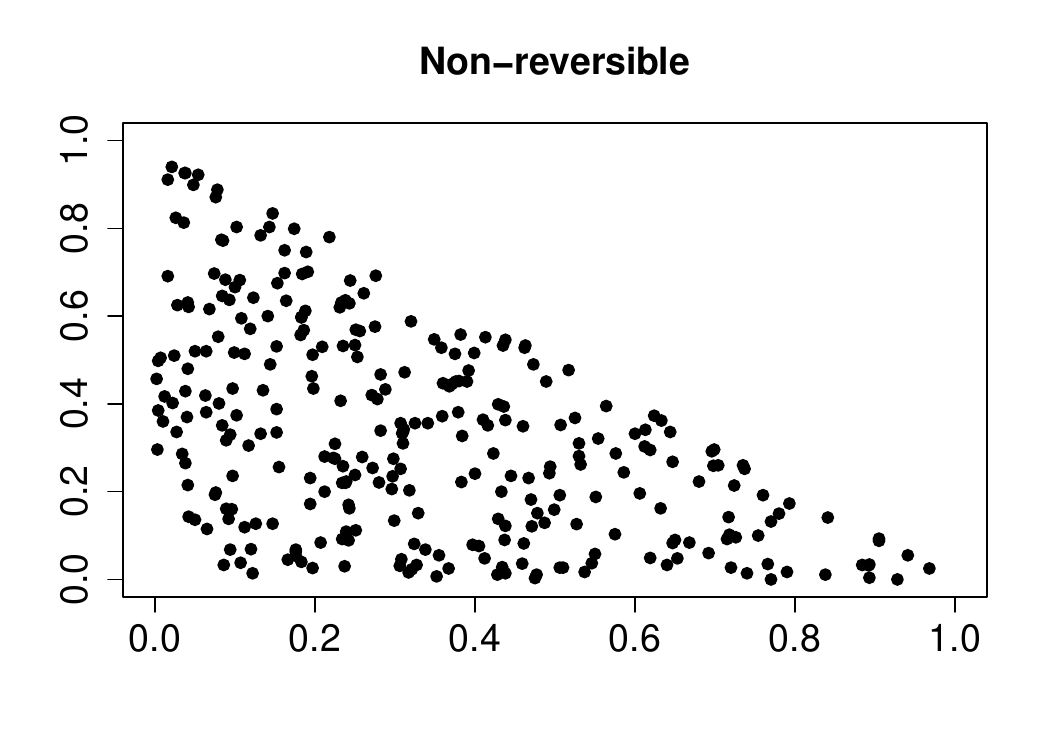} \,
\includegraphics[width=.31\textwidth]{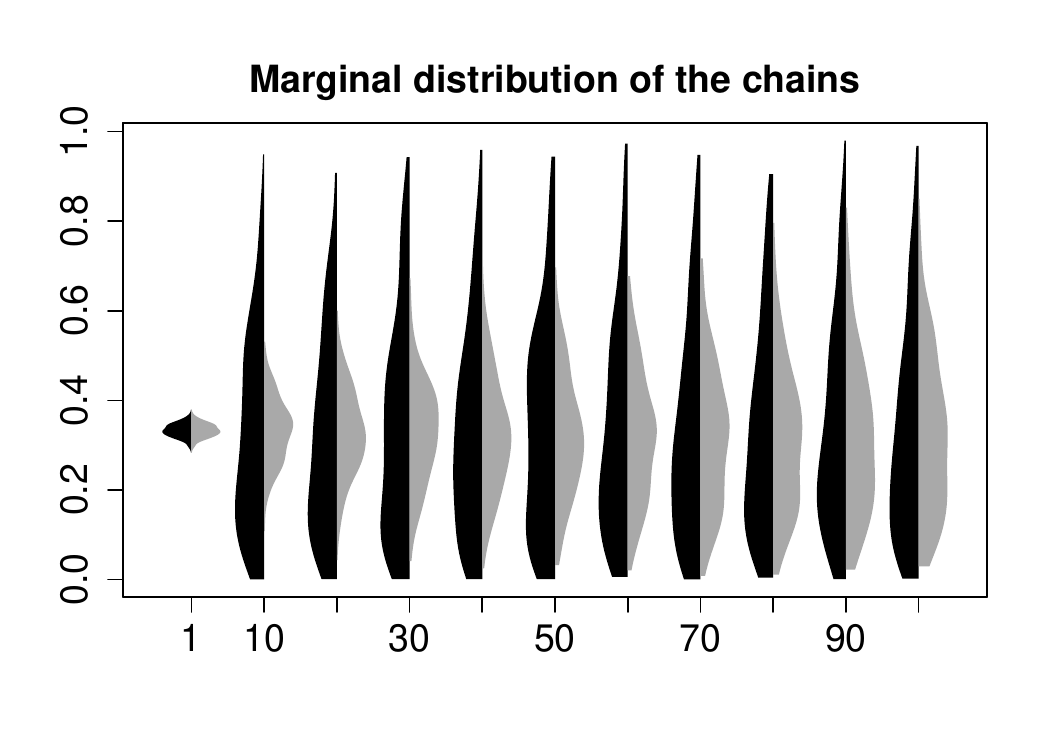}\\
\includegraphics[width=.31\textwidth]{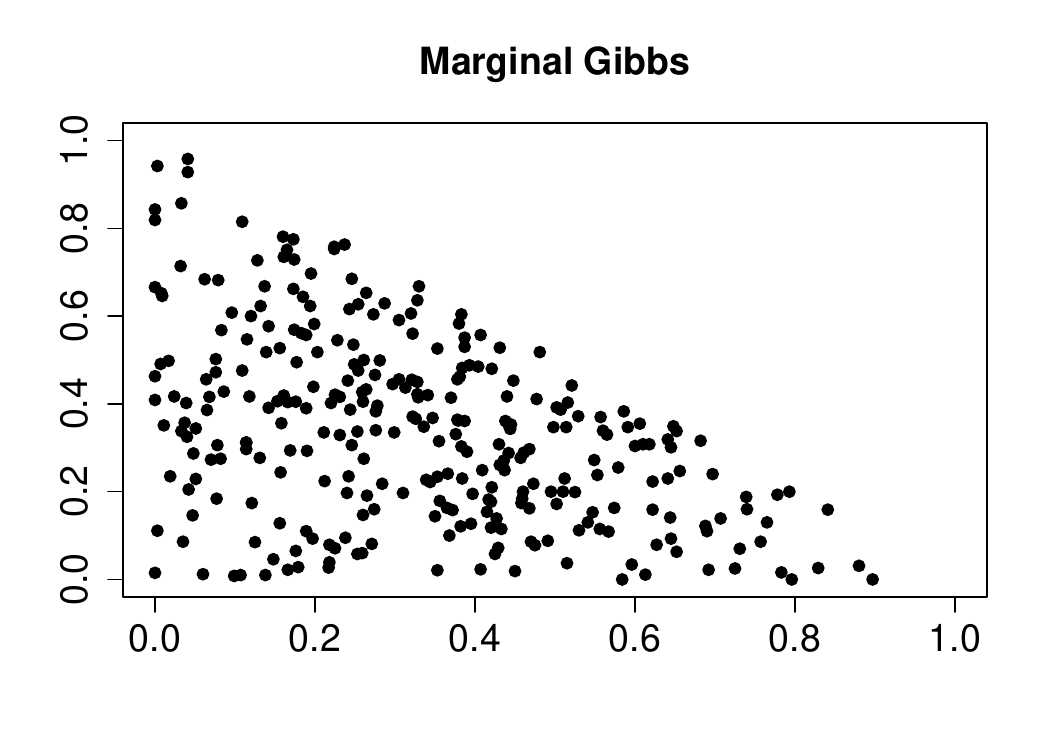} \,
\includegraphics[width=.31\textwidth]{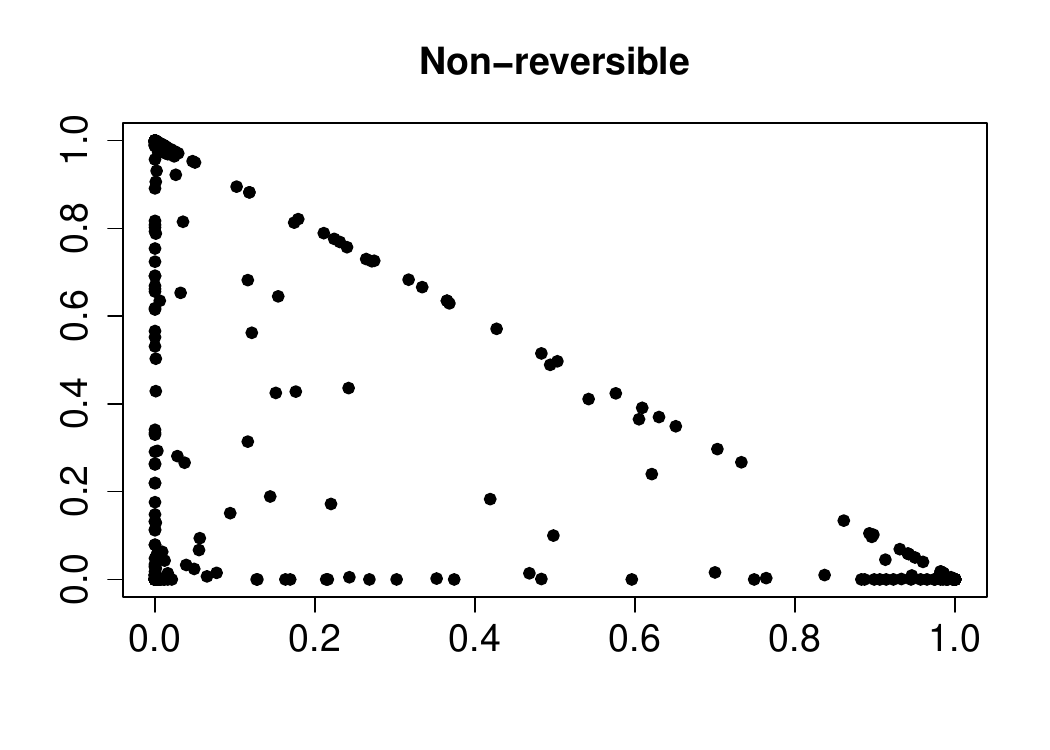} \,
\includegraphics[width=.31\textwidth]{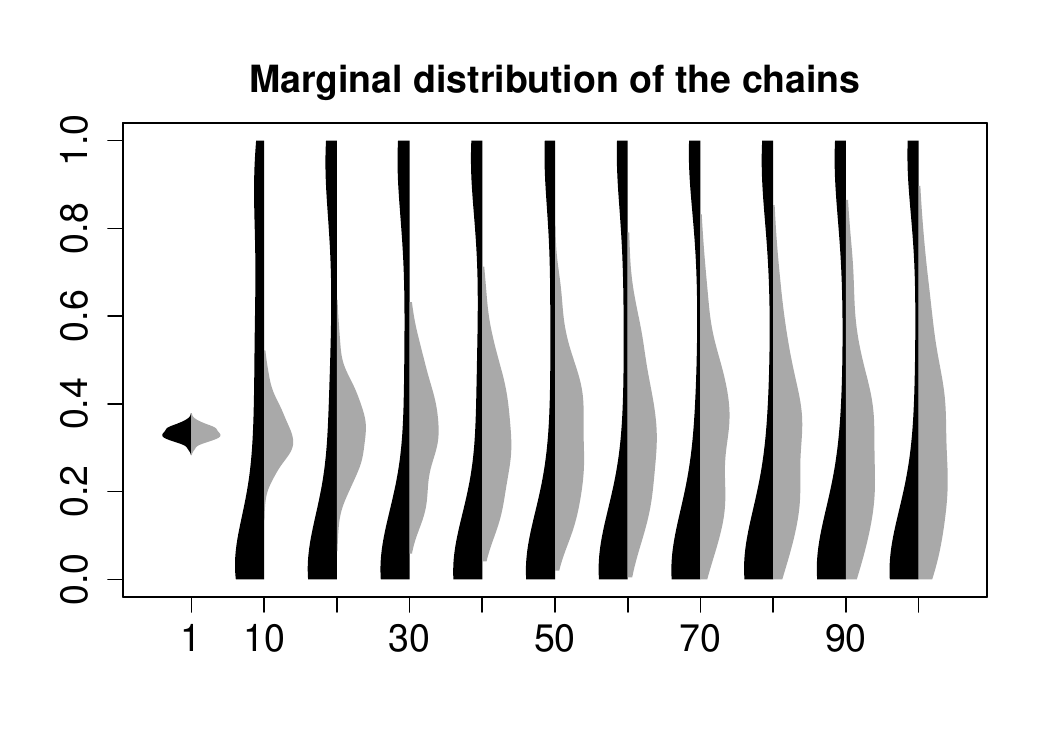}
 \caption{\scriptsize{
Left and center column: plot of the proportions of the first two components in the last of $100$ iterations (after a thinning of size $n$) over $300$ independent runs for $\PMG$ (left) and $\PNR$ (center). Right column: plot of the marginal distribution of the proportion of the first component at every $10$ iterations (after thinning) for $\PMG$ (gray) and $\PNR$ (black). The rows refer to $\alpha = 1$ and $\alpha = 0.1$ and the target distribution is given by the posterior of model \eqref{eq:original_model}, with $f_\theta(y)$ as in \eqref{eq:poisson_kernel} and $\beta_1 = \beta_2 = 1$.
  }}
 \label{fig:posterior_poisson}
\end{figure}

\section{Label switching and parallel tempering }\label{sec:tempering}

The posterior distribution arising \eqref{eq:marg_c}, with a fixed $K$, can exhibit different types of multimodality. In particular, permuting the labels of $c$ 
leads to $K!$ local modes: if in addition $\alpha_1 = \dots = \alpha_K$ such modes are all equivalent.
Even if this phenomenon, called label switching \citep{stephens2000dealing, marin2005bayesian}, might not be a real concern when inference is focused on label-invariant properties (e.g.\ the entropy or the density function), it can have a large impact on the performances of MCMC estimators for some quantities of interest (like $\theta_1$ or $w_1$).

A common strategy when dealing with multimodality is given by parallel tempering \citep{geyer1991markov}, which in this setting operates on a sequence of targets
\begin{equation} \label{eq:tempered_targets}
\pi_\beta(c) \, \propto \, \left[\prod_{k = 1}^K\Gamma\left(\alpha_k + n_k(c)\right)\right]\left[\prod_{k = 1}^K\int_{\Theta}\prod_{i \, : \, c_i = k}f_{\theta_{k}}(Y_i)p_0(\theta_k)\, \d \theta_k\right]^\beta,
\end{equation}
where $\beta \in \{\beta_1, \dots, \beta_T\}$ and $1 = \beta_1 > \cdots > \beta_T$ = 0, so that $\pi_{\beta_1} =\pi$ in \eqref{eq:marg_c} and $\pi_{\beta_T}$ corresponds to the prior in \eqref{eq:marg_priori}. 
The parallel tempering algorithm runs $T$ Markov chains ($\PMG$ or $\PNR$ in our case), each targeting one of the distributions in \eqref{eq:tempered_targets}, and periodically proposes swaps between the states of adjacent chains. In this section we empirically assess the ability of $\PMG$ and $\PNR$ to explore the modes due to label switching, within a parallel tempering procedure.

We simulate two datasets of size $n = 1000$ from a three-component mixture model with weights $\bm{w} = (0.1,0.7,0.2)$ and Gaussian kernel $\mathrm{N}(y \mid \bm{\theta}_0,1)$. We consider a overlapping scenario with $\bm{\theta}_0 = (-2,1,2)$ and a well separated setting with $\bm{\theta}_0 = (-8,1,8)$. We fit model \eqref{eq:original_model} with
\[
f_\theta = \mathrm{N}(y \mid \theta,1), \quad \alpha = (1,1,1), \quad p_0 = \mathrm{N}(\theta \mid 0,100).
\]
For both datasets the tempering schedules were selected using the approach of \citet{syed2022non}, which leads to $T = 21$ and $T = 40$ respectively. For each sampler, we performed $100$ independent runs, with each run consisting of $200000 \times n$ local updates at every temperature. Every $n$ iterations, swaps between chains with adjacent temperature parameters were proposed using the deterministic even-odd scheme of \citet{okabe2001replica}.

To assess global mixing, at each iteration $t$ we sample $\theta^{(t)} \sim \pi(\theta \mid c^{(t)})$ and define the label permutation statistic $\sigma^{(t)} = \left(R_1^{(t)},R_2^{(t)},R_3^{(t)}\right)$, where $R_k^{(t)}$ denotes the rank of $\theta_k^{(t)}$:
\[
R_k^{(t)} = \sum_{h=1}^K \mathbbm{1}\left(\theta_h^{(t)} \leq \theta_k^{(t)}\right).
\]
Since $\alpha_1 = \alpha_2 = \alpha_3$, the posterior distribution of $\sigma$ is uniform over the set of all $3! = 6$ permutations. The center panels of Figure \ref{fig:tempering} show the average total variation distance between the uniform distribution and the empirical distribution of $\sigma$ estimated from the samples collected up to each iteration, averaged over the $100$ independent runs. The x-axis reports the average number of predictive evaluations used by the chain up to that point. For the tempering schemes, this quantity is divided by $T$ to account for the fact that local updates can be performed in parallel across temperatures. The right panels instead show the ESS of the permutations, computed using the second half of each chain. Specifically, each permutation $\sigma$ is encoded as an integer from $1$ to $6$, and the resulting ESS is normalized by the average number of predictive evaluations per temperature.

\begin{figure}[h]
\centering
\includegraphics[width=.8\textwidth]{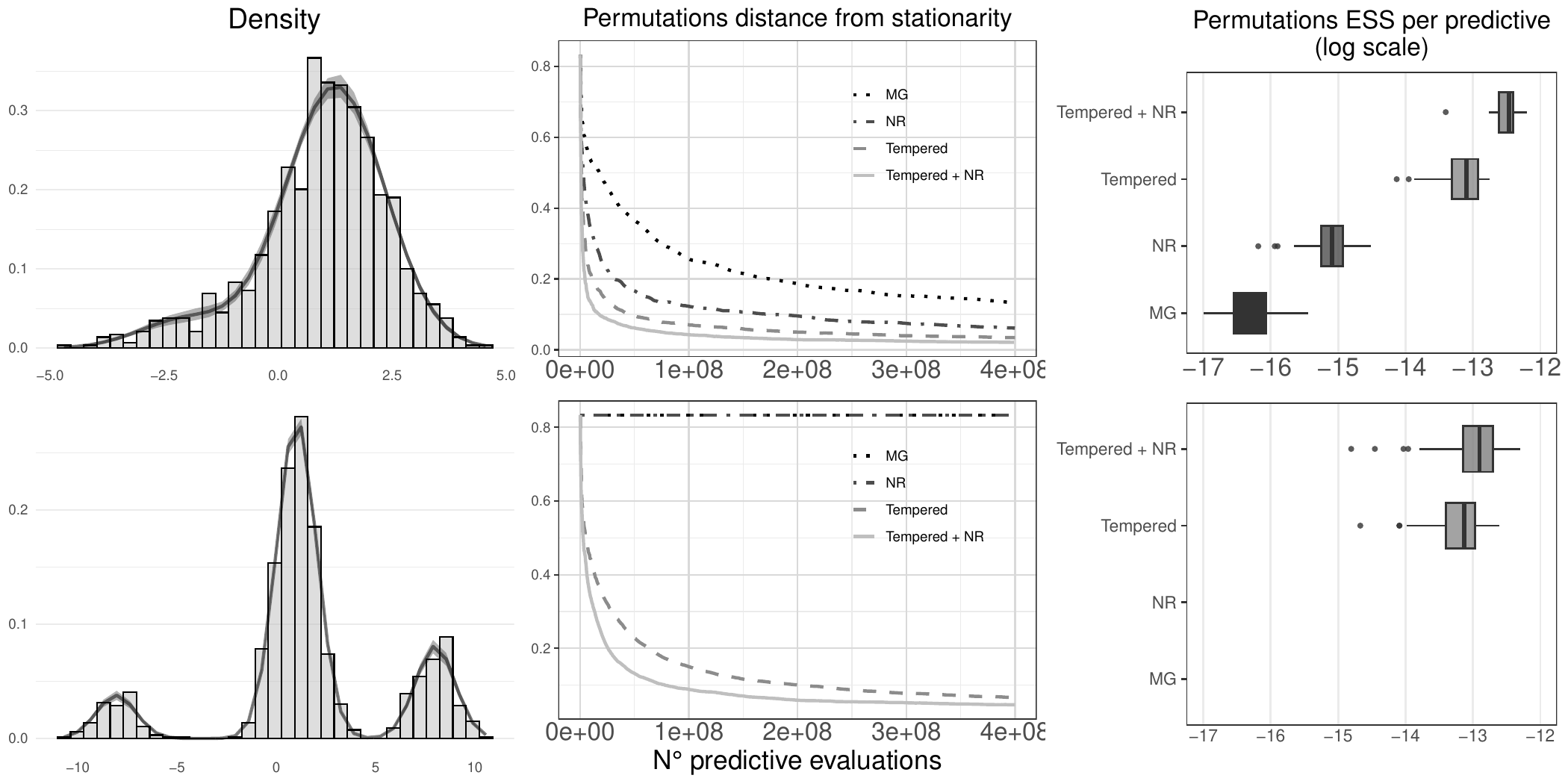}
 \caption{\scriptsize{ Comparison of posterior exploration by four MCMC samplers under two mixture configurations: overlapping components, $\theta_0 = (-2,1,2)$ (top row), and well-separated components, $\theta_0 = (-8,1,8)$ (bottom row). Left: posterior density estimated by the fitted model. Center: evolution of the Monte Carlo estimate of the distance between the empirical distribution of label permutations and the uniform distribution, averaged over $100$ independent runs. Right: boxplots of the effective sample size (ESS) of the permutations after convergence, normalized by the average number of predictive evaluations per temperature. In the bottom right plot $\PMG$ and $\PNR$ never reach convergence and therefore the boxplots are not shown.}}
 \label{fig:tempering}
\end{figure}

Notice that in the well-separated case both $\PMG$ and $\PNR$ get trapped in a local mode (i.e.\ in one of the permutations), while combining them with parallel tempering allows for a meaningful exploration. Notice also that a non-reversible local update always improves upon the reversible counterpart, also within a tempering scheme.

The results of this simulation study suggest that, when there is sufficient overlap between mixture components, local samplers are still able to explore the entire multimodal posterior. In contrast, when the components are well separated, they might fail to adequately explore the different permutations, regardless of reversibility, and techniques such as tempering become necessary to make inference on permutation-dependent quantities.

\section{Additional details on the real data application of Section \ref{sec:data}}\label{sec:more_data}

The model is a particular case of \eqref{eq:model_random}, where $\rho$ is in \eqref{eq:prior_PY} with
\[
\alpha \sim \text{Unif}(0,15),\quad\delta \sim \text{Unif}(0,1),
\]
while $f_\theta$ and $p_0$ are as in \eqref{eq:model_PY} with
\[
\mu_0 \sim \text{N}(0,10), \quad 
\lambda_0 \sim \text{Gamma}(1,1),\quad\alpha_0 \sim \text{Gamma}(1,1),\quad\beta_0 \sim \text{Gamma}(5,2).
\]
In Figure \ref{fig:galaxy} (left) we report the estimated posterior density obtained with a long run of $\PNR$, together with a point estimate of the clustering obtained by minimizing the posterior expectation of the Variation of Information loss \citep{wade2018bayesian, dahl2022sequentially}. This  leads to $7$ estimated clusters, with several overlapping components and units exhibiting high allocation uncertainty, as also highlighted by the estimated posterior similarity matrix in Figure \ref{fig:galaxy} (right).
 This suggests that non-reversible dynamics could be advantageous in exploring the posterior distribution, as confirmed by the results in Table \ref{tab:results}.
\begin{figure}[ht]
\centering
\includegraphics[width=1\textwidth]{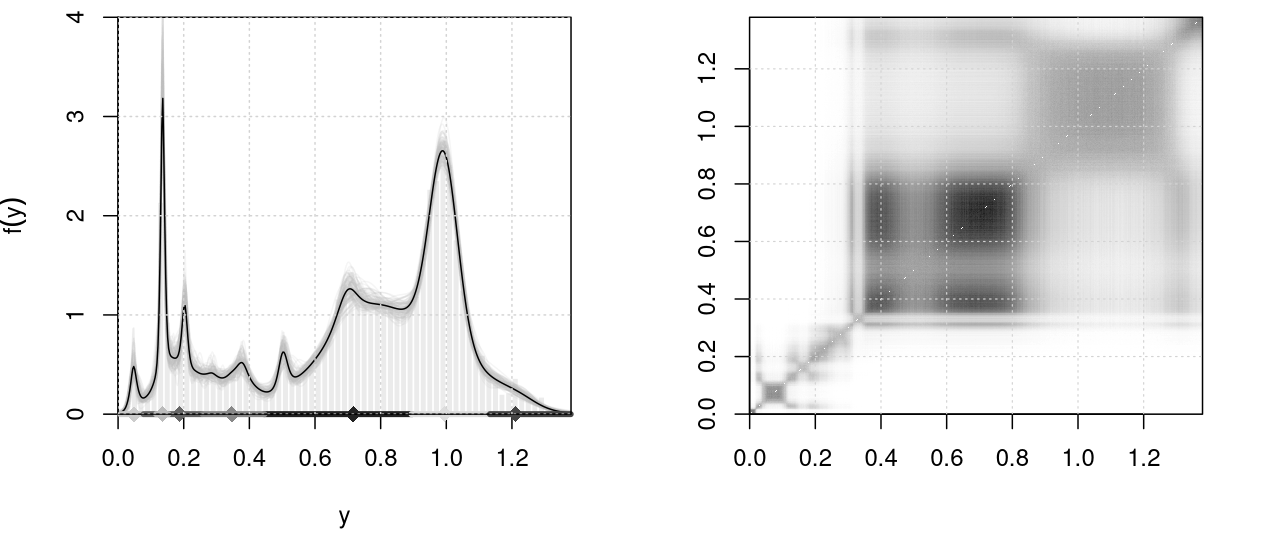}
\caption{\scriptsize{
Left: posterior density of model \eqref{eq:model_random} with prior as in \eqref{eq:prior_PY} and \eqref{eq:model_PY}, with posterior mean (solid line) and posterior draws of the density, together with cluster summaries from the optimal partition. 
Right: Posterior similarity matrix, with rows and columns ordered by galaxy recession velocity; darker blocks indicate higher posterior co-clustering probabilities.}}
\label{fig:galaxy}
\end{figure}

\end{appendix}

\end{document}